\documentclass[a4paper,UKenglish,cleveref, autoref, thm-restate]{lipics-v2021}

\usepackage{ifthen}
\newboolean{shortver}
\setboolean{shortver}{false}

\usepackage{array}
\usepackage{tabularx}
\usepackage{tikz} 
\usetikzlibrary{fit,shapes.geometric} 
\usepackage{placeins} 
\usepackage{todonotes} 
\usetikzlibrary{calc,decorations.pathreplacing,positioning} 

\definecolor{colorOneDark}{HTML}{9B2335}
\definecolor{colorOneLight}{HTML}{F1CDD3}
\definecolor{colorTwoDark}{HTML}{D97706}
\definecolor{colorTwoLight}{HTML}{F9E0B7}
\definecolor{colorThreeDark}{HTML}{087F8C}
\definecolor{colorThreeLight}{HTML}{C4E5E8}
\definecolor{colorFourDark}{HTML}{6A3D9A}
\definecolor{colorFourLight}{HTML}{E4D8F0}

\newcommand{\bw}{\operatorname{bw}}
\newcommand{\Spec}{\mathcal D}
\DeclareMathOperator{\tw}{tw}
\newcommand{\td}{\operatorname{td}}
\newcommand{\cl}{\operatorname{cl}}
\DeclareMathOperator{\vc}{vc}
\newcommand{\CSP}{\textnormal{\textsc{Binary CSP}}}
\newcommand{\Wthree}{\mathsf{W}[3]}
\newcommand{\vars}{\mathcal V}
\providecommand{\cutwidth}{\operatorname{cutwidth}}
\DeclareMathOperator{\cw}{cw}
\DeclareMathOperator{\rw}{rw}
\DeclareMathOperator{\mim}{mim}
\DeclareMathOperator{\mimw}{mimw}
\DeclareMathOperator{\GF}{GF}
\DeclareMathOperator{\con}{con}
\newcommand{\calL}{\mathcal L}
\newcommand{\ACDN}{\mathsf{A\&C\ DN}}

\theoremstyle{definition}

\newcommand{\ESP}{\textnormal{\textsc{ESP}}}
\newcommand{\SubESP}{\textnormal{\textsc{Subdivision-ESP}}}
\newcommand{\ml}{\operatorname{ml}}
\newcommand{\NP}{\mathsf{NP}}
\newcommand{\dist}{\operatorname{dist}}
\newcommand{\ecc}{\operatorname{ecc}}
\newcommand{\diam}{\operatorname{diam}}
\newcommand{\dcog}{\operatorname{dcog}}
\providecommand{\fen}{\operatorname{fen}}
\providecommand{\dtc}{\operatorname{dtc}}
\newcommand{\MSO}{\mathrm{MSO}_1}

\newcommand{\yes}{\textnormal{\textsc{yes}}}

\newcommand{\ds}{\textnormal{\textsc{Dominating Set}}}
\newcommand{\Wone}{\mathsf{W[1]}}
\newcommand{\FPT}{\mathsf{FPT}}
\newcommand{\Wtwo}{\mathsf{W[2]}}
\newcommand{\cc}{\vartheta}
\newcommand{\ta}{\operatorname{tree\text{-}\alpha}}

\title{Eccentricity Shortest Path: the structural parameter landscape} 

\author{Anand Babu}{Department of Computer Science and Engineering, National Institute of Technology Calicut, Kozhikode, India}{anandbabu@nitc.ac.in}{https://orcid.org/0000-0002-9904-0488}{}

\author{Ashwin Jacob}{Department of Computer Science and Engineering, National Institute of Technology Calicut, Kozhikode, India}{ashwinjacob@nitc.ac.in}{https://orcid.org/0000-0003-4864-043X}{Anusandhan National Research Foundation (ANRF) research grant ANRF/ECRG/2025/004197/PMS}

\author{Hitendra Kumar}{Department of Computer Science and Engineering, Indian Institute of Technology Dharwad, Dharwad, India}{ra.hitendra.kumar1729@iitdh.ac.in}{}{}

\author{R. B. Sandeep}{Department of Computer Science and Engineering, Indian Institute of Technology Dharwad, Dharwad, India}{sandeeprb@iitdh.ac.in}{https://orcid.org/0000-0003-4383-1819}{}

\authorrunning{A. Babu and A. Jacob and H. Kumar and R. B. Sandeep} 

\Copyright{Anand Babu and Ashwin Jacob and Hitendra Kumar and R. B. Sandeep} 



\category{} 

\relatedversion{} 

\acknowledgements{}

\nolinenumbers 

\EventEditors{John Q. Open and Joan R. Access}
\EventNoEds{2}
\EventLongTitle{42nd Conference on Very Important Topics (CVIT 2016)}
\EventShortTitle{CVIT 2016}
\EventAcronym{CVIT}
\EventYear{2016}
\EventDate{December 24--27, 2016}
\EventLocation{Little Whinging, United Kingdom}
\EventLogo{}
\SeriesVolume{42}
\ArticleNo{23}

\ifthenelse{\boolean{shortver}}{\hideLIPIcs\acknowledgements{}}{}

\long\def\mespFullAbstract{%
\begin{abstract}
The \emph{Eccentricity Shortest Path} problem asks, given a connected
graph and a nonnegative integer $R$, whether there is a shortest path
such that every vertex is within distance $R$ of the path.
We present a comprehensive study of its structural parameterized
complexity and establish the following results:

\begin{itemize}
\item \textbf{Fixed-parameter tractability.}
We give fixed-parameter algorithms for vertex-deletion distance to
cographs, vertex-deletion distance to interval graphs, and bandwidth.
These algorithms also provide alternative proofs of fixed-parameter
tractability for distance to cluster graphs, established by
Ku\v{c}era and Such\'y (\emph{Algorithmica}, 2023), and for distance
to disjoint unions of paths, established by Bhyravarapu, Jana,
Kanesh, Saurabh, and Verma (IWOCA~2023).
The interval-deletion algorithm additionally resolves a question
posed in the latter paper.
We also record fixed-parameter tractability for treedepth
as a consequence of earlier algorithms.

\item \textbf{XP algorithms.}
We obtain an XP algorithm for vertex-deletion distance to chordal
graphs, extending polynomial-time solvability on chordal graphs
established by Dragan and Leitert
(\emph{Journal of Graph Algorithms and Applications}, 2016).
We also obtain XP algorithms for clique-width and rank-width,
without requiring expressions or decompositions as part of the input.
Consequently, the problem is polynomial-time solvable for every fixed
treewidth, answering a question of Chakraborty, Das, Islam, Mitra, and
Sen (arXiv preprint, 2026) and providing alternative proofs of their
polynomial-time solvability results for $K_{2,3}$-minor-free graphs
and cactus graphs.
We additionally record XP membership for diameter as a consequence
of earlier algorithms.

\item \textbf{Parameterized hardness.}
We prove $\mathsf{W}[3]$-hardness parameterized by cutwidth, and hence
by pathwidth and treewidth, even on graphs of maximum degree at most
five. This strengthens the $\mathsf{W}[1]$-hardness for pathwidth and
treewidth established by Bhyravarapu and Mohanapriya (IPEC~2026).
We also establish $\mathsf{W}[2]$-hardness for vertex clique-cover
number.

\item \textbf{Para-NP-hardness.}
We prove NP-hardness on graphs of tree-independence number at most
two.

\item \textbf{Polynomial kernel.}
For maximum leaf number $\ell$, we obtain an $O(\ell^4)$-bit
compression and a polynomial kernel, strengthening the FPT result
of Ku\v{c}era and Such\'y (\emph{Algorithmica}, 2023).

\item \textbf{Kernel lower bounds.}
Unless $\mathsf{NP}\subseteq\mathrm{coNP}/\mathrm{poly}$,
we exclude polynomial kernels for bandwidth, feedback edge number,
vertex-deletion distance to a clique, and vertex cover number.
The last bound also excludes polynomial kernels for vertex-deletion
distance to proper interval graphs.
\end{itemize}
\end{abstract}
}

\long\def\mespLandscape{%
\begin{figure}[p]
\centering
\resizebox{\textwidth}{!}{%
\begin{tikzpicture}[
  x=0.8cm,
  node distance=8mm and 10mm,
  every node/.style={
    draw,
    rounded corners=2pt,
    align=center,
    font=\scriptsize,
    inner sep=2pt,
    minimum height=6mm,
    fill=white
  },
  color1dark/.style={fill=colorOneDark, text=white},
  color1light/.style={fill=colorOneLight, text=black},
  color2dark/.style={fill=colorTwoDark, text=white},
  color2light/.style={fill=colorTwoLight, text=black},
  color3dark/.style={fill=colorThreeDark, text=white},
  color3light/.style={fill=colorThreeLight, text=black},
  color4dark/.style={fill=colorFourDark, text=white},
  color4light/.style={fill=colorFourLight, text=black},
  edge/.style={draw, line width=0.4pt}
]

\node[color3dark] (dtclique) at (0,0) {Dist to Clique};
\node[star, star points=5, star point ratio=2.25, draw=blue, fill=blue, rounded corners=0pt, inner sep=0pt, minimum width=7pt, minimum height=7pt, font=\normalfont] at ([xshift=-2pt,yshift=-2pt]dtclique.north east) {};
\node[color3dark] (vc)       at (5.2,0) {Vertex Cover};
\node[star, star points=5, star point ratio=2.25, draw=blue, fill=blue, rounded corners=0pt, inner sep=0pt, minimum width=7pt, minimum height=7pt, font=\normalfont] at ([xshift=-2pt,yshift=-2pt]vc.north east) {};
\node[color4dark] (ml)       at (10.6,0) {Max Leaf \#};
\node[star, star points=5, star point ratio=2.25, draw=blue, fill=blue, rounded corners=0pt, inner sep=0pt, minimum width=7pt, minimum height=7pt, font=\normalfont] at ([xshift=-2pt,yshift=-2pt]ml.north east) {};

\node[color4dark] (nd)      at (-2.3,-1.6) {Neighborhood\\diversity};
\node[color2dark] (mcc)      at (-4.8,-1.6) {Clique Cover};
\node[star, star points=5, star point ratio=2.25, draw=blue, fill=blue, rounded corners=0pt, inner sep=0pt, minimum width=7pt, minimum height=7pt, font=\normalfont] at ([xshift=-2pt,yshift=-2pt]mcc.north east) {};
\node[color3light] (dtcoc)    at (0.0,-1.6) {Dist to\\Co-Cluster};
\node[color3light] (dtclu)    at (2.2,-1.6) {Dist to\\Cluster};
\node[color3light] (dtdp)     at (10.6,-1.6) {Dist to\\Paths};
\node[color3dark] (dts)     at (5.2,-1.6) {Dist to\\Split};
\node[color3dark] (td)       at (8.0,-1.6) {Treedepth};
\node[star, star points=5, star point ratio=2.25, draw=blue, fill=blue, rounded corners=0pt, inner sep=0pt, minimum width=7pt, minimum height=7pt, font=\normalfont] at ([xshift=-2pt,yshift=-2pt]td.north east) {};
\node[color3dark] (fes)      at (13.0,-1.6) {Feedback\\Edge Set};
\node[star, star points=5, star point ratio=2.25, draw=blue, fill=blue, rounded corners=0pt, inner sep=0pt, minimum width=7pt, minimum height=7pt, font=\normalfont] at ([xshift=-2pt,yshift=-2pt]fes.north east) {};
\node[color3dark] (bw)       at (15.4,-1.6) {Bandwidth};
\node[star, star points=5, star point ratio=2.25, draw=blue, fill=blue, rounded corners=0pt, inner sep=0pt, minimum width=7pt, minimum height=7pt, font=\normalfont] at ([xshift=-2pt,yshift=-2pt]bw.north east) {};

\node[color3dark] (mw)      at (-0.15,-3.2) {Modular\\width};
\node[color2dark] (cutw)       at (15.4,-3.2) {Cutwidth};
\node[star, star points=5, star point ratio=2.25, draw=blue, fill=blue, rounded corners=0pt, inner sep=0pt, minimum width=7pt, minimum height=7pt, font=\normalfont] at ([xshift=-2pt,yshift=-2pt]cutw.north east) {};
\node[color2dark] (dtsp)     at (8.0,-3.2) {Dist to\\Spider};
\node[color3dark] (dtcog)    at (1.7,-3.2) {Dist to\\Cograph};
\node[star, star points=5, star point ratio=2.25, draw=blue, fill=blue, rounded corners=0pt, inner sep=0pt, minimum width=7pt, minimum height=7pt, font=\normalfont] at ([xshift=-2pt,yshift=-2pt]dtcog.north east) {};
\node[color2light] (mis)      at (-4.4,-3.2) {Ind Set};
\node[color3dark] (dtint)    at (3.9,-3.2) {Dist to\\Interval};
\node[star, star points=5, star point ratio=2.25, draw=blue, fill=blue, rounded corners=0pt, inner sep=0pt, minimum width=7pt, minimum height=7pt, font=\normalfont] at ([xshift=-2pt,yshift=-2pt]dtint.north east) {};
\node[color1dark] (genus)    at (13.0,-3.2) {Genus};

\node[color1dark] (md)       at (15.4,-4.8) {Max\\Degree};
\node[color2light] (fvs)      at (8.0,-4.8) {Feedback\\Vertex Set};
\node[color2light] (mindom)   at (-4.4,-4.8) {Dom Set};

\node[color2light] (pw)       at (13.0,-6.4) {Pathwidth};
\node[color1light] (hidx)     at (15.4,-6.4) {$h$-index};
\node[color1dark] (dtbip)    at (1.8,-7.1) {Dist to\\Bipartite};
\node[color2light] (dtout)    at (8,-6.4) {Dist to\\Outerplanar};
\node[color2dark] (dtcho)    at (-2.2,-6.4) {Dist to\\Chordal};
\node[star, star points=5, star point ratio=2.25, draw=blue, fill=blue, rounded corners=0pt, inner sep=0pt, minimum width=7pt, minimum height=7pt, font=\normalfont] at ([xshift=-2pt,yshift=-2pt]dtcho.north east) {};
\node[color2dark] (diamc)    at (5.2,-6.4) {Diameter};
\node[star, star points=5, star point ratio=2.25, draw=blue, fill=blue, rounded corners=0pt, inner sep=0pt, minimum width=7pt, minimum height=7pt, font=\normalfont] at ([xshift=-2pt,yshift=-2pt]diamc.north east) {};

\node[color1dark] (dtp)       at (13.0,-8.0) {Dist to\\Planar};
\node[color2light] (tw)       at (8.0,-8.0) {Treewidth};

\node[color1light] (dtperf)   at (0.0,-9.6) {Dist to\\Perfect};
\node[color2dark] (cwd)      at (5.2,-9.6) {Cliquewidth};
\node[star, star points=5, star point ratio=2.25, draw=blue, fill=blue, rounded corners=0pt, inner sep=0pt, minimum width=7pt, minimum height=7pt, font=\normalfont] at ([xshift=-2pt,yshift=-2pt]cwd.north east) {};
\node[color1dark] (tind)       at (8.0,-9.6) {Tree ind no.};
\node[star, star points=5, star point ratio=2.25, draw=blue, fill=blue, rounded corners=0pt, inner sep=0pt, minimum width=7pt, minimum height=7pt, font=\normalfont] at ([xshift=-2pt,yshift=-2pt]tind.north east) {};

\draw[edge] (dtclique) -- (mcc);
\draw[edge] (dtclique) -- (nd);
\draw[edge] (dtclique) -- (dtcoc);
\draw[edge] (dtclique) -- (dtclu);
\draw[edge] (dtclique) -- (dts);

\draw[edge] (vc) -- (dtclu);
\draw[edge] (vc) -- (nd);
\draw[edge] (vc) -- (dtdp);
\draw[edge] (vc) -- (dts);

\draw[edge] (vc) -- (dtcoc);
\draw[edge] (vc) -- (td);

\draw[edge] (ml) -- (fes);
\draw[edge] (ml) -- (bw);

\draw[edge] (ml) -- (dtdp);

\draw[edge] (mcc) -- (mis);

\draw[edge] (dtcoc) -- (dtcog);
\draw[edge] (dtcoc) -- (mindom);

\draw[edge] (dtclu) -- (dtcog);
\draw[edge] (dtclu) -- (dtint);

\draw[edge] (dtdp) -- (dtsp);
\draw[edge] (dtdp) -- (hidx);

\draw[edge] (dts) -- (dtcho);
\draw[edge] (dts) -- (diamc);

\draw[edge] (dtdp) -- (dtint);

\draw[edge] (td) -- (pw);
\draw[edge] (td) -- (diamc);

\draw[edge] (dtsp) -- (pw);
\draw[edge] (dtsp) -- (fvs);

\draw[edge] (fes) -- (genus);

\draw[edge] (fes) -- (fvs);

\draw[edge] (bw) -- (cutw);

\draw[edge] (nd) -- (mw);

\draw[edge] (cutw) -- (pw);
\draw[edge] (cutw) -- (md);

\draw[edge] (mis) -- (mindom);

\draw[edge] (dtcog) -- (dtperf);
\draw[edge] (dtcog) -- (cwd);
\draw[edge] (mw) -- (cwd);
\draw[edge] (mw) -- (mindom);

\draw[edge] (dtcog) -- (diamc);

\draw[edge] (dtint) -- (dtcho);

\draw[edge] (fvs) -- (dtout);
\draw[edge] (fvs) -- (dtcho);

\draw[edge] (fvs) -- (dtbip);

\draw[edge] (pw) -- (tw);

\draw[edge] (md) -- (hidx);

\draw[edge] (mindom) -- (diamc);

\draw[edge] (dtcho) -- (dtperf);
\draw[edge] (dtcho) -- (tind);

\draw[edge] (dtout) -- (tw);

\draw[edge] (dtout) -- (dtp);

\draw[edge] (dtbip) -- (dtperf);

\draw[edge] (tw) -- (cwd);
\draw[edge] (tw) -- (tind);

\node[
  draw,
  rounded corners=2pt,
  fill=white,
  anchor=north west,
  align=left,
  font=\scriptsize,
  inner sep=3pt,
  minimum height=0mm
] at (-0.8,-10.8) {%
  \begin{tabular}{@{}ll@{}}
    \multicolumn{2}{@{}l@{}}{\textbf{ESP complexity}}\\[1pt]
    \colorbox{colorOneDark}{\rule{0pt}{1.2ex}\hspace{1.2em}}\;\colorbox{colorOneLight}{\rule{0pt}{1.2ex}\hspace{1.2em}} & Not in XP\\
    \colorbox{colorTwoDark}{\rule{0pt}{1.2ex}\hspace{1.2em}}\;\colorbox{colorTwoLight}{\rule{0pt}{1.2ex}\hspace{1.2em}} & XP but not FPT\\
    \colorbox{colorThreeDark}{\rule{0pt}{1.2ex}\hspace{1.2em}}\;\colorbox{colorThreeLight}{\rule{0pt}{1.2ex}\hspace{1.2em}} & FPT but no Poly Kernel\\[1pt]
    \colorbox{colorFourDark}{\rule{0pt}{1.2ex}\hspace{1.2em}}\;%
    \colorbox{colorFourLight}{\rule{0pt}{1.2ex}\hspace{1.2em}}
& Polynomial kernel\\[1pt]
    \multicolumn{2}{@{}l@{}}{Dark: main result; light: implied result}
  \end{tabular}%
};

\end{tikzpicture}
}
\caption{The structural parameter landscape of \ESP{} on connected
graphs, inspired by the graph parameter hierarchy of Sorge and
Weller~\cite{SorgeWeller2019}, and combining results from earlier
work~\cite{DraganLeitert2017,KuceraSuchy2023,BhyravarapuEtAl2023,BhyravarapuEtAl2023Full,Tomer2023,BhyravarapuMohanapriya2026}
with the results proved here. A line from an upper parameter $p$ to a
lower parameter $q$ means that $q$ is bounded by a function of $p$ on
connected graphs.
For the nontrivial parameter relations, see
\cite{FellowsEtAl2009,SorgeWeller2019} for the bounds from maximum leaf
number to bandwidth, feedback edge number, and distance to disjoint paths;
\cite[Lemma~2]{Lampis2012} for the bound on neighborhood diversity by
vertex cover number and, by complementation, by distance to a clique;
\cite[Theorem~3]{GajarskyLampisOrdyniak2013} for the bound on modular-width
by neighborhood diversity;
\cite{Gurski2017,CourcelleOlariu2000} for the bounds on clique-width in
terms of modular-width, cograph-deletion distance, and treewidth;
\cite{Kinnersley1992} for the bound on pathwidth by cutwidth;
\cite{NesetrilOssonaDeMendez2006,SorgeWeller2019} for the relations
between vertex cover number, treedepth, pathwidth, and diameter; and
\cite{DallardMilanicStorgel2024} for the bounds on tree-independence
number by treewidth and chordal-deletion distance.
The bound on genus by feedback edge number is given
in~\cite[Lemma~4.19]{SorgeWeller2019}.
The remaining relations follow from graph-class inclusions and elementary
structural arguments.
The four colors distinguish the complexity classes listed in the legend;
negative statements use the standard assumptions
$\mathsf{P}\neq\NP$, $\FPT\neq\Wone$, $\FPT\neq\Wtwo$,
$\FPT\neq\Wthree$, and
$\NP\nsubseteq\mathrm{coNP}/\mathrm{poly}$, as appropriate.
Dark and light shades distinguish directly established results from
consequences of parameter bounds. Blue stars mark parameters explicitly
treated in this paper, including consequences of earlier algorithms.
\emph{Dist to} denotes vertex-deletion distance, and \emph{Ind Set} and
\emph{Dom Set} denote the independence and domination numbers.}
\label{fig:parameter-landscape}
\end{figure}
}

\long\def\mespFullIntroduction{%
\section{Introduction}
\label{sec:introduction}

A shortest path provides an efficient route between its endpoints, but it
need not pass close to the rest of the graph. When a path is intended to
serve as an accessible route through a network, these two requirements
must be considered together: travel along the route should be shortest,
and the greatest distance needed to reach it should be small. The
\emph{Eccentricity Shortest Path} problem (\ESP) captures this
combination. For a connected graph $G$ and a path $P$, the eccentricity of
$P$ is
\[
 \ecc_G(P)=\max_{v\in V(G)}\min_{u\in V(P)}\dist_G(v,u).
\]
The optimization problem asks for a shortest path of minimum eccentricity.
Its decision version asks whether a given connected graph $G$ has a shortest
path $P$ satisfying
$\ecc_G(P)\leq R$, where $R$ is part of the input. The endpoints are not
prescribed, and a shortest path is always understood to be shortest
between its own endpoints.

Dragan and Leitert introduced \ESP{} and established its connection with
approximating minimum-distortion embeddings of graph metrics into the
line~\cite{DraganLeitert2017}. A path of low eccentricity provides a
one-dimensional structure near every vertex, while its shortest-path
property preserves distances along that structure. The problem therefore
combines a metric covering requirement with an exact constraint on the
chosen route. These requirements interact: extending a path can improve
coverage but destroy shortestness, whereas replacing it by a shorter route
can leave previously covered vertices too far away. This interaction is
central to both the algorithms and the lower bounds in this paper.

\subsection{Previous work and structural parameterization}

The general problem is computationally difficult. Dragan and
Leitert~\cite{DraganLeitert2017} proved NP-hardness even on planar
bipartite graphs of maximum degree three and $\Wtwo$-hardness when the
parameter is the requested eccentricity. They also gave an
$O(n^{2r+2}m)$-time algorithm for finding an optimum path, where $r$ is
the optimum eccentricity, together with constant-factor approximation
algorithms. Here $n$ and $m$ denote the numbers of vertices and edges.
Birmel\'{e}, de Montgolfier, and
Planche~\cite{BirmeleDeMontgolfierPlanche2016} gave a linear-time
$3$-approximation for \ESP{} and established tight bounds relating its
optimum eccentricity to laminarity, where the path is required to be
diametral.
In a separate paper published in 2016, Dragan and Leitert obtained a
linear-time algorithm on distance-hereditary graphs and polynomial-time
algorithms on chordal graphs, dually chordal graphs, and graphs of bounded tree-length or
bounded hyperbolicity~\cite{DraganLeitert2016Structured}. In particular,
the chordal-graph result implies polynomial-time solvability on interval
and split graphs.

Generalizations of \ESP{} have also been studied.
Birmel\'{e} et al.~\cite{BirmeleEtAl2020} introduced
\emph{hub-laminar decompositions}, which use several locally shortest
paths with small collective eccentricity that meet only near their
endpoints.
Chakraborty and Vax\`{e}s~\cite{ChakrabortyVaxes2025} considered
\emph{$k$-Geodesic Center}, which asks for $k$ shortest paths minimizing
the maximum distance from a vertex to their union; its case $k=1$ is
precisely \ESP. They gave a polynomial-time additive
$O(\delta)$-approximation for this generalization on $\delta$-hyperbolic
graphs and proved NP-hardness on partial grids for every fixed
$k\geq 1$, including \ESP.

Chakraborty et al.~\cite{ChakrabortyEtAl2026} recently obtained an
$O(n^4)$-time algorithm for \ESP{} on $K_{2,3}$-minor-free graphs,
a class containing all outerplanar graphs, and an $O(n^3)$-time
algorithm on cactus graphs. Their approach fixes the endpoints of the
solution path and uses the structure of the shortest paths between them
to construct a weighted auxiliary directed acyclic graph. Minimizing
the largest edge weight along a path in this auxiliary graph yields a
shortest path of minimum eccentricity for the chosen endpoints.
They leave open polynomial-time solvability on graphs of every fixed
treewidth, including the case of treewidth at most two.

These results motivate studying how structural restrictions on the input
control the complexity of \ESP. We study FPT and XP algorithms, as well
as polynomial kernelization, with the eccentricity threshold $R$
unrestricted, using structural parameters of $G$ alone.

All distances to graph classes below mean vertex-deletion distances.
Ku\v{c}era and Such\'y~\cite{KuceraSuchy2023} initiated a systematic
study of structural parameterizations of \ESP. They established
fixed-parameter tractability for modular-width, distance to cluster
graphs, and maximum leaf number, as well as for treewidth combined with
the requested eccentricity. They also asked whether treewidth alone
suffices for fixed-parameter tractability. Bhyravarapu
et al.~\cite{BhyravarapuEtAl2023} subsequently obtained FPT algorithms
for distance to a disjoint union of paths and for feedback vertex number
combined with the eccentricity threshold, and an FPT approximation
scheme parameterized by feedback vertex number. Their full
version~\cite{BhyravarapuEtAl2023Full} also proves fixed-parameter
tractability for distance to split graphs and $\Wtwo$-hardness for
distance to chordal graphs, and asks about distance to interval graphs.

More recently, Bhyravarapu and
Mohanapriya~\cite{BhyravarapuMohanapriya2026} proved that \ESP{} is
$\Wone$-hard parameterized by the vertex-deletion distance $s$ to a
disjoint union of spiders, where a spider is a tree with at most one
vertex of degree greater than two. They also exclude an $n^{o(s)}$-time
algorithm under the Exponential Time Hypothesis (ETH). Their result
implies $\Wone$-hardness for each of feedback vertex number,
pathwidth, and treewidth individually. This answers the earlier
tractability questions for feedback vertex number and treewidth
negatively unless $\FPT=\Wone$. On the positive side, they give a
$2^{O(t\log t)}n^{O(1)}$-time algorithm parameterized by feedback edge
number $t$.

For kernelization, Tomer's report~\cite{Tomer2023} reports a linear
kernel parameterized by neighborhood diversity.

\subsection{Our results}

We provide a comprehensive picture of the structural parameterized
complexity of \ESP{}, establishing FPT and XP algorithms,
parameterized hardness results, and both positive and negative
kernelization results. Table~\ref{tab:intro-results} summarizes
the principal statements proved in this paper.
Figure~\ref{fig:parameter-landscape} places these results among related
structural parameters and previously known bounds. 

\begin{table}[t]
\centering
\small
\renewcommand{\arraystretch}{1.16}
\begin{tabularx}{\textwidth}{@{}>{\raggedright\arraybackslash}p{0.30\textwidth}
 >{\raggedright\arraybackslash}X@{}}
\hline
\textbf{Parameter} & \textbf{Result established in this paper} \\
\hline
Distance to cographs $k$
 & FPT: $f(k)n^{O(1)}$
   (Section~\ref{sec:fpt-algorithms}) \\
Distance to interval graphs $k$
 & FPT: $13^k n^{O(1)}$
   (Section~\ref{sec:fpt-algorithms}) \\
Bandwidth $b$
 & FPT: $2^{O(b^2\log(b+1))}n^{O(1)}$; no polynomial kernel
   (Sections~\ref{sec:fpt-algorithms}
   and~\ref{sec:kernelization}) \\
Treedepth $d$
 & FPT: $F(d)n^3$
   (Section~\ref{sec:fpt-algorithms}) \\
Diameter $d$
 & XP: $O(n^{2d+2}m)$
   (Section~\ref{sec:xp-algorithms}) \\
Distance to chordal graphs $k$
 & XP: $n^{O(k)}$
   (Section~\ref{sec:xp-algorithms}) \\
Clique-width $k$; rank-width $q$
 & XP without a supplied expression or decomposition
   (Section~\ref{sec:xp-algorithms}) \\
Width $w$ of a supplied mim-width layout
 & $n^{O(\max\{1,w\})}$ time
   (Section~\ref{sec:xp-algorithms}) \\
Vertex clique-cover number
 & $\Wtwo$-hard
   (Section~\ref{sec:w-hardness}) \\
Cutwidth
 & $\Wthree$-hard even at maximum degree five
   (Section~\ref{sec:w-hardness}) \\
Tree-independence number
 & NP-hard already at value at most two
   (Section~\ref{sec:para-np-hardness}) \\
Maximum leaf number $\ell$
 & Polynomial kernel; $O(\ell^4)$-bit compression
   (Section~\ref{sec:kernelization}) \\
Distance to clique; feedback edge number; vertex cover number
 & No polynomial kernel
   (Section~\ref{sec:kernelization}) \\
\hline
\end{tabularx}
\caption{Results proved in this paper. Every kernel lower bound assumes
$\NP\nsubseteq\mathrm{coNP}/\mathrm{poly}$. Treedepth FPT and diameter
XP are consequences of earlier algorithms~\cite{KuceraSuchy2023,DraganLeitert2017}.}
\label{tab:intro-results}
\end{table}

\mespLandscape

\paragraph*{Algorithms.}
Our $13^k n^{O(1)}$-time algorithm for interval deletion distance
resolves the question posed by Bhyravarapu
et al.~\cite{BhyravarapuEtAl2023Full}. We also obtain fixed-parameter
tractability for distance to cographs and bandwidth. For bandwidth, the
running time is $2^{O(b^2\log(b+1))}n^{O(1)}$, and a bandwidth ordering
need not be supplied. On the XP side, we give algorithms for distance to
chordal graphs, clique-width, and rank-width. In particular, our chordal-deletion algorithm complements the known
$\Wtwo$-hardness result, highlighting a distinction in the
parameterized complexity of \ESP{} between the vertex-deletion
distances to interval graphs and to chordal graphs.
The clique-width and rank-width results follow from an
algorithm on a supplied mim-width layout whose polynomial exponent
depends on the layout width but not on $R$.
Since clique-width is bounded by a function of treewidth, our clique-width
XP algorithm also yields polynomial-time solvability for every fixed
treewidth, including series-parallel graphs. This answers the question
raised by Chakraborty et al.~\cite{ChakrabortyEtAl2026}; their specialized
algorithms give explicit quartic and cubic running-time bounds on
$K_{2,3}$-minor-free graphs and cactus graphs, respectively, both of which have bounded treewidth. For completeness, we also record treedepth FPT and diameter XP as
consequences of the earlier
treewidth--eccentricity and optimum-eccentricity algorithms,
respectively.

\paragraph*{Hardness.}
We prove $\Wthree$-hardness parameterized by cutwidth, even for graphs
of maximum degree five and even when a suitable ordering is supplied.
Since pathwidth, and hence treewidth, is at most cutwidth, this also
gives $\Wthree$-hardness for both parameters, strengthening the earlier
$\Wone$-hardness of Bhyravarapu and
Mohanapriya~\cite{BhyravarapuMohanapriya2026}.
Together with the bandwidth algorithm, this separates two familiar
ordering parameters for \ESP: bounding edge lengths in an ordering
supports an FPT algorithm, whereas bounding the number of edges crossing
each cut does not, unless $\FPT=\Wthree$. We also prove
$\Wtwo$-hardness for the vertex clique-cover number, including when a
clique cover is supplied. Finally, \ESP{} is NP-hard on graphs of
tree-independence number at most two. This remains true with an explicit
tree decomposition whose every bag is partitioned into two cliques.
Hence bounding the independence number within bags does not suffice for
an XP algorithm, unless $\mathsf{P}=\NP$.

\paragraph*{Polynomial preprocessing.}
For maximum leaf number $\ell$, we obtain an $O(\ell^4)$-bit compression
into a fixed language in NP and derive an ordinary unweighted \ESP{}
kernel of polynomial encoding length. This strengthens the previously
known FPT result for this parameter.
In contrast, \ESP{}
has no polynomial kernel parameterized by bandwidth, feedback edge
number, distance to clique, or vertex cover number, unless
$\NP\subseteq\mathrm{coNP}/\mathrm{poly}$. In particular, the
feedback-edge lower bound complements the FPT algorithm of Bhyravarapu
and Mohanapriya~\cite{BhyravarapuMohanapriya2026}, showing that this
tractability does not extend to polynomial kernelization under the
stated assumption. The vertex-cover lower bound also yields lower bounds
for deletion distance to cluster, split,
proper interval, and interval graphs. Thus the interval-deletion
and bandwidth FPT algorithms cannot be accompanied by polynomial kernels
under this assumption. The maximum-leaf and feedback-edge results
further show that different restrictions on sparse graphs can have
opposite consequences for polynomial preprocessing.

\FloatBarrier

\subsection{Overview of the techniques}

\paragraph*{Summarizing distances to a geodesic.}
The algorithms for interval and chordal deletion distance use the
structure of shortest paths to keep small descriptions of their
interaction with a modulator. In an interval graph, distances to a
geodesic can be related to distances to its endpoints. We use this
property to obtain a parameter-bounded family of candidate distance
profiles on the modulator. For each candidate profile, a sweep through an interval representation stores only the last three selected vertices in the processed prefix under right-endpoint order and the first three selected vertices in the unprocessed suffix under left-endpoint order, together with degree counters and witnesses for the guessed distances. Summing over these choices gives the $13^k$ factor. For a
chordal modulator, the relevant decomposition has bags consisting of a
clique together with the modulator. A geodesic uses at most two vertices
of the clique, and its distance profile on a bag has a small set of
representatives. Storing profiles for both sides of each separator gives
the XP algorithm despite the bags having unbounded size.

\paragraph*{Few distance values instead of a bounded radius.}
The bandwidth algorithm addresses a different obstacle: the optimum
eccentricity can be arbitrarily large even on cycles, which have bounded
bandwidth. Consequently, simply applying an algorithm parameterized by
treewidth and $R$ cannot give the desired result. We instead bound the
number of attainable distances from a vertex to a shortest path with
fixed endpoints. A packing argument for separated vertices in
bounded-bandwidth trees yields a bound involving a parameter-dependent
power of $\log n$. A dynamic program over a tree decomposition uses
these distance values rather than all integers up to $R$. Absorbing the
remaining powers of $\log n$ into the parameter-dependent factor gives
the claimed FPT running time.

\paragraph*{Logical descriptions of shortest paths.}
For distance to cographs, a small modulator bounds both clique-width and
diameter. A cotree refined by neighborhoods in the modulator provides a
bounded-label expression, and the diameter bound permits an
$\mathrm{MSO}_1$ description of the required path with
parameter-bounded formula length. The clique-width XP algorithm uses a
different description: a connected vertex set containing $s$ and $t$
and having exactly $\dist_G(s,t)+1$ vertices must induce an
$s$--$t$ geodesic. Combining this certificate with the mim-width
model-checking theorem of Bergougnoux, Dreier, and
Jaffke~\cite{BergougnouxDreierJaffke2023} handles the unbounded
eccentricity threshold. Rank decompositions supply the required layouts,
so no clique-width expression needs to be given.

\paragraph*{Encoding choices by distances.}
Our hardness constructions force a shortest path to make a sequence of
choices and use distances to test vertices to enforce consistency.
For clique-cover number, these choices encode a dominating set.
For cutwidth, we reduce from the vertex-cover parameterization of
\CSP{} studied by Bodlaender, Groenland, and
Pilipczuk~\cite{BGP23}. Variable gadgets encode values through
complementary distances to two ports, and test gadgets recognize
forbidden pairs. The construction simultaneously bounds the cutwidth
and maximum degree. For tree-independence number, clause vertices form
one clique that can be added to a decomposition of the remaining
construction, yielding bags covered by two cliques.

\paragraph*{Compressing structure and preserving information.}
For maximum leaf number, suppressing degree-two chains leaves a small
weighted core, but its binary edge lengths may still depend on $n$.
We combine a direct algorithm with this representation to obtain a
polynomial-size compression. A crucial additional step is to verify
the encoded path and its eccentricity without expanding the chains;
this places the compressed language in NP and permits a reduction back
to ordinary \ESP.
The negative results use OR-cross-compositions and polynomial parameter
transformations. Their constructions encode many source instances or
choices while keeping the relevant parameter small: through bounded
modules for bandwidth, pendant structures around a fixed cyclic core
for feedback edge number, and large clique or independent parts for
the deletion-distance and vertex-cover lower bounds.

\paragraph*{Organization.}
Section~\ref{sec:preliminaries} fixes the notation and graph parameters.
Sections~\ref{sec:fpt-algorithms} and~\ref{sec:xp-algorithms} present the
FPT and XP algorithms. Sections~\ref{sec:w-hardness} and
\ref{sec:para-np-hardness} establish the parameterized and
para-NP-hardness results. Section~\ref{sec:kernelization} presents the
polynomial kernel and kernel lower bounds.

}

\long\def\mespFullProofs{%
\section{Preliminaries}
\label{sec:preliminaries}

We use standard graph-theoretic and parameterized-complexity terminology; see,
for example, Cygan et al.~\cite{cygan-et-al-2015}. All graphs are finite,
undirected, simple, and unweighted, unless stated otherwise. For a graph $G$,
we write $V(G)$ and $E(G)$ for its vertex and edge sets, and put
$n=|V(G)|$ and $m=|E(G)|$. The open and closed neighborhoods of a vertex
$v$ are denoted by $N_G(v)$ and $N_G[v]$, respectively. For
$X\subseteq V(G)$, the graph induced by $X$ is denoted by $G[X]$, while
$G-X=G[V(G)\setminus X]$. We omit the subscript $G$ whenever the underlying
graph is clear.
Subdividing an edge $uv$ means replacing it by the path $u,w,v$, where
$w$ is a new vertex. More generally, a \emph{subdivision} replaces each
edge by a path of positive length with the same endpoints; the replacement
paths have pairwise disjoint sets of new internal vertices.

For vertices $u,v\in V(G)$, let $\dist_G(u,v)$ be their distance in $G$.
For a nonempty set $S\subseteq V(G)$, define
\[
 \dist_G(v,S)=\min_{x\in S}\dist_G(v,x),\qquad
 \ecc_G(S)=\max_{v\in V(G)}\dist_G(v,S).
\]
A path from $u$ to $v$ is a \emph{shortest path}, or a \emph{geodesic}, if
its length is $\dist_G(u,v)$. We allow one-vertex paths. Every geodesic is
induced. For a path $P$, we abbreviate $\dist_G(v,V(P))$ and
$\ecc_G(V(P))$ to $\dist_G(v,P)$ and $\ecc_G(P)$, respectively.
A subgraph $H$ of $G$ is \emph{isometric} if
$\dist_H(x,y)=\dist_G(x,y)$ for all $x,y\in V(H)$.

\paragraph*{Eccentricity Shortest Path (\ESP).}
The input consists of a connected graph $G$ and a nonnegative integer $R$.
The question is whether $G$ has a shortest path $P$ with
$\ecc_G(P)\leq R$. The optimization version asks for a shortest path of
minimum eccentricity. We call such a path a \emph{solution path} and write
$(G,R)\in\ESP$ when the instance is a yes-instance. The problem was
introduced by Dragan and Leitert~\cite{DraganLeitert2017}.

\paragraph*{Parameterized complexity.}
For a graph parameter $k$, an FPT
algorithm has running time $f(k)n^{O(1)}$, while an XP algorithm may have
running time $f(k)n^{g(k)}$, for computable functions $f$ and $g$.
Thus FPT algorithms keep the polynomial exponent independent of the
parameter. A complementary question asks whether polynomial-time
preprocessing can reduce an instance to an equivalent instance whose
encoding length is polynomial in $k$; such a reduction is a
\emph{polynomial kernel}. We refer the reader to the textbook~\cite{cygan-et-al-2015}
for standard terminology in parameterized complexity.

\subsection{Graph parameters}
\label{sec:graph-parameters}

We use the following graph parameters.

\medskip
\noindent\textbf{Basic numerical parameters.}
The degree of $v$ is $\deg_G(v)=|N_G(v)|$, and the maximum degree is
$\Delta(G)=\max_{v\in V(G)}\deg_G(v)$. A set of pairwise nonadjacent vertices
is \emph{independent}; the independence number $\alpha(G)$ is the maximum
size of such a set. For a connected graph $G$, its diameter and optimum
ESP eccentricity~\cite{DraganLeitert2017} are, respectively,
\[
 \diam(G)=\max_{u,v\in V(G)}\dist_G(u,v),\qquad
 r(G)=\min\{\ecc_G(P):P\text{ is a shortest path of }G\}.
\]
When referring to the maximum diameter of a disconnected graph's components,
we take the maximum of $\diam(C)$ over its connected components $C$.

\medskip
\noindent\textbf{Vertex-deletion distances.}
For a graph class $\mathcal H$, the \emph{deletion distance to $\mathcal H$}
is
\[
 \min\{|X|:X\subseteq V(G)\text{ and }G-X\in\mathcal H\}.
\]
A set $X$ with $G-X\in\mathcal H$ is an
\emph{$\mathcal H$-modulator}~\cite{cygan-et-al-2015}.
All distances to graph classes in this paper refer to vertex deletion.
We consider the following classes:
\begin{itemize}
 \item A \emph{clique} is a graph whose vertices are pairwise adjacent.
 \item A \emph{cograph} is a graph with no induced $P_4$~\cite{Corneiletal1985}.
 \item A \emph{chordal graph} has no induced cycle of length at least
 four~\cite{Dirac1961}.
 \item An \emph{interval graph} is the intersection graph of intervals on a
 line~\cite{LekkerkerkerBoland1962}. It is a \emph{proper interval graph}
 if it has such a representation in which no interval properly contains
 another~\cite{Roberts1969}.
 \item A \emph{cluster graph} is a disjoint union of
 cliques~\cite{ShamirSharanTsur2004}.
 \item A \emph{split graph} has a vertex partition into a clique and an
 independent set~\cite{FoldesHammer1977}.
\end{itemize}
We denote the distances to cographs and to cliques by $\dcog(G)$ and
$\dtc(G)$, respectively. Thus a cograph modulator $X$ satisfies that $G-X$
is a cograph, and a clique modulator $X$ satisfies that $G-X$ is a clique.
The distances to chordal, interval, proper interval, cluster, and split
graphs are defined by the same minimum above.

Equivalently, nonempty cographs are obtained from one-vertex graphs by repeated
disjoint unions and complete joins. A \emph{cotree} records this
construction: its leaves are the graph's vertices, and its internal nodes
are labelled $\oplus$ or $\otimes$ for disjoint union or complete
join~\cite{Corneiletal1985}.

\medskip
\noindent\textbf{Treewidth, pathwidth, and tree-independence number.}
A \emph{tree decomposition} of $G$ consists of a tree $T$ and bags
$B_u\subseteq V(G)$ for $u\in V(T)$ such that every vertex belongs to a bag,
every edge has both endpoints in a common bag, and, for every $v\in V(G)$,
the nodes whose bags contain $v$ induce a connected subtree of $T$.
Its width is $\max_{u\in V(T)}|B_u|-1$, and its independence number is
$\max_{u\in V(T)}\alpha(G[B_u])$. The \emph{treewidth}
$\tw(G)$~\cite{RobertsonSeymour1986} and
\emph{tree-independence number} $\ta(G)$~\cite{DallardMilanicStorgel2024}
are the respective minima over all tree decompositions of $G$. A
\emph{path decomposition} is a tree decomposition whose underlying tree is a path; the \emph{pathwidth} is the
minimum width of a path decomposition~\cite{RobertsonSeymour1983}.

\medskip
\noindent\textbf{Treedepth.}
The \emph{closure} $\cl(F)$ of a rooted forest $F$ joins each vertex to all
its proper ancestors. The \emph{treedepth} $\td(G)$ is the minimum height of
a rooted forest $F$ on $V(G)$ with $G\subseteq\cl(F)$, where height counts
vertices on a longest root-to-leaf path~\cite{NesetrilOssonaDeMendez2006}.

\medskip
\noindent\textbf{Bandwidth and cutwidth.}
For an ordering $\pi=(v_1,\ldots,v_n)$ of $V(G)$, put
\[
 \bw(G,\pi)=\max_{v_iv_j\in E(G)}|i-j|,\qquad
 \cutwidth(G,\pi)=\max_{0\leq h\leq n}
 |\{v_iv_j\in E(G):i\leq h<j\}|.
\]
For an edgeless graph, put $\bw(G,\pi)=0$. The \emph{bandwidth} $\bw(G)$
and \emph{cutwidth} $\cutwidth(G)$ are the minima of these quantities over
all vertex orderings $\pi$~\cite{DiazPetitSerna2002}. We also view an
ordering as a bijection $\pi:V(G)\to\{1,\ldots,n\}$, with $\pi(v_i)=i$.

\medskip
\noindent\textbf{Clique-width.}
A \emph{$k$-expression} constructs a labelled graph using labels from
$\{1,\ldots,k\}$ and four operations: creating a labelled vertex, taking
disjoint union, adding all edges between two distinct labels, and relabelling
all vertices of one label. The \emph{clique-width} $\cw(G)$ is the minimum
$k$ for which $G$ has a $k$-expression~\cite{CourcelleOlariu2000}.

\medskip
\noindent\textbf{Rank-width and mim-width.}
For $A\subseteq V(G)$, put $\overline A=V(G)\setminus A$. Let $\mim_G(A)$
be the maximum size of an induced matching in the bipartite cut graph
$G[A,\overline A]$. Edges with both endpoints on one side of the cut are
not part of this bipartite graph. Let $\rho_G(A)$ be the rank over $\GF(2)$
of its bipartite adjacency matrix.

A \emph{rooted layout} $\calL=(T,\delta)$ consists of a rooted full binary
tree $T$ and a bijection $\delta$ from $V(G)$ to its leaves. If $V_x$ is the
set of vertices corresponding to leaves below a node $x$, its width is
\[
 \mimw_G(\calL)=\max_{x\in V(T)}\mim_G(V_x).
\]
The \emph{mim-width} of $G$ is the minimum of $\mimw_G(\calL)$ over all
rooted layouts $\calL$~\cite{Vatshelle2012,BergougnouxDreierJaffke2023}.
A \emph{rank-decomposition} is a subcubic tree with leaves bijectively
labelled by $V(G)$; its width is the maximum cut-rank over
the cuts obtained by deleting its edges. The minimum such width is the
\emph{rank-width} $\rw(G)$~\cite{OumSeymour2006}. For a one-vertex graph,
both rank-width and mim-width are zero.

\medskip
\noindent\textbf{Vertex cover number and feedback edge number.}
The \emph{vertex cover number} $\vc(G)$ is the minimum size of a
set $X\subseteq V(G)$ such that $G-X$ is independent~\cite{Karp1972}.
The \emph{feedback edge number} $\fen(G)$ is the minimum number of edges
whose deletion makes $G$ a forest~\cite{BalabanGanianRocton2024}.

\medskip
\noindent\textbf{Clique cover number.} The \emph{vertex clique-cover number}
$\cc(G)$ is the minimum number of cliques whose vertex sets partition
$V(G)$~\cite{Karp1972}. This is different from the edge clique-cover number.

\medskip
\noindent\textbf{Maximum leaf number.} For connected $G$, the
\emph{maximum leaf number} $\ml(G)$ is the maximum number of leaves
in a spanning tree of $G$~\cite{FellowsEtAl2009}. We write
$\ell=\ml(G)$ when using this parameter.

\section{FPT algorithms}
\label{sec:fpt-algorithms}

In this section, we present FPT algorithms for \ESP{} parameterized
by each of the following: distance to cographs, distance to interval
graphs, bandwidth, and treedepth. 

\subsection{Distance to cographs}
\label{sec:mesp-cographs}

In this subsection, we show that \ESP{} is fixed-parameter tractable
parameterized by distance to cographs. Recall that $\dcog(G)$ is the
minimum size of a cograph modulator; see
Subsection~\ref{sec:graph-parameters}.

\begin{restatable}{theorem}{thmcograph}
\label{thm:maincograph}
There is a computable function $f$ such that \ESP{} can be solved in time
$f(k)|V(G)|^{O(1)}$ on every instance $(G,R)$ with $k=\dcog(G)$.
Consequently, \ESP{} is fixed-parameter tractable parameterized by distance
to cographs. The optimization version is FPT for the same parameter.
\end{restatable}

Our proof has two ingredients. First,
from a cograph modulator we construct a clique-width expression whose width is
bounded by a function of the modulator size. Second, we express the existence
of a suitable shortest path by a formula whose length depends only on the
parameter. Then we use the following meta-theorem of Courcelle, Makowsky, and Rotics~\cite{CMR2000}.

\begin{restatable}{proposition}{cmrmetatheorem}[Courcelle–Makowsky–Rotics~\cite{CMR2000}]
\label{thm:cmr}
There is a computable function $F$ with the following property. Given a
$q$-expression $\tau$ constructing a graph $G$ and an $\MSO$ sentence
$\varphi$, whether $G\models\varphi$ can be decided in
$F(q,|\varphi|)|\tau|$ time. The same statement holds for a formula with free
individual variables when values for those variables are supplied.
\end{restatable}

\subsubsection{A bounded-width clique-width expression}

We first construct a clique-width expression for a graph that is close to a
cograph. The following standard cotree fact will be useful.

\begin{restatable}{proposition}{binarycotree}
\label{cor:binary-cotree}
Let $G$ be a cograph with $n$ vertices and $m$ edges. Then $G$ admits a
rooted binary cotree in which the leaves correspond bijectively to $V(G)$ and
every internal node is labelled either $\oplus$ or $\otimes$, representing
disjoint union and complete join, respectively. Moreover, such a binary
cotree can be constructed in $O(n+m)$ time.
\end{restatable}

\begin{proof}
Corneil, Perl, and Stewart~\cite{Corneiletal1985} showed that a cotree of a
cograph can be constructed in $O(n+m)$ time. A cotree node may have more
than two children. As observed by Bodlaender and M\"ohring
~\cite{BodlaenderMohring1993}, an internal node with children
$u_1,\ldots,u_r$, where $r>2$, can be replaced by a rooted binary tree whose
leaves are $u_1,\ldots,u_r$ and whose internal nodes all receive the label of
the original node. This does not change the represented graph, since both
disjoint union and complete join are associative. The number of new nodes is
linear in the size of the original cotree, so the transformation takes linear
time; see also Lyons~\cite[Section~2.1]{LYONS2011}.
\end{proof}

The next lemma is based on the standard clique-width construction for adding
vertices; compare Gurski~\cite{Gurski2017}. We include the direct construction
because we need an explicit bound on the size of the expression.

\begin{restatable}{lemma}{cographmodulatorexpression}
\label{lem:expression}
Let $X$ be a cograph modulator of $G$ with $|X|=t$, and let
$n=|V(G)|$. One can construct a $(2^{t+1}+t)$-expression $\tau$ for $G$
with $|\tau|=O(4^t n)$ in $f_0(t)n^{O(1)}$ time, for a computable
function $f_0$.
\end{restatable}

\begin{proof}
Put $H=G-X$. By Proposition~\ref{cor:binary-cotree}, a rooted binary cotree
gives a $2$-expression for $H$ of size $O(|V(H)|)$, ending with all vertices
labelled $1$. Disjoint-union nodes need no additional operation. At a join
node, relabel the output of the second child to $2$ before taking the disjoint
union, insert all edges between labels $1$ and $2$, and then relabel $2$ to
$1$.

For $v\in V(H)$, let $\sigma(v)=N_G(v)\cap X$. Refine each label
$i\in\{1,2\}$ to labels $(i,S)$ for all $S\subseteq X$, and create $v$
with the refined label determined by $\sigma(v)$. Replace an edge-insertion
operation between labels $1$ and $2$ by insertions between every pair
$(1,S),(2,T)$. Similarly, replace a relabelling $i\to j$ by
$(i,S)\to(j,S)$ for every $S\subseteq X$. A straightforward induction over
the expression shows that this constructs exactly $H$ and ends with each
vertex $v$ labelled $(1,\sigma(v))$.

Now create every $x\in X$ with its own private label. Insert the edges of
$G[X]$ individually, and join the label of $x$ to $(1,S)$ precisely when
$x\in S$. These operations add exactly the missing edges of $G$. The
construction uses $2\cdot 2^t+t$ labels and
\[
 O\bigl(4^t|V(H)|+t2^t+t^2+t\bigr)=O(4^t n)
\]
operations. If $H$ is empty, only the private labels and the edges of $G[X]$
are needed.
\end{proof}

\subsubsection{A logical description of ESP}

Set $\delta_{-1}(x,y):=\bot$ and $\delta_0(x,y):=(x=y)$. For $r\geq1$,
define recursively, always using fresh bound variables,
\[
 \delta_r(x,y):=\exists z\,\bigl((x=z\lor E(x,z))\land
 \delta_{r-1}(z,y)\bigr).
\]
After expansion, this formula describes a walk of at most $r$ edges; equality
allows a step to remain at its current vertex. Hence $\delta_r(u,v)$ holds
exactly when $\dist_G(u,v)\leq r$.

For nonnegative integers $D$ and $R$, define
\refstepcounter{equation}\label{eq:formula}
\[
 \begin{aligned}
 \Phi_{D,R}:={}&
 \bigvee_{\ell=0}^{D}\exists p_0\cdots\exists p_\ell
 \Biggl[
 \left(\bigwedge_{i=0}^{\ell-1}E(p_i,p_{i+1})\right)
 \land\neg\delta_{\ell-1}(p_0,p_\ell) \land
 \left(\forall v\,\bigvee_{i=0}^{\ell}\delta_R(v,p_i)\right)
 \Biggr].
 \end{aligned}
 \qquad\text{(\theequation)}
\]
The empty conjunction for $\ell=0$ is true.

The following diameter bound is a direct consequence of a result of
Kobayashi and Otachi~\cite{KobayashiOtachi2022}.

\begin{restatable}{proposition}{cographmodulatordiameter}
\label{lem:diameter-cograph-modulator}
Let $G$ be connected and let $X\subseteq V(G)$ be a cograph modulator. If
$|X|=k$, then $\operatorname{diam}(G)\leq 4k+3$.
\end{restatable}

\begin{proof}
Kobayashi and Otachi~\cite{KobayashiOtachi2022} observed that, if every
component of every graph in a class $\mathcal H$ has diameter at most $d$,
then every component of a graph at vertex-deletion distance at most $k$ from
$\mathcal H$ has diameter less than $(k+1)(d+2)$. Every connected cograph has
diameter at most $2$. Therefore
$\operatorname{diam}(G)<4(k+1)$, and the integral bound
$\operatorname{diam}(G)\leq4k+3$ follows.
\end{proof}

\begin{restatable}{lemma}{mespformulacorrectness}
\label{lem:formula-correct}
Let $X$ be a cograph modulator of a graph $G$, put $t=|X|$, and let
$R\geq0$. Then $G\models\Phi_{4t+3,R}$ if and only if $(G,R)$ is a
yes-instance of \ESP{}.
\end{restatable}

\begin{proof}
Suppose first that $G\models\Phi_{4t+3,R}$. Some disjunct indexed by an
integer $\ell\leq4t+3$ has a satisfying assignment. Its first conjunct
describes a walk of length $\ell$, while the second conjunct rules out a
shorter walk between its endpoints. The walk is therefore a geodesic path.
The last conjunct says that every vertex of $G$ is at distance at most $R$
from this path.

Conversely, let $P=p_0p_1\cdots p_\ell$ be a shortest path with
$\ecc_G(P)\leq R$. By Proposition~\ref{lem:diameter-cograph-modulator},
$\ell\leq\operatorname{diam}(G)\leq4t+3$. Assigning the vertices of $P$ to
$p_0,\ldots,p_\ell$ satisfies $\Phi_{4t+3,R}$.
\end{proof}

The number of symbols in $\Phi_{4t+3,R}$ is bounded by a computable function
of $t+R$.

\subsubsection{The FPT algorithm}

We can now combine the expression and the logical description.

\begin{proof}[Proof of Theorem~\ref{thm:maincograph}]
Let $(G,R)$ be an instance, put $k=\dcog(G)$, and let $n=|V(G)|$. Greedily
take a maximal collection of pairwise vertex-disjoint induced
$P_4$'s, and let $X$ be the union of their vertex sets. Maximality implies
that $G-X$ is a cograph. Every cograph modulator must meet every selected
$P_4$. Thus, writing $t=|X|$, we have $t\leq4k$. The collection can be found
in $O(n^4)$ time by scanning all four-vertex sets and selecting an induced
$P_4$ whenever its vertices avoid the previously selected paths.

Proposition~\ref{lem:diameter-cograph-modulator}, applied to $X$, gives
$\operatorname{diam}(G)\leq4t+3$. If $R\geq4t+3$, the answer is immediately
\yes: a path consisting of any one vertex has eccentricity at most
$\operatorname{diam}(G)\leq R$. We may therefore assume that $R\leq4t+3$.

By Lemma~\ref{lem:expression}, we can construct a $q$-expression $\tau$ for
$G$ in $f_0(t)n^{O(1)}$ time such that
$
 q\leq2^{t+1}+t$ and
$
 |\tau|=O(4^t n).
$
The number of symbols in $\Phi_{4t+3,R}$ is bounded by a computable function
$g(t)$, since $R\leq4t+3$. The total running
time is at most
\[
 O(n^4)+f_0(t)n^{O(1)}+
 F\bigl(2^{t+1}+t,g(t)\bigr)O(4^t n)
 =f(t)n^{O(1)}\leq f'(k)n^{O(1)}
\]
for computable functions $f$ and $f'$, where the last equality uses
$t\leq4k$. Hence the decision problem is FPT.

For the optimization version, test $R=0,1,\ldots,4t+3$ and return the first
accepted value. The value $R=4t+3$ is always accepted, so this requires at
most $4t+4$ calls. A witnessing path can be recovered by standard
self-reduction. First determine an accepted disjunct $\ell$, and then fix
$p_0,p_1,\ldots,p_\ell$ one at a time, using the model checker for the formula
with free individual variables and trying all vertices as the next value.
This uses at most $n(\ell+1)=O(tn)$ further model-checking calls. Since $t \leq 4k$, it preserves fixed-parameter tractability.
\end{proof}

\subsection{Distance to interval graphs}
\label{sec:mesp-interval}

Recall that the distance to interval graphs is the minimum number of
vertices whose deletion leaves an interval graph, as defined in
Subsection~\ref{sec:graph-parameters}.

The \ESP{} problem is polynomial-time solvable on interval graphs, as a
special case of the polynomial-time algorithm for chordal graphs
given by Dragan and Leitert~\cite{DraganLeitert2016Structured}. Bhyravarapu
et al.~\cite{BhyravarapuEtAl2023Full} identified the
parameterization by interval vertex deletion number as an open question.
We resolve this question by presenting an FPT algorithm for \ESP{}
parameterized by the deletion distance to interval graphs.

Let $G$ be a graph and let $S\subseteq V(G)$, $|S|\leq k$, be a set such
that $H=G-S$ is an interval graph. We assume that such a modulator
$S$ is given; if not we can find such a set using the
FPT algorithm for Interval Vertex Deletion~\cite{cao2015interval}, which runs in time $10^k n^{O(1)}$. We also assume that an interval representation of $H$ is available; if not it can be computed in polynomial  time~\cite{BoothLueker1976}.

\subsubsection{Enumerating distance profiles}

The algorithm first check if the path is a singleton vertex in polynomial time. If not, it enumerates distinct endpoints $a,b$ and sets $D=d(a,b)$. For this guess let
\[
 \tau(v)=
 \begin{cases}
  1,&v\in\{a,b\},\\
  2,&v\notin\{a,b\}.
 \end{cases}
\]

$\tau(v)$ denote the degree of the vertices in the solution path. We begin with the following lemma.
\begin{lemma}\label{lem:interval-endpoint}
Let $Q=q_0\cdots q_\ell$ be a geodesic of an interval graph $H$. If $r=\dist_H(u,V(Q))$ is finite and $r\geq2$, then $
 \min\{\dist_H(u,q_0),\dist_H(u,q_\ell)\}\leq r+1$.
\end{lemma}
\begin{proof}
If $Q$ is a singleton, the claim is immediate. Otherwise, choose a vertex $q_i$ in $Q$ that is nearest to $u$ and a shortest $u$--$q_i$ path $W$. Suppose both endpoint distances exceed $r+1$. Then $i\geq2$ and $\ell-i\geq2$, by following $W$ and the appropriate subpath of $Q$.
The vertices $u,q_0,q_\ell$ are pairwise nonadjacent. The path $Q$ avoids
$N_H[u]$ because $r\geq2$.

Also, $W$ followed by $Q[q_i,q_0]$ avoids $N_H[q_\ell]$. An internal vertex
$w$ of $W$ equal or adjacent to $q_\ell$ would give a route from $u$ to
$q_\ell$ of length at most $r$, a contradiction. The subpath $Q[q_i,q_0]$
avoids $N_H[q_\ell]$ because $Q$ is induced and $\ell-i\geq2$. The analogous
path to $q_\ell$ avoids $N_H[q_0]$. These concatenations are paths: an
internal vertex of $W$ on $Q$ would be closer to $u$ than $r$.
Thus $u,q_0,q_\ell$ form an asteroidal triple, which is impossible for interval graphs~\cite{LekkerkerkerBoland1962}.
\end{proof}

For an $a$--$b$ geodesic $P$, define for all $s \in S$, $\rho_s=d(s,V(P))$,
and distance profile $\rho=(\rho_s)_{s\in S}$. Let $T=\{s\in S:\rho_s\leq 2\}$. We use the vertices of $T$ as anchors. For each $s\in S$, define
\begin{equation}\label{eq:anchor}
 \beta_s=\min\left\{d(s,a),d(s,b),
        \min_{t\in T}\bigl(d(s,t)+\rho_t\bigr)\right\}.
\end{equation}
We guess the set $T$ and the values $\rho_t$ for $t\in T$. Once these values are fixed, $\beta_s$ gives an estimate of
$\rho_s$ using only the two endpoints of $P$ and the nearby anchors.
The next lemma shows that this estimate is exact on $T$ and is within
additive two for every vertex of $S$.

\begin{lemma}\label{lem:interval-anchor}
For every $s\in S$, $\rho_s\leq\beta_s\leq\rho_s+2$. Moreover, if $s\in T$, then $\beta_s=\rho_s$.
\end{lemma}
\begin{proof}
Every term in the definition of $\beta_s$ corresponds to a walk from
$s$ to $P$. Hence $\beta_s\geq \rho_s$. If $s\in T$, then $\rho_s\leq 2$, and the term corresponding to $t=s$ in the definition of $\beta_s$ gives
$\beta_s\leq d(s,s)+\rho_s=\rho_s$. Thus $\beta_s=\rho_s$.

It remains to prove $\beta_s\leq \rho_s+2$ for an arbitrary $s\in S$.
Let $W$ be a shortest $s$--$P$ path, and let $p\in V(P)$ be its
endpoint. If $p\in S$, then $p\in T$ and $\rho_p=0$. The corresponding
anchor term gives $\beta_s\leq d(s,p)+\rho_p=d(s,p)=\rho_s$.

Suppose now that $p\notin S$. Let $z$ be the last vertex of $S$ on $W$,
and let $u$ be the vertex immediately following $z$ on $W$. Let $Q$ be
the component of $P-S$ containing $p$. The suffix of $W$ from $u$ to
$p$ lies entirely in $H$. It is a shortest $u$--$Q$ path in $H$;
otherwise, replacing this suffix by a shorter $u$--$Q$ path would give a
shorter $s$--$P$ path in $G$. Put $r=d_H(u,V(Q))$. Then $\rho_s=d(s,z)+1+r$. We first note that $\rho_z=1+r$. Indeed, the path from $z$ through $u$ to $Q$ has length $1+r$, so $\rho_z\leq 1+r$. On the other hand, $\rho_s\leq d(s,z)+\rho_z \leq d(s,z)+1+r =\rho_s$, and therefore equality holds throughout. In particular, $\rho_z=1+r$. 

If $r\leq 1$, then $\rho_z\leq2$, so $z\in T$. The anchor term for $z$
therefore gives $\beta_s\leq d(s,z)+\rho_z=\rho_s$. It remains to consider $r\geq2$. The path $Q$ is a geodesic of $H$: otherwise a shorter path between its endpoints in $H$ would shorten $P$ in $G$. By Lemma~\ref{lem:interval-endpoint}, $Q$ has an endpoint $e$ such that
$ d_H(u,e)\leq r+1$. Consequently,
\[
d(s,e)\leq d(s,z)+1+d_H(u,e)
        \leq d(s,z)+r+2
        =\rho_s+1.
\]

If $e\in\{a,b\}$, the corresponding endpoint term in the definition of
$\beta_s$ gives $\beta_s\leq d(s,e)\leq \rho_s+1$. Otherwise, $e$ has a neighbor $m\in S\cap V(P)$ along $P$. Since $m\in V(P)$, we have $\rho_m=0$, and hence $m\in T$. The anchor term for $m$ gives
\[
\beta_s\leq d(s,m)+\rho_m
        \leq d(s,e)+1
        \leq \rho_s+2.
\]
Thus $\beta_s\leq\rho_s+2$ in all cases.
\end{proof}

For fixed endpoints $a,b$, we enumerate six possible labels for every
$s\in S$:
\[
 \mathsf{near}_0,\ \mathsf{near}_1,\ \mathsf{near}_2,
 \quad
 \mathsf{far}_0,\ \mathsf{far}_1,\ \mathsf{far}_2.
\]
A label $\mathsf{near}_j$ declares that $s\in T$ and sets
$\rho_s=j$. After all near labels have been chosen, we compute
$\beta_s$ from these guessed values using \eqref{eq:anchor}. A label
$\mathsf{far}_j$ then sets $\rho_s=\beta_s-j$. We discard the candidate if a vertex labeled $\mathsf{far}_j$ receives $\rho_s\leq 2$, since such a vertex should belong to $T$. Finally, we define $M=\{s\in S:\rho_s=0\}$. 
Thus $M$ is determined by the distance assignments and need not be
guessed separately.

\begin{corollary}\label{cor:interval-enumeration}
For fixed endpoints, the enumeration produces at most $6^k$ candidates
$(M,\rho)$, and contains the true candidate for every $a$--$b$ geodesic.
For a fixed prescribed zero set $M$ of size $m$, at most
$5^{k-m}$ candidates have zero set exactly $M$.
\end{corollary}
\begin{proof}
Consider a distance profile $\rho$. If $\rho_s\leq2$, assign
$s$ the label $\mathsf{near}_{\rho_s}$. Otherwise, by
Lemma~\ref{lem:interval-anchor}, $0\leq\beta_s-\rho_s\leq2$, so assigning $s$ the label $\mathsf{far}_{\beta_s-\rho_s}$ produces the correct value of $\rho_s$. Hence the actual distance profile is included in the enumeration.

There are six labels for each of the $k$ vertices of $S$, giving at
most $6^k$ candidates. If the zero set $M$ is prescribed, every
$s\in M$ is forced to receive the label $\mathsf{near}_0$. Each
$s\in S\setminus M$ has the remaining five possible labels, giving at
most $5^{k-m}$ candidates.
\end{proof}

\subsubsection{Coverage and interval summaries}

\begin{lemma}\label{lem:interval-mod-separator}
Let $Y=M\cup X\ne\varnothing$, where $M\subseteq S$ and $X\subseteq V(H)$. Suppose that for every $s \in S$, $\rho_s=d(s,Y)$. Then, for every $v\in V(H)$,
\begin{equation}\label{eq:mod-separator}
 d(v,Y)=
 \min\left\{
   \dist_H(v,X),
   \min_{s\in S}\bigl(d(v,s)+\rho_s\bigr)
 \right\}.
\end{equation}
\end{lemma}
\begin{proof}
Every finite term on the right-hand side gives a walk from $v$ to $Y$.
Indeed, a path in $H$ from $v$ to $X$ gives the first term, while for
$s\in S$, a shortest $v$--$s$ path followed by a shortest $s$--$Y$ path
gives a walk of length $d(v,s)+\rho_s$. Hence
\[
d(v,Y)\leq
\min\left\{
   \dist_H(v,X),
   \min_{s\in S}\bigl(d(v,s)+\rho_s\bigr)
\right\}.
\]

For the reverse inequality, let $P$ be a shortest $v$--$Y$ path. If
$P$ avoids $S$, then $P$ lies entirely in $H$ and its endpoint belongs
to $X$. Thus $d(v,Y)=|P|\geq \dist_H(v,X)$. Otherwise, let $s$ be the first vertex of $P$ that belongs to $S$. The
prefix of $P$ from $v$ to $s$ has length at least $d(v,s)$, while the
remaining suffix has length at least $d(s,Y)=\rho_s$. Therefore 
\[
d(v,Y)=|P|\geq d(v,s)+\rho_s
\geq \min_{x\in S}\bigl(d(v,x)+\rho_x\bigr).
\]
Combining the two cases proves \eqref{eq:mod-separator}.
\end{proof}

Fix a candidate $(M,\rho)$ from the enumeration. We first perform the
following consistency checks, rejecting the candidate if any of them
fails:
\begin{equation}\label{eq:interval-prechecks}
 \begin{gathered}
 0\leq\rho_s\leq R\quad \forall s\in S,\qquad
 \{a,b\}\cap S\subseteq M,\qquad |M|\leq D+1,\\
 d(s,m)\geq\rho_s\quad \forall s\in S,\ \forall m\in M,\qquad
 |N_G(m)\cap M|\leq\tau(m)\quad \forall m\in M.
 \end{gathered}
\end{equation}

The first two conditions ensure that the guessed distances are feasible
and that every modulator vertex used by the endpoints belongs to the
guessed path. The remaining conditions ensure that the vertices of $M$
do not already violate the guessed distances or their required induced
degrees.

Next, define the set of vertices that are eligible to be selected by
\[
E_\rho=
\{v\in V(H): d(s,v)\geq\rho_s\text{ for every }s\in S\},
\]
and the set of vertices whose coverage cannot be satisfied using the
modulator alone by
\[
U_\rho=
\{v\in V(H): d(v,s)+\rho_s>R\text{ for every }s\in S\}.
\]
Thus, only vertices of $E_\rho$ may be selected. Moreover, when $\rho$ is the actual distance profile of the selected path, i.e., $\rho_s=d(s,Y)$ for every $s\in S$, Lemma~\ref{lem:interval-mod-separator} shows that the remaining coverage requirement is precisely $\dist_H(v,X)\leq R$ for every $v\in U_\rho$.

Fix an interval representation $I_v=[L(v),R(v)]$ of $H$. We may assume
that all interval endpoints are distinct: at each tied
coordinate, place all left endpoints before all right endpoints, and
break ties arbitrarily within each group.
Let $\prec_R$ and $\prec_L$ denote the orders of the vertices by
increasing right and left endpoints, respectively.

\begin{lemma}\label{lem:interval-monotone}
For distinct vertices $v,x,y$ of $H$:
\begin{enumerate}
 \item if $R(x)<R(y)<R(v)$ and $\dist_H(x,v)<\infty$, then
       $\dist_H(y,v)\leq\dist_H(x,v)$;
 \item if $R(v)<R(x)$, $R(v)<R(y)$, $L(x)<L(y)$, and
       $\dist_H(v,y)<\infty$, then
       $\dist_H(v,x)\leq\dist_H(v,y)$.
\end{enumerate}
\end{lemma}

\begin{proof}
For the first statement, let $x=p_0,p_1,\ldots,p_d=v$ be a shortest $x$--$v$ path. The union of the intervals $I_{p_0},\ldots,I_{p_d}$ is connected. Since $R(x)<R(y)<R(v)$, this union contains the point $R(y)$. On the other hand, $R(x)<R(y)$, so $I_x$ does not contain $R(y)$. Hence some vertex $p_j$, $j\geq1$, has an interval containing $R(y)$. Therefore $p_j$ is equal or adjacent to $y$. Starting at $y$, move to $p_j$ if necessary and then follow the
suffix $p_j,p_{j+1},\ldots,p_d=v$. This gives a $y$--$v$ walk of
length at most $1+(d-j)\leq d$, since $j\geq1$. Thus $\dist_H(y,v)\leq d=\dist_H(x,v)$.

For the second statement, first suppose that $L(x)<L(v)$. Since
$R(v)<R(x)$, the interval $I_x$ contains $R(v)$, and hence $xv\in E(H)$.
Therefore $\dist_H(v,x)=1\leq\dist_H(v,y)$. It remains to consider the case $L(v)<L(x)<L(y)$. Let $v=p_0,p_1,\ldots,p_d=y$ be a shortest $v$--$y$ path. The union of the intervals along this path is connected and contains $L(x)$. Since $L(y)>L(x)$, the interval $I_y$ does not contain $L(x)$. Thus some $p_j$, with $j\leq d-1$, has an interval intersecting $I_x$. Consequently $p_jx\in E(H)$. Following the path from $v$ to $p_j$ and then taking the edge $p_jx$ gives a $v$--$x$ walk of length at most $j+1\leq d$.
Hence $\dist_H(v,x)\leq d=\dist_H(v,y)$. 
\end{proof}

Write $h_1\prec_R\cdots\prec_R h_N$ where $N=|V(H)|$, and let $
L_i=\{h_1,\ldots,h_i\}$ and $\overline L_i=V(H)\setminus L_i$. For a vertex set $Z$, let $\mathrm R_3(Z)$ be the last three vertices
of $Z$ in the order $\prec_R$, and let $\mathrm L_3(Z)$ be the first
three vertices of $Z$ in the order $\prec_L$. If $|Z|\leq3$, we retain
all vertices of $Z$.

\begin{corollary}\label{cor:interval-summaries}
Let $v=h_{i+1}$ and $t\geq1$.

For any $Z^-\subseteq L_i$, let $q^-$ be the number of vertices of
$Z^-$ at distance at most $t$ from $v$. Then $\mathrm R_3(Z^-)$
contains exactly $\min\{3,q^-\}$ such vertices.

Similarly, for any $Z^+\subseteq\overline L_{i+1}$, if $q^+$ is the
number of vertices of $Z^+$ at distance at most $t$ from $v$, then
$\mathrm L_3(Z^+)$ contains exactly $\min\{3,q^+\}$ such vertices.
\end{corollary}

\begin{proof}
Consider first $Z^-$. By Lemma~\ref{lem:interval-monotone}, among the vertices preceding $v$ in the order $\prec_R$, moving to the right can only
decrease the distance to $v$. Hence the vertices of $Z^-$ that are
within distance $t$ of $v$ form a suffix of $Z^-$ in the order
$\prec_R$. If there are at most three such vertices, all of them are among the last three vertices of $Z^-$. If there are more than three, the last three vertices of $Z^-$ are all within distance $t$ of $v$. Thus $\mathrm R_3(Z^-)$ contains exactly $\min\{3,q^-\}$ such vertices.

The argument for $Z^+$ is analogous. By Lemma~\ref{lem:interval-monotone}, the vertices of $Z^+$ within distance $t$ of $v$ form a prefix in the order
$\prec_L$. Hence the first three vertices, $\mathrm L_3(Z^+)$, preserve
their number up to three.
\end{proof}

\subsubsection{The dynamic programming algorithm}

Put $S_0=S\setminus M$. For each $m\in M$, let $b_m=|N_G(m)\cap M|$ and $t_m=\tau(m)-b_m$. The preliminary checks ensure that $0\leq t_m\leq2$. Since every vertex of $M$ belongs to the selected path and has $\rho_m=0$, each vertex of $M$ already witnesses its own distance. Thus we only need
distance-equality witnesses for vertices in $S_0$, and degree counters
only for vertices in $M$.

After processing $L_i$, a state is a tuple $(A,B,c,q,W)$, where
\[
\begin{gathered}
A\subseteq L_i,\qquad
B\subseteq\overline L_i,\qquad
|A|,|B|\leq3,\\
c(m)\in\{0,\ldots,t_m\}\quad \forall m\in M,\qquad
0\leq q\leq D+1,\qquad
W\subseteq S_0.
\end{gathered}
\]
The state components have the following intended meanings for a solution
vertex set $X\subseteq V(H)$:
\begin{align*}
A&=\mathrm R_3(X\cap L_i),\\
B&=\mathrm L_3(X\cap\overline L_i),\\
c(m)&=|N_G(m)\cap X\cap L_i|
&& \forall m\in M,\\
q&=|M|+|X\cap L_i|,\\
W&=\{s\in S_0:
      d(s,M\cup(X\cap L_i))=\rho_s\}.
\end{align*}

Thus, $A$ records the three rightmost selected vertices already
processed, while $B$ records the three leftmost selected vertices that
are promised for the unprocessed part. The counter $c(m)$ records how
many selected neighbors of $m$ have been seen so far, $q$ records the
number of selected vertices counted so far including $M$, and $W$
records the vertices of $S_0$ whose required distance $\rho_s$ already
has a witness among $M$ and the processed vertices.

A table entry is retained whenever it can be obtained by the transitions below. The correctness of these intended meanings, including the promised information in $B$, is proved for states occurring in an accepting computation.

\paragraph*{Initialization.} At $i=0$, no vertices of $H$ have been processed. Thus set $A=\varnothing, c\equiv0, q=|M|$, and let $W=\{s\in S_0:d(s,M)=\rho_s\}$. We enumerate every set $B\subseteq V(H)$ with $|B|\leq3$ and retain the corresponding initial states $(\varnothing,B,0,|M|,W)$. 

\paragraph*{Processing a vertex.} Consider a state that is retained at layer $i<N$, and let $v=h_{i+1}$. Choose $\varepsilon\in\{0,1\}$, where $\varepsilon=1$ means that $v$ is
selected. Let
\[
Z=
\begin{cases}
\{v\},&\varepsilon=1,\\
\varnothing,&\varepsilon=0.
\end{cases}
\]
If $v\in\{a,b\}$, we force $\varepsilon=1$. If $v$ is selected, we also
require $v\in E_\rho$.

We next enumerate a new future promise $B'\subseteq\overline L_{i+1}, |B'|\leq3$, and require
\begin{equation}\label{eq:interval-suffix}
 B=\mathrm L_3(B'\cup Z).
\end{equation}
The selected vertices in $\overline L_i$ are precisely $
Z\cup(X\cap\overline L_{i+1})$, where $Z=\{v\}$ if $v$ is selected and $Z=\varnothing$ otherwise. Since $B'$ stores the first three vertices of
$X\cap\overline L_{i+1}$ in the order $\prec_L$, equation
\eqref{eq:interval-suffix} gives the required value of $B$. 

If $v$ is selected, we must verify its final degree in $G[M\cup X]$.
By Corollary~\ref{cor:interval-summaries} with $t=1$, the last three
selected vertices among the processed vertices and the first three
promised selected vertices among the unprocessed vertices suffice to determine whether $v$ has one, two, or at least three selected neighbors in $X$.
Therefore, require
\begin{equation}\label{eq:interval-degree}
 \min\left\{
  3,\,
  |N_G(v)\cap M|
  +|N_H(v)\cap A|
  +|N_H(v)\cap B'|
 \right\}
 =\tau(v).
\end{equation}

If $v$ is skipped and $v\in U_\rho$, then $v$ cannot be covered through
the modulator and must therefore have a selected vertex of $X$ within
distance $R$ in $H$. If $R=0$, this is impossible, so we reject. For
$R\geq1$, Corollary~\ref{cor:interval-summaries} with $t=R$ implies that
it suffices to check
\begin{equation}\label{eq:interval-cover}
 \exists\,w\in A\cup B'
 \quad\text{such that}\quad
 \dist_H(v,w)\leq R.
\end{equation}
A selected vertex is covered by itself, while a skipped vertex outside
$U_\rho$ is already covered through the modulator. After these tests, set
$A'=\mathrm R_3(A\cup Z)$. 

The remaining state components are updated in the following way: $
c'(m)=c(m)+\varepsilon\,\mathbf 1_{mv\in E(G)}$ for all $m\in M$, $
q'=q+\varepsilon$, and $W' = W\cup \{s\in S_0:  \varepsilon=1\text{ and }d(s,v)=\rho_s\}$. Reject the transition if $c'(m)>t_m$ for some $m\in M$ or if $q'>D+1$. Otherwise retain $(A',B',c',q',W')$ at layer $i+1$.

\paragraph*{Acceptance.} After all vertices of $H$ have been processed, accept a state exactly when
\begin{equation}\label{eq:interval-accept}
 B=\varnothing,\qquad
 q=D+1,\qquad
 W=S_0,\qquad
 c(m)=t_m\quad \forall m\in M.
\end{equation}

The condition $B=\varnothing$ ensures that no selected vertex remains
in the future promise. The equality $q=D+1$ gives the required number
of vertices in the geodesic, while $W=S_0$ certifies that every
$s\in S_0$ attains its guessed distance $\rho_s$. Finally,
$c(m)=t_m$ for every $m\in M$ ensures that all selected vertices of
the modulator have their required degree.

\subsubsection{Correctness and running time}

The following lemma proves the correctness of the algorithm. 
\begin{lemma}\label{lem:interval-correct}
For every candidate from Corollary~\ref{cor:interval-enumeration}, the dynamic programming algorithm accept if and only if there is an $a$--$b$ geodesic $P$ of eccentricity at most $R$ such that
\[
V(P)\cap S=M
\qquad\text{and}\qquad
d(s,V(P))=\rho_s
\quad\text{for every }s\in S.
\]
\end{lemma}

\begin{proof}
For completeness, let $Y=V(P)=M\cup X$ be such a path. At each layer
$i$, take $A=\mathrm R_3(X\cap L_i)$,$B=\mathrm L_3(X\cap\overline L_i)$,
let $c$ and $q$ have their values according to the definitions above,
and let $W=\{s\in S_0: d(s,M\cup(X\cap L_i))=\rho_s\}$. 

The set $A$ has the required meaning by the definition of the update:
after processing $v$, the selected processed vertices are
$X\cap L_{i+1}=(X\cap L_i)\cup Z$, and hence $A'=\mathrm R_3(A\cup Z)
   =\mathrm R_3(X\cap L_{i+1})$. Thus the meaning of $A$ is preserved.

For $B$, the new promise $B'$ is chosen so that $B'=\mathrm L_3(X\cap\overline L_{i+1})$. Since $X\cap\overline L_i
=(X\cap\overline L_{i+1})\cup Z$, equation~\eqref{eq:interval-suffix} gives
$ B=\mathrm L_3(B'\cup Z)  =\mathrm L_3(X\cap\overline L_i)$. Thus the meaning of $B$ is preserved.

All vertices of $X$ belong to $E_\rho$: for every $s\in S$, $
d(s,x)\geq d(s,Y)=\rho_s$. The degree of every selected vertex is correctly tested by Corollary~\ref{cor:interval-summaries} with $t=1$. Since the required degree is at most two, retaining three vertices is sufficient to determine whether the degree is one or two. The contribution of vertices of $M$ is known exactly. Moreover, the counters $c(m)$ never exceed their final values $t_m$. The candidate also satisfies the preliminary checks:
$\rho_s\leq R$ because $\ecc_G(P)\leq R$,
$\{a,b\}\cap S\subseteq M$, $|M|\leq D+1$, and the remaining
inequalities follow from $Y=V(P)$ and the fact that $P$ is induced.

Consider now a skipped vertex $v\in U_\rho$. Since $\rho$ is the actual distance profile, Lemma~\ref{lem:interval-mod-separator} gives $\dist_H(v,X)\leq R$. For $R\geq1$, Corollary~\ref{cor:interval-summaries} with $t=R$ implies that a suitable witness can be found among $A\cup B'$, so \eqref{eq:interval-cover} is satisfied. If $R=0$, a vertex of $U_\rho$
cannot be skipped. A selected vertex is covered by itself, while a
vertex outside $U_\rho$ is already covered through the modulator.

At the end of the dynamic programming algorithm, $B=\varnothing$ because
$\overline L_N=\varnothing$. Since $|Y|=D+1$, we have $q=D+1$. Every
$s\in S_0$ has a selected vertex of $Y$ at distance exactly $\rho_s$,
so $W=S_0$, and the selected vertices of $M$ have exactly $t_m$
neighbors in $X$. Hence all conditions in \eqref{eq:interval-accept}
are satisfied, and the algorithm accepts.

For soundness, suppose that the dynamic programming accepts. Follow the stored predecessors and let $X$ be the set of selected vertices of $H$.
We first establish the intended meanings of the state components.

The meanings of $A$, $c$, and $q$ follow by forward induction on the
layers. The update rules directly maintain the last three selected
processed vertices, the number of selected neighbors seen so far for
each $m\in M$, and the number of selected vertices in $M\cup L_i$.

For $W$, once a vertex $s\in S_0$ has a witness $x$ with $d(s,x)=\rho_s$, that equality remains true after further vertices are selected: every subsequently selected vertex belongs to $E_\rho$ and therefore has distance at least $\rho_s$ from $s$. Thus the update of $W$ maintains exactly $W=\{s\in S_0: d(s,M\cup(X\cap L_i))=\rho_s\}$. 

It remains to justify the meaning of $B$. At the final layer,
$B_N=\varnothing=\mathrm L_3(X\cap\overline L_N)$. Suppose that $
B_{i+1}=\mathrm L_3(X\cap\overline L_{i+1})$. At layer $i$, the unprocessed selected vertices consist of $X\cap\overline L_{i+1}$ together with $v$ if $v$ is selected. Therefore, by \eqref{eq:interval-suffix},
$B_i =\mathrm L_3(B_{i+1}\cup Z)  =\mathrm L_3(X\cap\overline L_i)$.  Backward induction proves the claimed meaning of $B$ at every layer.
Consequently, every degree and coverage test uses actual selected
vertices rather than nonexistent future vertices.

Put $Y=M\cup X$. The endpoint conditions imply $a,b\in Y$, and the
final cardinality gives $|Y|=D+1$. The exact meanings of $A$ and $B$ make the set in \eqref{eq:interval-degree} precisely the relevant selected vertices before and after $v$. Hence this test certifies the true degree of every selected vertex of $X$. The final counters certify the exact degrees of the vertices in $M$. Together with the endpoint and cardinality
conditions we can conclude that $G[Y]$ is an $a$--$b$ geodesic.

It remains to show that the guessed distance profile $\rho$ is the actual distances. The preliminary inequalities give $d(s,Y)\geq\rho_s\qquad(s\in S)$. For every $s\in S_0$, the condition $W=S_0$ gives a vertex of $Y$ at
distance exactly $\rho_s$ from $s$. For every $m\in M$, we have
$\rho_m=0$ and $m\in Y$ itself. Thus $d(s,Y)=\rho_s$ for all $s\in S$.

Finally, consider any vertex $v\in V(H)$. If $v\notin U_\rho$, then
there is some $s\in S$ with $d(v,s)+\rho_s\leq R$, and the exactness of $\rho$ together with Lemma~\ref{lem:interval-mod-separator} shows that $v$ is within distance $R$ of $Y$. If $v\in U_\rho$, then either $v$ is selected, or \eqref{eq:interval-cover} provides a selected vertex of $X$ within distance $R$ of $v$. Thus every vertex of $H$ is within distance $R$ of $Y$. Every vertex of $S$ is also within distance $R$ of $Y$ because $\rho_s\leq R$. Therefore $\ecc_G(Y)\leq R$. Hence $Y$ is an $a$--$b$ geodesic of eccentricity at most $R$.
\end{proof}

We now give the running time of the algorithm. For a fixed candidate $(M,\rho)$ with $m=|M|$, each of the summaries $A$ and $B$ has $O(n^3)$ possible values, since it contains at most three vertices. For each $m\in M$, the counter $c(m)$ has at most three possible values, giving at most $3^m$ counter vectors. There are $O(n)$ possible values of $q$, and there are $2^{k-m}$ possible witness sets $W\subseteq S_0$. Hence each layer contains at most $3^m2^{k-m}n^{O(1)}$ states. From each state, there are two choices for whether the current vertex is selected and $O(n^3)$ choices for the new promise $B'$. All tests and updates take polynomial time. Since there are $N\leq n$ layers, the total time for the sweep for one fixed candidate is
$3^m2^{k-m}n^{O(1)}$.

For a prescribed zero set $M$ of size $m$, Corollary~\ref{cor:interval-enumeration} gives at most $5^{k-m}$ candidates. Therefore, summing over all choices of $M$ gives
\begin{equation}\label{eq:thirteen}
 \sum_{M\subseteq S}
  5^{k-|M|}3^{|M|}2^{k-|M|}
 =
 \sum_{m=0}^k
  \binom{k}{m}3^m10^{k-m}
 =13^k.
\end{equation}
The $n^2$ choices of the endpoint pair $(a,b)$ and all polynomial-time preprocessing are absorbed into the $n^{O(1)}$ factor. Since the running time for finding the interval vertex deletion set $S$ is $10^k n^{O(1)}$, we can conclude that the algorithm runs in $13^k n^{O(1)}$ time. 

\subsection{Bandwidth}

Recall that $\bw(G)$ is the minimum, over all vertex orderings, of the
maximum distance between the positions of adjacent vertices. We use the
bandwidth, treewidth, and path-decomposition conventions from
Subsection~\ref{sec:graph-parameters}.

\begin{theorem}\label{thm:bw-fpt}
\ESP{} can be solved in $2^{O(b^2\log(b+1))}n^{O(1)}$ time, where
$b=\bw(G)$. The algorithm returns a witnessing path and does not require
a bandwidth ordering. The optimization version has the same running-time
bound.
\end{theorem}

The known algorithm parameterized by treewidth and $R$
\cite{KuceraSuchy2023} does not suffice here, since the optimum eccentricity
can be unbounded even on cycles. We instead bound the number of attainable
distance values. Throughout, put
\[
 L=1+\lfloor\log_2 n\rfloor,\qquad
 M_0=1,\qquad M_b=(40bL)^b\quad(b\ge1).
\]

\subsubsection{Separated vertices in a rooted tree}
Let $T$ be a tree rooted at a vertex $r$. The \emph{depth} of a vertex
$x\in V(T)$ is $h_T(x)=d_T(r,x)$, the number of edges on the unique
$r$--$x$ path in $T$. In particular, $h_T(r)=0$. Call
$A\subseteq V(T)$ \emph{separated} if
\begin{equation}\label{eq:separated}
 d_T(x,y)\ge2|h_T(x)-h_T(y)|
 \qquad(x,y\in A).
\end{equation}
The following argument adapts the spine-removal and dyadic-grouping
arguments of Eppstein, Goodrich, and Liu
~\cite[Lemmas~12 and~13]{EppsteinGoodrichLiu2025}.

\begin{lemma}\label{lem:tree}
If a rooted tree $T$ has at most $n$ vertices and bandwidth at most $b$,
then every separated set has size at most $M_b$.
\end{lemma}
\begin{proof}
Let $A\subseteq V(T)$ be an arbitrary separated set. We prove that
$|A|\le M_b$ by induction on $b$. If $b=0$, then $T$ has exactly
one vertex, so $|A|\le 1=M_0$. Assume now that $b\ge1$.
If $A=\varnothing$, the claim is immediate, so assume that
$A\ne\varnothing$.

No member of $A$ is a proper ancestor of another, since their distance
would equal their positive depth difference, contradicting
\eqref{eq:separated}. Prune $T$ to the union of the paths from its root
to vertices of $A$, retaining the root and the notation $T$.
This preserves the depths and pairwise distances of vertices of $A$,
does not increase bandwidth, and makes $A$ exactly the set of vertices
with no children.

Fix an ordering $\pi$ of bandwidth at most $b$. Choose $a_-,a_+\in A$
whose root paths contain its first and last vertices, respectively;
such choices exist by the pruning. Let $S$ be the union of these paths.
Then $S$ is connected, contains both
ordering extremes, and contains at most two members of $A$.

Deleting $S$ from the ordering and closing the gaps reduces bandwidth to
at most $b-1$. Indeed, every edge $xy$ outside $S$ whose span is $b$ has a
vertex of $S$ strictly between its endpoints: otherwise the path in $S$
between the ordering extremes would cross this gap using an edge of span
greater than $b$. Deleting that intervening vertex reduces the span of
$xy$; edges of smaller span already satisfy the bound.

Each component $K$ of $T-S$ is joined to $S$ by exactly one edge:
at least one exists because $T$ is connected, and two would create
a cycle since both $K$ and $S$ are connected. Let $u\in S$ be the
endpoint of this edge. Different components may have the same
attachment vertex $u$.
Assign $K$ to one of the two root paths containing $u$, with terminal
vertex $a\in\{a_-,a_+\}$. Choose $x\in A\cap V(K)$, which exists by
the pruning. The lowest common ancestor of $a,x$ is $u$, so
\eqref{eq:separated} gives
\[
 d_T(u,a)+d_T(u,x)\ge2|d_T(u,a)-d_T(u,x)|,
 \qquad d_T(u,a)\le3d_T(u,x)\le3|V(K)|.
\]
Consequently, $d_T(a,y)\le4|V(K)|$ for every $y\in V(K)$.

For each terminal $a$, group its assigned components by
$2^j\le|V(K)|<2^{j+1}$, for $j=0,\ldots,\lfloor\log_2 n\rfloor$.
Set $\ell=2^j$ in one such group. Every vertex
of these components is within distance $8\ell$ of $a$, so its position
differs from $\pi(a)$ by at most $8b\ell$. Their union therefore has at
most $16b\ell+1\le17b\ell$ vertices. Since each component has at least
$\ell$ vertices, there are at most $17b$ components in the group.
There are at most $L$ groups for each terminal, giving at most $34bL$
components in total.

Root each component at its vertex adjacent to $S$. Its vertices of $A$
remain separated, since their distances and depth differences are
unchanged. Each component has bandwidth at most $b-1$, so induction gives
\[
 |A|\le2+34bL M_{b-1}
 \le2+34bL(40bL)^{b-1}
 \le(40bL)^b=M_b.
\]
\end{proof}

\subsubsection{Bounding the distance spectra}
Fix distinct vertices $s,t$, put $D=d(s,t)$ and $a(x)=d(s,x)$, and let
\[
 I=\{x:d(s,x)+d(x,t)=D\}.
\]
On vertex set $I$, form the directed graph $H$ with an arc $x\to y$
whenever $xy\in E(G)$ and $a(y)=a(x)+1$. Its directed $s$--$t$ paths
are exactly the shortest $s$--$t$ paths of $G$. For each vertex $v$, define
its \emph{distance spectrum}
\[
 \Spec_{s,t}(v)=
 \{d(v,V(P)):P\text{ is a shortest }s\text{--}t\text{ path}\}.
\]

\begin{lemma}\label{lem:spectrum}
Once graph distances are available, $\Spec_{s,t}(v)$ can be computed in
$O(n+m)$ time. If $\bw(G)\le b$, then $|\Spec_{s,t}(v)|\le M_b^3$.
\end{lemma}
\begin{proof}
Put $w(x)=d(v,x)$. For $x\in I$, let $A(x)$ be the maximum possible
minimum $w$-value on a directed $s$--$x$ path in $H$; define $B(x)$
analogously for directed $x$--$t$ paths. These values are computed in
forward and reverse layer order by
\begin{align*}
 A(s)&=w(s),& A(x)&=\min\{w(x),\max_{y\to x}A(y)\}\quad(x\ne s),\\
 B(t)&=w(t),& B(x)&=\min\{w(x),\max_{x\to y}B(y)\}\quad(x\ne t).
\end{align*}
We have
\begin{equation}\label{eq:spectrum}
 \Spec_{s,t}(v)=\{w(x):x\in I,\ A(x)=B(x)=w(x)\}.
\end{equation}
Indeed, a minimum-weight vertex $x$ of a shortest path satisfies these
equalities. Conversely, witnessing prefix and suffix paths concatenate
to a path with minimum weight $w(x)$: their strictly increasing layers
ensure that they meet only at $x$. This proves the computation claim.

List the spectrum as $r_1<\cdots<r_q$ and, using \eqref{eq:spectrum},
choose $x_i$ with $w(x_i)=A(x_i)=B(x_i)=r_i$. Write $a_i=a(x_i)$.
For each $x\in I\setminus\{s\}$, choose a predecessor attaining the maximum in the
recurrence for $A$. The chosen edges form a tree $T_s$ rooted at $s$,
with depths $a(x)$, whose root-to-$x$ path attains $A(x)$.
The analogous successor choices give a tree $T_t$ rooted at $t$, with
depths $D-a(x)$, whose root-to-$x$ path attains $B(x)$.
Finally, let $T_v$ be a breadth-first-search tree of $G$ rooted at $v$.
All three trees have bandwidth at most $b$, being subgraphs of $G$.

For $i<j$, put $\Delta r=r_j-r_i>0$ and $\Delta a=a_j-a_i$.
If $z$ is the lowest common ancestor of $x_i,x_j$ in $T_s$, then
$w(z)\ge r_j$. Since distances from $v$ change by at most one along
an edge, $d_{T_s}(z,x_i)\ge\Delta r$. Using the depths in $T_s$,
and arguing symmetrically in $T_t$, gives
\begin{align}
 d_{T_s}(x_i,x_j)&\ge\Delta a+2\Delta r,
 &d_{T_t}(x_i,x_j)&\ge-\Delta a+2\Delta r.\label{eq:trees}
\end{align}
Also, the depth of $x_i$ in $T_v$ is $r_i$, and
\begin{equation}\label{eq:Tv}
 d_{T_v}(x_i,x_j)\ge d(x_i,x_j)\ge|\Delta a|.
\end{equation}

Call indices $i_1<\cdots<i_k$ an \emph{increasing chain} if each
consecutive pair satisfies $0\le\Delta a\le2\Delta r$, and a
\emph{decreasing chain} if each satisfies $-2\Delta r\le\Delta a\le0$.
Summation extends these inequalities to every pair in the chain.
By \eqref{eq:trees}, increasing chains give separated sets in $T_s$,
and decreasing chains give separated sets in $T_t$. Thus each has length
at most $M_b$ by Lemma~\ref{lem:tree}.

Let $p_i^+$ and $p_i^-$ be the maximum lengths of increasing and decreasing
chains ending at $i$. Both lie in $\{1,\ldots,M_b\}$. If $i<j$ and
$|a_j-a_i|\le2(r_j-r_i)$, appending $j$ to a suitable chain ending at
$i$ increases at least one rank. Hence equal rank pairs imply
\[
 (p_i^+,p_i^-)=(p_j^+,p_j^-)
 \quad\Longrightarrow\quad |a_j-a_i|>2(r_j-r_i).
\]
By \eqref{eq:Tv}, each class with the same rank pair is separated in
$T_v$, so it contains at most $M_b$ representatives. There are at most
$M_b^2$ classes, proving $q\le M_b^3$.
\end{proof}

\subsubsection{The algorithm}
First test all single-vertex paths. For each pair of distinct endpoints
$s,t$, construct $H$ and the spectra above, and allow the labels
\[
 \Lambda(v)=\Spec_{s,t}(v)\cap\{0,\ldots,R\}.
\]
If any domain is empty, reject this pair. Otherwise put
$q=\max_v|\Lambda(v)|$.

\begin{lemma}\label{lem:certificate}
A shortest $s$--$t$ path of eccentricity at most $R$ exists if and only if
there are labels $\lambda(v)\in\Lambda(v)$ such that, writing
$Z=\{v:\lambda(v)=0\}$,
\begin{enumerate}
\item $s,t\in Z$ and $|Z|\le D+1$;
\item every vertex of $Z\setminus\{s\}$ has a predecessor in $Z$ in $H$;
\item every vertex $v\in V(G)\setminus Z$ has a neighbor $u$ in $G$
such that $\lambda(u)=\lambda(v)-1$.
\end{enumerate}
\end{lemma}
\begin{proof}
For a witnessing path $P$, take $\lambda(v)=d(v,V(P))$. Then $Z=V(P)$,
and all three conditions hold.

Conversely, follow zero-labeled predecessors from $t$. Each step decreases
the distance from $s$ by one, yielding a directed $s$--$t$ path $P$ in
$H$ with $D+1$ vertices. Since $V(P)\subseteq Z$ and $|Z|\le D+1$,
we have $Z=V(P)$. From any positive label, repeatedly following a neighbor
of label one smaller reaches $Z$ after exactly the starting label many
steps. Thus $P$ is shortest and has eccentricity at most $R$.
\end{proof}

\begin{lemma}\label{lem:DP}
Given a tree decomposition of width $w$, the certificate in
Lemma~\ref{lem:certificate} can be found, or its existence ruled out, in
$2^{O(w)}q^{O(w+1)}n^{O(1)}$ time.
\end{lemma}
\begin{proof}
Use a rooted nice tree decomposition of width $w$ with empty root and
leaves, processed bottom up. At node $x$, let $B_x$ be its bag,
$U_x$ the union of bags in its subtree, and $F_x=U_x\setminus B_x$.
All labels belong to their domains, and $s,t$ receive label zero
whenever present. For a labeling $\lambda$, a \emph{support} for
$v\ne s$ is a zero-labeled predecessor in $H$ when $\lambda(v)=0$,
and a neighbor in $G$ of label $\lambda(v)-1$ otherwise.

A state at $x$ consists of a label $\ell(v)\in\Lambda(v)$ for each
$v\in B_x$ and a Boolean flag $\sigma(v)$ for each
$v\in B_x\setminus\{s\}$. A labeling $\lambda$ of $U_x$ realizes
the state if it obeys the label restrictions, agrees with $\ell$
on $B_x$, and, for every $v\in B_x\setminus\{s\}$,
$\sigma(v)=1$ exactly when $v$ has a support in $F_x$ under $\lambda$.
A false flag therefore allows a support in the current bag or outside
$U_x$. The state's value is
\[
 \min_{\lambda}\bigl|\{u\in F_x:\lambda(u)=0\}\bigr|,
\]
over all realizing labelings for which every vertex of
$F_x\setminus\{s\}$ has a support in $U_x$.
Thus forgotten vertices other than $s$ must already have support,
while requirements for bag vertices may be deferred.
Bag vertices do not contribute to the value until they are forgotten.
If no such labeling exists, the value is $+\infty$.

The empty-leaf state has value zero. At every node, retain the minimum
value over transitions producing the same state.

At an introduce node $x$ with child $y$, let $B_x=B_y\cup\{v\}$.
Extend each child state by an allowed label for $v$, preserving its
other labels, flags, and value. If $v\ne s$, initialize its flag to
false. This is correct because $F_x=F_y$ and $v$ has no neighbor in
$F_y$, by the decomposition axioms.

At a forget node $x$ with child $y$, let $B_x=B_y\setminus\{v\}$.
Here $U_x=U_y$ and $F_x=F_y\cup\{v\}$, so $v$ is the only newly
forgotten vertex. Consider a child state $(\ell,\sigma)$ with finite
value $c$. If $v\ne s$, its support requirement must now be satisfied:
$\sigma(v)=1$ certifies a support in $F_y$; if $\sigma(v)=0$, we
require a support in $B_x$ under $\ell$. If neither exists, discard
the transition. These checks are exhaustive: since $v$ occurs in
no bag outside the subtree rooted at $y$, every neighbor of $v$
lies in $F_y\cup B_x$.
For each remaining vertex $u\in B_x\setminus\{s\}$, set its new flag
to true exactly when $\sigma(u)=1$ or $v$ supports $u$ under $\ell$.
The first possibility accounts for supports in $F_y$, and the
second accounts for the newly forgotten vertex $v$; together they
record precisely whether $u$ has support in $F_x$.
The parent labels are $\ell|_{B_x}$, and the candidate value is
$c+1$ if $\ell(v)=0$, and $c$ otherwise, because only $v$ is newly
counted. When $v=s$, omit its support check, but still update the
remaining flags and add one for its zero label.

At a join node $x$ with children $y,z$, combine child states with
identical bag labels, OR corresponding flags, and add their values.
The sets $F_y,F_z$ are disjoint, have no edges between them, and
their union is $F_x$. Hence their realizing labelings combine
consistently, and their zero counts add without duplication.

The decomposition axioms also ensure that vertices of $F_x$ have
no neighbors outside $U_x$. Consequently, these transitions preserve
valid realizations, and every valid realization restricts to child
states included in the transitions. Induction therefore proves
that the stored minima are exact.

At the empty root, the value is the minimum $|Z|$ under the local
conditions of Lemma~\ref{lem:certificate}; accept if it is at most
$D+1$. There are at most $(2q)^{w+1}$ states per bag.
Pairwise enumeration at joins gives the claimed running time.
Backpointers recover the labels and hence a witnessing path.
\end{proof}

\begin{proof}[Proof of Theorem~\ref{thm:bw-fpt}]
The case $n=1$ is immediate, as is $R\ge n-1$, so assume $b\ge1$ and
$R<n-1$. Precompute all graph distances. Lemma~\ref{lem:spectrum}
computes all spectra over all endpoint pairs in $O(n^3(n+m))$ time.

A bandwidth-$b$ ordering gives a path decomposition of width at most $b$
using consecutive windows of $b+1$ vertices, so $\tw(G)\le b$.
To obtain a decomposition without the ordering, use Korhonen's algorithm
\cite{Korhonen2021}: on input $G,k$, it takes $2^{O(k)}n$ time and either
returns a decomposition of width at most $2k+1$ or certifies
$\tw(G)>k$. Trying $k=1,2,\ldots$ until a decomposition is returned
therefore gives width at most $2b+1$ in $2^{O(b)}n$ time.

For each pair of endpoints, Lemma~\ref{lem:spectrum} gives
$q\le M_b^3$. The computed decomposition has width $w\le2b+1$,
so $w+1=O(b)$, since $b\ge1$.
Lemma~\ref{lem:DP} therefore gives running time
$2^{O(b)}(M_b^3)^{O(b)}n^{O(1)}$ for this pair.
The $O(n^2)$ endpoint choices and the polynomial-time distance and
spectrum computations are absorbed into the factor $n^{O(1)}$.
Substituting $M_b=(40bL)^b$, the total time is consequently
\begin{align*}
 2^{O(b)}(M_b^3)^{O(b)}n^{O(1)}
 &=2^{O(b)}(40bL)^{O(b^2)}n^{O(1)}\\
 &=2^{O(b^2\log(b+1))}L^{O(b^2)}n^{O(1)}.
\end{align*}
Here the last equality separates the factors involving $b$ and $L$,
using $(40b)^{O(b^2)}=2^{O(b^2\log(b+1))}$.

We next absorb the remaining power of $L$ into a factor depending
only on $b$ and one additional factor of $n$.
For an integer $K\ge1$, consider
$f_K(x)=(1+x)^Ke^{-x}$ for $x\ge0$. Its derivative is
\[
 f_K'(x)=(1+x)^{K-1}e^{-x}(K-1-x),
\]
so its maximum is attained at $x=K-1$ and equals
$K^Ke^{1-K}$. Substituting $x=\ln n$ gives
\[
 \frac{(1+\ln n)^K}{n}
 =f_K(\ln n)
 \le K^Ke^{1-K}
 \le eK^K.
\]
Since $L=1+\lfloor\log_2 n\rfloor\le2(1+\ln n)$, we obtain
\[
 L^K
 \le 2^K(1+\ln n)^K
 \le e(2K)^K n
 =2^{O(K\log(K+1))}n.
\]

Choose an absolute constant $C\ge1$ large enough that the exponent
hidden in $L^{O(b^2)}$ is at most $Cb^2$, and put
$K=\lceil Cb^2\rceil$. Since $L\ge1$, the remaining factor is
bounded by $L^K$. Moreover,
$K\log(K+1)=O(b^2\log(b+1))$, and hence
\[
 L^{O(b^2)}
 \le 2^{O(b^2\log(b+1))}n.
\]
Substituting this estimate into the running time above proves
the bound $2^{O(b^2\log(b+1))}n^{O(1)}$.
In particular, the exponent of $n$ is independent of $b$.

Backtracking from an accepting dynamic-programming state returns
a witnessing path. Feasibility is monotone in $R$, so binary search
over $R\in\{0,\ldots,n-1\}$ finds the optimum using $O(\log n)$
decision calls. Running the algorithm at the optimum threshold
and backtracking then returns an optimal path.
These additional calls contribute only a polynomial factor in $n$,
so the optimization version has the same running-time bound.
\end{proof}

\subsection{Treedepth}

Recall that $\td(G)$ is the minimum height of a rooted forest whose
closure contains $G$, with height counted in vertices. The notation
$\cl(F)$, $\tw(G)$, and $\diam(G)$ is as in
Subsection~\ref{sec:graph-parameters}.

We use the following known result directly.

\begin{proposition}[Ku\v{c}era and Such\'y~\cite{KuceraSuchy2023},
Corollary~26]\label{prop:tw-r}
For some computable function $A$, \ESP\ on an $n$-vertex graph $G$
with desired eccentricity $r$ can be solved in time
$A(\tw(G),r)n^3$. If a solution exists, a witnessing shortest path
can be found within the same bound.
\end{proposition}

Allowing singleton paths does not affect this bound: their
eccentricities can all be computed in $O(n^3)$ time.
By enlarging $A$ if necessary, we may assume that it is nondecreasing
in both arguments.

\begin{lemma}\label{lem:td-bounds}
If $d=\td(G)$, then
\[
  \tw(G)\leq d-1
  \qquad\text{and}\qquad
  \diam(G)\leq 2^d-2.
\]
\end{lemma}

\begin{proof}
Let $F$ be a treedepth decomposition of height $d$. Since $G$ is
connected, $F$ is a rooted tree. Use $F$ as the decomposition tree,
with each bag $B_x$ consisting of $x$ and its ancestors. Every vertex
belongs to its own bag. For every edge of $G$, its endpoints are
comparable in the ancestor order and hence belong to the bag indexed
by the lower endpoint. The bags containing a fixed vertex $v$ are
exactly those indexed by $v$ and its descendants, so they form a
connected subtree. Thus these bags form a tree decomposition, and
$|B_x|\leq d$ gives $\tw(G)\leq d-1$.

For the second bound, we prove by induction on $h\geq1$ that the
closure of any rooted tree of height at most $h$ contains no path
with more than $2^h-1$ vertices. The case $h=1$ is immediate.
For $h\geq2$, deleting the root from a path leaves at most two path
segments. Each segment lies in a single child subtree, since there
are no edges between different child subtrees. By induction, each
segment has at most $2^{h-1}-1$ vertices. The original path therefore
has at most
\[
  1+2(2^{h-1}-1)=2^h-1
\]
vertices. This also covers a path avoiding the root, which lies in
one child subtree. Every path in $G$ is a path in $\cl(F)$, so every
shortest path has at most $2^d-2$ edges. Hence
$\diam(G)\leq2^d-2$.
\end{proof}

\begin{theorem}\label{thm:td-fpt}
For some computable function $F$, \ESP\ can be solved in time
$F(\td(G))n^3$. A shortest path of minimum eccentricity can also be
found within such a running time. No treedepth decomposition is
required as part of the input.
\end{theorem}

\begin{proof}
Given $(G,r)$, compute $\Delta:=\diam(G)$ by breadth-first search
from every vertex, in $O(n^3)$ time. If $r\geq\Delta$, return any
singleton path: it is shortest and has eccentricity at most
$\Delta\leq r$.

Otherwise, let $d:=\td(G)$. Lemma~\ref{lem:td-bounds} gives
\[
  \tw(G)\leq d-1,
  \qquad
  0\leq r<\Delta\leq2^d-2.
\]
Apply Proposition~\ref{prop:tw-r}. Its running time is at most
$A(d-1,2^d-2)n^3$, which has the required form. The algorithm uses
neither $d$ nor a treedepth decomposition; these are needed only for
the running-time analysis.

For optimization, apply the decision algorithm for
$r=0,1,\ldots,\Delta$ and return the witness from the first
successful call. The first such value is the minimum eccentricity
of a shortest path, and a successful call exists because
$r=\Delta$ always suffices. There are at most
$\Delta+1\leq2^d-1$ calls, so their total running time is still
$F(d)n^3$ for a computable function $F$.
\end{proof}

\section{XP algorithms}
\label{sec:xp-algorithms}

\subsection{Diameter}

Recall that $\diam(G)$ is the maximum distance between two vertices of
the connected graph $G$, and $r(G)$ is its optimum ESP eccentricity; see
Subsection~\ref{sec:graph-parameters}.

\begin{lemma}
\label{lem:rleqdiameter}
Let $G$ be connected and let $r(G)$ be the optimum ESP eccentricity.  Then
$
 r(G)\leq \diam(G).
$
\end{lemma}

\begin{proof}
Choose an arbitrary vertex $v\in V(G)$. By definition,
\[
 \ecc_G(\{v\})=\max_{x\in V(G)}\dist_G(v,x).
\]
Every pairwise distance in $G$, including $\dist_G(v,x)$, is at most
$\diam(G)$.  Taking the maximum over $x$ gives $\ecc_G(\{v\})
 \leq \diam(G)$. The one-vertex path $P=(v)$ is a
shortest path of length zero.  By optimality of $r(G)$, we have
$
 r(G)
 \leq \ecc_G(P)
 =\ecc_G(\{v\})
 \leq \diam(G).
$
\end{proof}

\begin{proposition}[Corollary~12, \cite{DraganLeitert2017}]
\label{cor:DL12}
If the optimum eccentricity of a shortest path in $G$ is $r$, then the optimization version of ESP can be solved in
$O(n^{2r+2}m)$ time, even if $r$ is not known in advance.
\end{proposition}

We can now state the main result.

\begin{theorem}
\label{thm:maindiameter}
ESP is in $\mathsf{XP}$ when parameterized by
the maximum diameter $d$ of a connected component.
\end{theorem}

\begin{proof} 
By Proposition~\ref{cor:DL12} ESP is in $\mathsf{XP}$ parameterized by $r(G)$ and by Lemma~\ref{lem:rleqdiameter} $r(G) \leq \diam(G)$, hence ESP is in $\mathsf{XP}$ parameterized by diameter.
\end{proof}

\subsection{Distance to chordal graphs}
\label{sec:mesp-chordal}

Recall that the distance to chordal graphs is the minimum size of a set
$S$ such that $G-S$ is chordal; see Subsection~\ref{sec:graph-parameters}.

The \ESP{} problem is polynomial-time solvable on chordal graphs~\cite{DraganLeitert2016Structured}. On the other hand, Bhyravarapu et al.\
showed that the problem is $W[2]$-hard when parameterized by deletion distance to chordal graphs~\cite{BhyravarapuEtAl2023Full}. We complement this hardness result by showing that \ESP{} is in
XP parameterized by the distance to chordal graphs.

\subsubsection{Clique bags and distance profiles}

We have a graph $G$ with $|V(G)|=n$ and a set $S \subseteq V(G), |S| \leq k$ such that $H=G-S$ is chordal. If $S$ is not given, we enumerate all subsets of $V(G)$ of size at most $k$ and test whether the resulting graph is chordal. This takes $n^{O(k)}$ time. Hence, for the XP running time, we may assume that such a set $S$ is given. We have the following lemma.

\begin{lemma}\label{lem:chordal-decomposition}
Let $G$ be a graph on $n$ vertices, let $S\subseteq V(G)$, and let
$H=G-S$ be chordal. From $G$ and $S$ one can construct in polynomial time a rooted nice tree
decomposition $(T,\mathcal B)$ of $G$ with the following properties.
\begin{enumerate}
 \item The root and every leaf have bag $S$.
 \item Every bag has the form   $B_t=S\mathbin{\dot\cup} C_t$, where $C_t$ is a clique of $H$.
 \item Only vertices of $H$ are introduced or forgotten.
 \item The decomposition has $O(n^2)$ nodes.
\end{enumerate}
\end{lemma}
\begin{proof}
Since $H$ is a chordal graph, a tree decomposition of $H$ whose bags are cliques can be constructed in polynomial time and has
$O(|V(H)|)$ bags ~\cite{Gavril1974, BlairPeyton1993}. Add $S$ to every bag. This gives a tree decomposition of $G$ in which every bag is $S$ together with a clique of $H$.

Apply the standard conversion to a nice tree decomposition
~\cite{Kloks1994}. Along an original tree edge with endpoint bags
$S\dot\cup C$ and $S\dot\cup C'$, first forget the vertices of
$C\setminus C'$ one at a time and then introduce the vertices of
$C'\setminus C$ one at a time. Every intermediate non $S$ part in the bag is a subset of $C$ or of $C'$, and hence remains a clique. Duplicate bags to create join nodes and attach the usual chains ending in the bag $S$ at the root and the leaves. Thus only vertices of $H$ are introduced or forgotten. There are $O(|V(H)|)$ original bags and at most $O(|V(G)|)$ new nodes per original tree edge, so the resulting decomposition has $O(|V(G)|^2)$ nodes.
\end{proof}

Our algorithm is a dynamic programming algorithm on the above nice tree decomposition of $G$. The decomposition is useful only if distances from a potentially large clique to a candidate geodesic can be summarized compactly. The next lemma provides exactly such a summary: one vertex of the geodesic
simultaneously realizes the minimum distance from every vertex of the
clique. 

\begin{lemma}\label{lem:chordal-projection}
Let $J$ be a chordal graph, $C \subset V(J) $ be a clique, and $Q$ be a geodesic in $J$ disjoint from $C$. There exists a vertex $q\in V(Q)$ such that 
 $\forall c\in C, \operatorname{dist}_J(c,q)=\operatorname{dist}_J(c,V(Q))$.
\end{lemma}
\begin{proof}
We will prove by induction on $|V(J)|$. The statement is immediate when $|V(J)|=1$ , $|V(Q)|=1$, or if $J$ is a clique. Since $C \subset V(J)$, there is a simplicial
vertex $v\notin C$: this follows from Dirac's theorem that a noncomplete chordal graph has two nonadjacent simplicial vertices~\cite{Dirac1961}. 

Since $v$ is simplicial, no shortest path between two other vertices can
use $v$ as an internal vertex: if it did, its two neighbors on the path
would be adjacent, yielding a shorter path. Thus, deleting $v$ does not change the distances between any two vertices of $V(J)\setminus\{v\}$. If $v \notin V(Q)$, the statement follows by applying induction on the graph $J-v$. Since distances between vertices of $Q$ are unchanged, $Q$ remains a geodesic in $J-v$. As $v$ is simplicial, it also cannot be an internal vertex of $Q$. Thus $v$ is an endpoint of $Q$. Let $u$ be unique neighbor of $v$ in $Q$. 

Suppose that $v$ is a closest vertex of $Q$ to some $c\in C$. Look at the shortest path from $c$ to $Q$ ending at $v$. Let $x$ be the penultimate vertex of this shortest path. If $x=u$, then we have a shorter path from $c$ to $Q$ ending at $u$. Thus $x \neq u$. Then since $v$ is simplicial and $x, u \in N(v)$, $xu \in E(J)$. Then we can replace $xv$ by $xu$ giving a $c$-$u$ path of the same length as the $c$-$v$ path. Hence $u$ is also a closest vertex of $Q$ to $c$. Thus, $\forall c\in C, \operatorname{dist}_{J}(c,V(Q-v)) = \operatorname{dist}_{J}(c,V(Q))$.

Now $Q-v$ is a geodesic in $J-v$: it is a subpath of the geodesic $Q$, and the distances between its endpoints are unchanged after deleting $v$. Since $C\subseteq V(J-v)$ is still a clique, the induction hypothesis applied to $J-v$ and $Q-v$ gives a vertex $q\in V(Q-v)$ such that $\forall c\in C, \operatorname{dist}_{J-v}(c,q) = \operatorname{dist}_{J-v}(c,V(Q-v)) $.
Using distance preservation and the above claim, we obtain $\forall c\in C, \operatorname{dist}_{J}(c,q) = \operatorname{dist}_{J}(c,V(Q))$. This proves the lemma.

\end{proof}

We have the following lemma.
\begin{lemma}\label{lem:chordal-representatives}
Let $C$ be a clique of $H=G-S$. Suppose $X\subseteq V(H)\setminus C$ is the union of at most $h$ geodesics of $H$.
Then there is a set $W\subseteq X$ with $|W|\leq h+k$ such that
$d(b,W)=d(b,X)$ for every $b\in S\cup C$.
\end{lemma}
\begin{proof}
If $X=\varnothing$, take $W=\varnothing$. Otherwise, for each of
the at most $h$ geodesics $Q$, clique $C$ and chordal graph $H$, apply Lemma~\ref{lem:chordal-projection} to obtain vertices $q \in Q$ and add it to $W$. For every $s\in S$, choose a vertex
$x_s\in X$ minimizing $d(s,x_s)$, and add it to $W$.
Then $|W|\leq h+k$, and $d(s,W)=d(s,X)$ for every $s\in S$.

Fix $c\in C$. If $X$ is unreachable from $c$, then both $d(c,X)$ and
$d(c,W)$ are infinite. Otherwise, take a shortest $c$--$X$ path. If it
avoids $S$, it lies in $H$ and ends in one of the geodesic pieces of $X$;
the vertex obtained from applying Lemma~\ref{lem:chordal-projection} for that piece is no farther from $c$. If the
path meets a vertex $s\in S$, then $ d(c,x_s) \leq d(c,s)+d(s,X) \leq d(c,X)$.
In either case $d(c,W)\leq d(c,X)$. The reverse inequality follows from
$W\subseteq X$, and hence the equality holds for every $c\in C$.
\end{proof}

Let us fix the decomposition from Lemma~\ref{lem:chordal-decomposition}. For a node $t$, let $V_t$ be the union of all bags in the subtree rooted at $t$, and let $I_t=V_t\setminus B_t$ and $O_t=V(G)\setminus V_t$. The separator property of tree decompositions implies that there is no edge
between $I_t$ and $O_t$.

Let $P$ be a geodesic of $G$. Since $P$ is induced and $C_t$ is a clique, $P\cap C_t$ consists of at most two consecutive vertices. Delete this block, together with the at most $k$ vertices of $P\cap S$. At most $k+2$ subpaths remain. Each lies entirely in $I_t$ or $O_t$, because there is no edge between the two sets. Moreover, every such subpath is a geodesic of $H$: it is a geodesic subpath of $P$, and a shorter route in $H$ would also be a shorter route in $G$. Therefore Lemma~\ref{lem:chordal-representatives} applies to each of
$V(P)\cap I_t$ and $V(P)\cap O_t$ with $h=k+2$.

We use the convention $d(v,\emptyset)=\infty$. For $Z\in\{I_t,O_t\}$, enumerate the following family of vectors on $B_t$:
\begin{equation}\label{eq:universe}
 \mathcal P_t(Z)=
 \left\{
   \bigl(d(b,W)\bigr)_{b\in B_t}:
   W\subseteq Z,\ |W|\leq 2k+2
 \right\}.
\end{equation}
We call the vectors $ \bigl(d(b,W)\bigr)_{b\in B_t}$ as {\em distance profiles}. Include the all-$\infty$ vector, obtained from $W=\varnothing$, and discard duplicate vectors. Each family has at most $(n+1)^{2k+2}$ members and can be enumerated from one all-pairs distance matrix of $G$. Let $\mathcal P_t^-:=\mathcal P_t(I_t)$ and $ \mathcal P_t^+:=\mathcal P_t(O_t)$.

For a geodesic $P$ and a side $Z\in\{I_t,O_t\}$, Lemma~\ref{lem:chordal-representatives} provide a set $W\subseteq V(P)\cap Z$ of size at most $2k+2$ such that for all $b \in B_t, d(b,W)=d\bigl(b,V(P)\cap Z\bigr)$. Thus the vector recording the distances from the bag vertices to
$V(P)\cap I_t$ belongs to $\mathcal P_t^-$, and the corresponding vector
for $V(P)\cap O_t$ belongs to $\mathcal P_t^+$. Explicitly, $ \bigl(d(b,V(P)\cap I_t)\bigr)_{b\in B_t}\in\mathcal P_t^-$, $\bigl(d(b,V(P)\cap O_t)\bigr)_{b\in B_t}\in\mathcal P_t^+$.

The dynamic program will also need to transfer a distance vector across a
separator. For a vector $p$ indexed by a set $K$, define its extension by
\begin{equation}\label{eq:transfer}
 (\mathcal T_Kp)(v)
 =\min_{b\in K}\bigl(d(v,b)+p(b)\bigr),
\end{equation}
Suppose that $p(b)=d(b,X)$ for every $b\in K$ and that every $v$-$X$ path meets $K$. Then we can prove that  
\begin{equation}\label{eq:separator}
 d(v,X)=(\mathcal T_Kp)(v).
\end{equation}
Indeed, concatenating a shortest $v$-$b$ path with a shortest $b$-$X$
path gives for all $b \in K, d(v,X)\leq d(v,b)+p(b)$. Conversely, a shortest $v$-$X$ path meets $K$ at some vertex $b$, and splitting the path at $b$ gives $ d(v,X)\geq d(v,b)+d(b,X) \geq(\mathcal T_Kp)(v)$.

\subsubsection{Certifying a geodesic}

The algorithm guesses distinct endpoints $a,b$, put $D=d(a,b)$. Let
$F=\{v\in V(G):d(a,v)+d(v,b)=D\}$ and $\ell(v)=d(a,v)$. The set $F$ consists exactly of the vertices that lie on some $a$-$b$ geodesic, so only vertices of $F$ need to be selected. Moreover, if $xy\in E(G)$, then $|\ell(x)-\ell(y)|\leq1$.

For a selected set $Y\subseteq F$, forbid edges of $G[Y]$ whose endpoints
belong to the same layer. Orient every remaining edge from the smaller layer to the larger layer, and define the required degree pair
\[
 \tau(v)=
 \begin{cases}
  (0,1),&v=a,\\
  (1,0),&v=b,\\
  (1,1),&v\notin\{a,b\},
 \end{cases}
\]
where the two coordinates are indegree and outdegree.

\begin{lemma}\label{lem:chordal-degrees}
A set $Y\subseteq F$ induces an $a$-$b$ geodesic if and only if it
contains $a,b$, has no same-layer edge, and every vertex $v\in Y$ has
degree pair $\tau(v)$ in the orientation above.
\end{lemma}

\begin{proof}
If $G[Y]$ is an $a$-$b$ geodesic, then it is induced, and the distance
from $a$ increases by one along every path edge. The three conditions are
therefore necessary.

Conversely, assume that $Y$ contains $a$ and $b$, that $G[Y]$ has no same-layer edge, and that every vertex $v\in Y$ has degree pair $\tau(v)$. For every edge $xy$, the triangle inequality gives $|\ell(x)-\ell(y)|\leq1$. Since same-layer edges are excluded, every oriented edge goes from layer $i$ to layer $i+1$. Thus the orientation of edges in $Y$ is acyclic.

The degree conditions say that $a$ has one outgoing edge and no incoming
edge, $b$ has one incoming edge and no outgoing edge, and every other
selected vertex has one edge in each direction. Thus every component of $G[Y]$ is a directed path. No second component can exist: every directed path has a first vertex with indegree zero and a last vertex with outdegree zero, but only $a$ and $b$ have these degree types. Hence all vertices of $Y$ form one directed path from $a$ to $b$. 
Each edge of this path is from layer $i$ to layer $i+1$. Since $a$ is in layer zero and $b$ is in layer $D$, the path has length $\ell(b)-\ell(a)=d(a,b)=D$. Thus $G[Y]$ is an $a$-$b$ geodesic.
\end{proof}

\subsubsection{The dynamic programming algorithm}
Fix an endpoint guess $(a,b)$ and a target eccentricity $R$. At a node $t$,
a state consists of
\begin{itemize}
 \item a set $A_t\subseteq B_t\cap F$ of selected bag vertices;
 \item a pair $e_t(x)\in\{0,1\}^2$ for every $x\in A_t$; and
 \item distance profiles $\alpha_t\in\mathcal P_t^-$ and
       $\omega_t\in\mathcal P_t^+$.
\end{itemize}
The two coordinates of $e_t(x)$ are the current indegree and outdegree of
$x$. For a selected solution set $Y$, the intended meaning of the state is
\begin{equation}\label{eq:meaning}
 \begin{split}
 A_t&=Y\cap B_t,\\
 e_t(x)&=\text{the degree pair of $x$ in the oriented graph }
          G[Y\cap V_t] \qquad \forall x\in A_t,\\
 \alpha_t(x)&=d(x,Y\cap I_t) \qquad \forall x\in B_t,\\
 \omega_t(x)&=d(x,Y\cap O_t) \qquad \forall x\in B_t.
 \end{split}
\end{equation}

We keep only states satisfying the following conditions:
\[
 |A_t\cap C_t|\leq2,
 \qquad
 \{a,b\}\cap B_t\subseteq A_t,
 \qquad
 e_t(x)\leq\tau(x)\quad \forall x\in A_t,
\]
where the last inequality is coordinatewise. We also require that
$G[A_t]$ has no same-layer edge. The first condition is valid because an
induced path contains at most two vertices of a clique; the second ensures
that an endpoint is selected whenever it occurs in a bag; and the degree
bound prevents a partial solution from already exceeding its required final
degree.

Finally, for every vertex $v\in B_t$, require
\begin{equation}\label{eq:domination}
 \min\{\alpha_t(v),\omega_t(v),d(v,A_t)\}\leq R.
\end{equation}
The three terms measure the distance from $v$ to the selected vertices in
$I_t$, in $O_t$, and in $B_t$, respectively. Hence this condition says that
every vertex currently appearing in the bag is within distance $R$ of the
eventual path. 

A transition is valid precisely when the parent and child states satisfy
the corresponding equations below and both states satisfy all state
conditions. For two vectors on the same set, let $p\wedge q$ denote their
pointwise minimum. When a transition handles a vertex $v$, let
$ \sigma_v(x)= d(x,v)$ when $v$ is selected and $ \sigma_v(x)= \infty$ otherwise.

\paragraph*{Transitions}

For a selected set $A\subseteq F$ with no same-layer edge, let $e_A(x)$
denote the degree pair of $x$ in the oriented graph $G[A]$. All additions,
subtractions, and comparisons of degree pairs are coordinate wise. We now
describe how a state at a node is obtained from states at its children.

\subparagraph*{Leaf Node}
Let $t$ be a leaf. Here $B_t=V_t=S$ and $I_t=\varnothing$. Therefore $\alpha_t$ is an all $\infty$ vector, and no selected vertex lies strictly inside the subtree. For every admissible choice of $A_t$ and $\omega_t\in\mathcal P_t^+$, set $e_t(x)=e_{A_t}(x)$ for all $x\in A_t$. These are all leaf states.

\subparagraph*{Introduce Node}
Suppose that $t$ introduces a vertex $v$ and has child $c$. Thus $B_t=B_c\cup\{v\}, I_t=I_c, O_c=O_t\cup\{v\}$. We require $A_t\cap B_c=A_c$. If $v\notin A_t$, the degree pairs do not change. If $v\in A_t$, then the only new selected edges are those from $v$ to selected vertices in $A_c$; there is no edge from $v$ to $I_c$. 

Since $I_t=I_c$, the selected solution vertices in the strict interior are the same at $t$ and $c$. The values of $\alpha_t$ on $B_c$ therefore agree with $\alpha_c$. To determine the new value $\alpha_t(v)$, observe that every path from $v$ to $I_c$ meets $B_c$. The separator identity \eqref{eq:separator} therefore computes the distance from $v$ to the selected interior using $\alpha_c$.

Since $O_c=O_t\cup\{v\}$, the selected solution vertices outside the child subtree are the selected vertices in $O_t$, together with $v$ when $v$ is selected. The distance from a vertex $x\in B_c$ to this union is the smaller of $\omega_t(x)$ and $d(x,v)$; if $v$ is not selected, the second value is interpreted as $\infty$. Thus, we require
\begin{equation}\label{eq:introduce}
 \begin{split}
  \alpha_t(x)&=(\mathcal T_{B_c}\alpha_c)(x)
       \qquad \forall x\in B_t,\\
  \omega_c(x)&=\min\{\omega_t(x),\sigma_v(x)\}
       \qquad \forall x\in B_c.
 \end{split}
\end{equation}

\subparagraph*{Forget Node}
Suppose that $t$ forgets a vertex $v$ and has child $c$. Then
\[
 B_c=B_t\cup\{v\},
 \qquad
 I_t=I_c\cup\{v\},
 \qquad
 O_c=O_t.
\]
We require $A_t=A_c\setminus\{v\}$ and $e_t(x)=e_c(x)$ for every $x\in A_t$. If $v\in A_c$, also require $e_c(v)=\tau(v)$. The latter is because after $v$ is forgotten, it belongs to $I_t$, and the separator property implies that it has no neighbor in $O_t$.

If $v$ is selected, it is added to the solution part of the interior; otherwise this is unchanged. The exterior is unchanged, but its distance profile must be extended from $B_t$ to the larger child bag $B_c$. Hence, we require
\begin{equation}\label{eq:forget}
 \begin{split}
  \alpha_t(x)&=\min\{\alpha_c(x),\sigma_v(x)\}
       \qquad \forall x\in B_t,\\
  \omega_c(x)&=(\mathcal T_{B_t}\omega_t)(x)
       \qquad \forall x\in B_c.
 \end{split}
\end{equation}
The second identity is valid because $B_t$ separates the forgotten vertex $v$ from $O_t$.

\subparagraph*{Join Node}
Let $t$ be a join node with children $c_1,c_2$. All three bags are the same. Thus, we require the solution intersecting with the bag $A_t=A_{c_1}=A_{c_2}=:A$.

The two child subgraphs intersect exactly in $G[A]$. Consequently, each
selected boundary edge is counted in both child degree pairs and must be
subtracted once. Thus, we require
\begin{equation}\label{eq:join-degrees}
e_t(x)=e_{c_1}(x)+e_{c_2}(x)-e_A(x)
\qquad \forall x\in A.
\end{equation}

The interior of the parent is the union of the two child interiors. From the point of view of either child, its exterior is the union of the parent exterior and the other child's interior. Since the distances to a union are obtained by taking the pointwise minimum, we require
\begin{equation}\label{eq:join-profiles}
 \begin{split}
  \alpha_t&=\alpha_{c_1}\wedge\alpha_{c_2},\\
  \omega_{c_1}&=\omega_t\wedge\alpha_{c_2},\\
  \omega_{c_2}&=\omega_t\wedge\alpha_{c_1}.
 \end{split}
\end{equation}

\subparagraph*{Root and evaluation.}
At the root $\rho$, the exterior is empty, so require $\omega_\rho\equiv\infty$. The root bag is $S$, and every selected vertex
outside $S$ has already had its final degree checked at its forget node.
For the remaining selected vertices $x \in A_\rho$, require $  e_\rho(x)=\tau(x)$.

A state at each node is retained when it satisfies all the required conditions at that node. The tables are computed bottom-up. The endpoint guess $(a,b)$ succeeds if some retained root state satisfies the root requirements. 

\subsubsection{Correctness and running time}

We now prove that for every fixed choice of endpoints $a,b$, the dynamic program accepts if only if there is a $a-b$ geodesic of eccentricity at most $R$. 

\paragraph*{Completeness.}
Suppose that $P$ is an $a-b$ geodesic with $\ecc_G(P)\le R$, and let $Y=V(P)$. We show that there is a valid state with respect to $P$ at every node of the decomposition. For a node $t$, let $A_t = Y \cap B_t$. For $x \in A_t$, let $e_t(x)$ be the indegree and outdegree of $x$ in the oriented graph induced by $Y \cap V_t$. Define $\alpha_t(x) = d(x, Y \cap I_t)$ and $\omega_t(x) = d(x, Y \cap O_t)$ for every $x \in B_t$.

Consider a bag $B_t = S \cup C_t$. Since $P$ is a geodesic, it is induced.
As $C_t$ is a clique, $P$ contains at most two vertices of $C_t$, and if
it contains two, they are consecutive on $P$. Delete these at most two
vertices together with the at most $k$ vertices of $P\cap S$. The path
$P$ is then split into at most $k+2$ subpaths. Since there are no edges between $I_t$ and $O_t$, each remaining subpath lies completely in one of $I_t$ and $O_t$. Moreover, every such subpath is a geodesic of $H$; it is a subpath of the geodesic $P$, and any shorter path in $H$ between the endpoints would also be a shorter path in $G$. Thus $Y \cap I_t$ and $Y \cap O_t$ are each the union of at most $k+2$ geodesics of $H$. By Lemma \ref{lem:chordal-representatives},  for each side there is a set of at most $(k+2)+k = 2k+2$ vertices inducing exactly the same distances to every vertex of $B_t$. Thus, we have $\alpha_t \in \mathcal P_t^-$ and $\omega_t \in \mathcal P_t^+$. 

We now prove the remaining conditions needed for a valid DP state. Since $P$ is induced and $C_t$ is a clique, $|A_t \cap C_t| \leq 2$ and since $a, b \in Y$, every occurrence of $a$ or $b$ in a bag belong to $A_t$. The orientation of $P$ by increasing distance from $a$ gives the required degree pair $\tau(v)$, and no edge of $P$ has both endpoints in the same layer. Hence all partial degree pairs are coordinatewise at most $\tau$. 

It remains to check the coverage condition. For every $v \in B_t$, the sets $Y \cap I_t, A_t, Y \cap O_t$ partition $Y$, Therefore 
\[
d(v,Y)
=
\min\bigl\{
d(v,Y\cap I_t),\,d(v,A_t),\,d(v,Y\cap O_t)
\bigr\}
=
\min\{\alpha_t(v),d(v,A_t),\omega_t(v)\}.
\] 
Since $\ecc_G(P)\le R$, this value is at most $R$, condition \ref{eq:domination} is satisfied at every bag. 

We now have to verify the transitions. At a leaf node,
$I_t=\emptyset$, so $\alpha_t\equiv\infty$ and the stated leaf state is
obtained. At an introduce node, $I_t=I_c$ and $B_c$ separates $B_t$ from
$I_c$. Since $\alpha_c(b)=d(b,Y\cap I_c)$ for every $b\in B_c$,
equation~\eqref{eq:separator}, applied with $K=B_c$ and
$X=Y\cap I_c$, yields
\[
\alpha_t(x)
=d(x,Y\cap I_t)
=d(x,Y\cap I_c)
=\min_{b\in B_c}\{d(x,b)+\alpha_c(b)\}
=(\mathcal T_{B_c}\alpha_c)(x)
\]
for every $x\in B_t$. This is precisely the first identity in
equation~\eqref{eq:introduce}. The second identity in
equation~\eqref{eq:introduce} follows from
$O_c=O_t\cup\{v\}$.

At a forget node, the forgotten vertex $v$ is added to $I_t$, giving the
first identity in equation~\eqref{eq:forget}. Moreover, $O_c=O_t$ and
$B_t$ separates $B_c$ from $O_t$, so the separator identity gives the
second identity in equation~\eqref{eq:forget}.

At a join node, the interiors of the two children are disjoint and their
union is $I_t$, so the distance to the selected vertices in $I_t$ is the
pointwise minimum of the two child distances. Similarly, the exterior of
each child is the union of the parent's exterior and the other child's
interior, which gives the remaining two identities in
equation~\eqref{eq:join-profiles}. The degree equation is
equation~\eqref{eq:join-degrees}, since the common boundary edges are
counted in both children and hence must be subtracted once.

Thus, the states induced by $P$ are present throughout the decomposition,
and the root satisfies all acceptance conditions. Hence every solution
path gives an accepting computation.

\paragraph*{Soundness.}
Conversely, suppose that the dynamic program accepts for the endpoint
guess $(a,b)$. Fix an accepting root state and, for every state in this
accepting computation, choose the predecessor state(s) witnessing the
corresponding transition. These states determine a global set $Y$ of
selected vertices. Indeed, once a vertex is introduced, its selection
status is carried through all subsequent bags containing that vertex
until it is forgotten; at a join node, both children are required to
have the same selection on their common bag. Thus every vertex has a
unique, consistent selection status, and hence $Y\subseteq F$ with
$a,b\in Y$.

We claim that, for every node $t$ and every $x\in A_t$, the value $e_t(x)$ is the indegree and outdegree of $x$ in the oriented graph $G[Y\cap V_t]$. This follows by induction over the decomposition. At an introduce node, the transition accounts for exactly the newly exposed edges incident with the introduced vertex; at a forget node, the degree of the forgotten vertex is already final and is checked against $\tau(v)$. At a join node, the two child subgraphs intersect precisely in $G[A_t]$, so the common boundary edges are counted twice and are therefore subtracted once.

Since every edge of $G$ is contained in some bag, the condition that $G[A_t]$ should not contain a same-layer edge implies that $G[Y]$ has no
same-layer edge. Every selected vertex outside $S$ is eventually
forgotten, at which point its degree is final and is required to equal
$\tau(v)$. For vertices of $S$, this is checked at the root. Thus every vertex of $Y$ has degree pair $\tau(v)$.

Hence, by Lemma \ref{lem:chordal-degrees}, $G[Y]$ is an $a$--$b$ geodesic. Its length is $\ell(b)-\ell(a)=d(a,b)$.

It remains to prove that the stored distance profiles have their intended
meanings. More precisely, for every node $t$ and every $x\in B_t$, the
profiles in the states chosen above satisfy
\begin{equation}\label{eq:profile-invariant}
\alpha_t(x)=d(x,Y\cap I_t)
\qquad\text{and}\qquad
\omega_t(x)=d(x,Y\cap O_t).
\end{equation}

The statement for the inside profiles $\alpha_t$ follows by induction from the leaves towards the root. At a leaf, $I_t=\emptyset$ and hence $\alpha_t\equiv\infty$. At an introduce node, $I_t=I_c$ and $B_c$ separates $B_t$ from $I_c$, so the first identity in equation~\eqref{eq:introduce} follows from the separator identity. At a forget node, the forgotten vertex is added to the strict interior, and therefore the new profile is obtained by taking the pointwise minimum with its distance vector, as in equation~\eqref{eq:forget}. At a join node, the two strict interiors are disjoint and their union is $I_t$, so the first identity in equation~\eqref{eq:join-profiles} is exactly the distance to $Y\cap I_t$.

The statement for the outside profiles $\omega_t$ is proved symmetrically from the root towards the leaves. At the root the exterior is empty. The introduce, forget, and join transitions in equations \eqref{eq:introduce}--\eqref{eq:join-profiles} update the exterior by adding the introduced vertex, extending across the separator, and taking the union with the sibling interior, respectively. Hence the second identity in equation~\eqref{eq:profile-invariant} holds throughout the
chosen computation.

Now let $v\in V(G)$ and choose a bag $B_t$ containing $v$. Since
$Y=(Y\cap I_t)\,\dot\cup\,A_t\,\dot\cup\,(Y\cap O_t)$, we obtain
$d(v,Y) = \min\{\alpha_t(v),d(v,A_t),\omega_t(v)\} \le R$, where the inequality follows from condition~\eqref{eq:domination}. Therefore $\ecc_G(Y)\le R$, and the reconstructed geodesic is a solution.

\paragraph*{Running time.}
For a bag $B_t$, there are at most $2^k\left(1+|C_t|+\binom{|C_t|}{2}\right) = n^{O(1)}2^k$ choices for $A_t$. For each selected boundary vertex there are at most four possibilities for its degree pair, and each distance profile has at most $(n+1)^{2k+2}$ possibilities. Hence the number of states at a node is $n^{O(k)}$. Since the decomposition has $O(n^2)$ nodes and there are $O(n^2)$ choices for the endpoints $a,b$, all transitions can be checked in $n^{O(k)}$ time. Thus the algorithm runs in
$n^{O(k)}$ time.

\subsection{Clique-width}

Recall that $\cw(G)$ is the minimum number of labels in an expression
constructing $G$ by the four operations in
Subsection~\ref{sec:graph-parameters}.

For a connected $n$-vertex input graph, we may replace $R$ by
$\min\{R,n-1\}$.

\begin{theorem}\label{thm:cliquewidth-xp}
Let $G$ be a connected $n$-vertex graph.
\begin{enumerate}
\setlength{\itemsep}{0pt}
\setlength{\parskip}{0pt}
 \item[\textup{(a)}] Given a clique-width $k$-expression $\tau$ for $G$, \ESP{} can be
 solved in $|\tau|^{O(1)}+n^{O(k)}$ time.
 \item[\textup{(b)}] Without a supplied expression, \ESP{} can be solved in
 $f(k)n^{O(k)}$ time, where $k=\cw(G)$ and $f$ is computable.
\end{enumerate}
Both algorithms return a witnessing shortest path whenever one exists.
\end{theorem}

Consequently, \ESP{} is in XP parameterized by clique-width. We first
prove the stronger intermediate result for a supplied mim-width layout.

\subsubsection{Layouts and the model-checking theorem}

\subparagraph*{Width measures.}
We use the cut functions $\mim_G(A)$ and $\rho_G(A)$, rooted layouts
$\calL$, and rank-decompositions from Subsection~\ref{sec:graph-parameters}.
Recall that $\mim_G(A)$ is the maximum size of an induced matching in the
bipartite cut graph, while $\rho_G(A)$ is the rank over $\GF(2)$ of its
bipartite adjacency matrix.
The notation $V_x$ and $\mimw_G(\calL)$ is as defined there.

\begin{lemma}\label{lem:mim-rank}
For every $A\subseteq V(G)$, $\mim_G(A)\leq\rho_G(A)$. Hence a
rank-decomposition of width $q$ yields a rooted layout of mim-width at
most $q$ in polynomial time.
\end{lemma}

\begin{proof}
An induced matching of size $a$ across the cut gives an $a\times a$
identity submatrix of the cut adjacency matrix, after ordering its rows
and columns by the matched pairs. Thus its size is at most the matrix
rank. To root a rank-decomposition on at least two leaves, first suppress
any degree-two vertices, subdivide an edge, and make the new vertex the
root. This preserves the nontrivial cuts. The one-vertex case is immediate.
\end{proof}

\subparagraph*{Distance neighborhood logic.}
For integers $r,d\geq1$ and $U\subseteq V(G)$, define
\[
 N_d^r(U)=
 \bigl\{v\in V(G):
   |\{u\in U:1\leq\dist_G(u,v)\leq r\}|\geq d
 \bigr\}.
\]
The logic used below permits set variables, unary vertex predicates,
Boolean set operations, comparisons between set terms, cardinality
constraints such as $|X|=a$, and the operators $N_d^r$. It also permits
$\con(X)$, asserting that $G[X]$ is connected. This is a fragment of
distance neighborhood logic with acyclicity and connectivity constraints
($\ACDN$); we do not use its acyclicity predicate.

For a formula $\varphi$, let $d(\varphi)$ be the largest neighborhood
counting threshold occurring in it, with a minimum value of two. In the
formula-size convention below, numerical values in size measurements and
neighborhood operators each count as one symbol.

\begin{theorem}[Mim-width model checking {\cite[Theorem~1.2]{BergougnouxDreierJaffke2023}}]
\label{thm:meta}
Given an $\ACDN$ formula $\varphi(X_1,\ldots,X_p)$, an $n$-vertex
vertex-colored graph, and a rooted layout of mim-width $w\geq1$, one can
find sets satisfying $\varphi$, or correctly conclude that none exist, in
\[
 n^{O(d(\varphi)w|\varphi|^2)}
\]
time.
\end{theorem}

The important point is that the \emph{numerical radius} does not appear
in this exponent. In particular, a fixed formula schema with fixed
counting thresholds and a constant number of cardinality measurements
bounded by $n$ is handled in $n^{O(w)}$ time, even when its radii depend on
the input. This numerical convention is part of the cited theorem, not an
assumption that $R$ is constant.

\subsubsection{A simpler certificate for a shortest path}

\begin{lemma}[Connected-set certificate]\label{lem:cliquewidth-xp-certificate}
Let $s,t\in V(G)$ and $D=\dist_G(s,t)$. For $X\subseteq V(G)$, the
following are equivalent:
\begin{enumerate}
\setlength{\itemsep}{0pt}
\setlength{\parskip}{0pt}
 \item[\textup{(i)}] $G[X]$ is a geodesic $s$--$t$ path;
 \item[\textup{(ii)}] $s,t\in X$, $|X|=D+1$, and $G[X]$ is connected.
\end{enumerate}
\end{lemma}

\begin{proof}
The forward implication is immediate. Conversely, let $Q$ be a simple
$s$--$t$ path in the connected graph $G[X]$. Then
\[
 D\leq |E(Q)|=|V(Q)|-1\leq |X|-1=D.
\]
All inequalities are equalities, so $Q$ is a shortest $s$--$t$ path and
$V(Q)=X$. Any additional edge of $G[X]$ would join two nonconsecutive
vertices of $Q$ and yield a shorter $s$--$t$ path. Thus $G[X]=Q$.
When $s=t$, the condition $|X|=D+1=1$ gives the same conclusion directly.
\end{proof}

In particular, no acyclicity, degree, or explicit path predicate is needed.

\subsubsection{The algorithm on a supplied mim-width layout}

For fixed $s,t$, give the graph two unary predicates
$\mathbf S=\{s\}$ and $\mathbf T=\{t\}$; when $s=t$, both predicates
hold at that vertex. These predicates do not change the graph or the
layout width. With $D=\dist_G(s,t)$, define
\[
 \mathsf{Dom}_R(X)=
 \begin{cases}
  \overline X\subseteq N_1^R(X),&R\geq1,\\
  \overline X=\varnothing,&R=0,
 \end{cases}
\]
and use the formula
\begin{equation}\label{eq:cliquewidth-formula}
 \Phi_{R,D}(X):=
 (\mathbf S\cup\mathbf T\subseteq X)
 \ \wedge\ |X|=D+1
 \ \wedge\ \con(X)
 \ \wedge\ \mathsf{Dom}_R(X).
\end{equation}
By \cref{lem:cliquewidth-xp-certificate}, the first three conjuncts say exactly that
$G[X]$ is a geodesic $s$--$t$ path. The last conjunct says exactly that
$\ecc_G(X)\leq R$. Conversely, the vertex set of a geodesic
$s$--$t$ path of eccentricity at most $R$ satisfies all four conjuncts.
Thus \eqref{eq:cliquewidth-formula} is an exact certificate for the fixed-endpoint
instance.

\begin{theorem}\label{thm:mim-algorithm}
Given a connected $n$-vertex graph $G$, an integer $R\geq0$, and a rooted
layout of mim-width $w$, \ESP{} can be solved, with a witnessing path
returned, in $n^{O(\max\{1,w\})}$ time. For $n\geq2$, the bound is
$n^{O(w)}$.
\end{theorem}

\begin{proof}
Handle $n=1$ directly. If $n\geq2$, connectivity implies $w\geq1$: every
singleton leaf cut has an edge. Compute all-pairs distances by
breadth-first search in $O(n(n+m))$ time. For every ordered pair
$(s,t)\in V(G)^2$, compute $D=\dist_G(s,t)$, add the two singleton
predicates, and apply \cref{thm:meta} to \eqref{eq:cliquewidth-formula} on the supplied
layout. Accept if any call succeeds, and reject otherwise.

The correctness of the formula proves soundness. For completeness, every
feasible geodesic path has an endpoint pair considered by the algorithm.
A returned set $X$ induces that path, so its vertex order is recovered by
traversing $G[X]$ from $s$ (or simply returning $s$ when $X=\{s\}$).

The formula has constant length, uses only the neighborhood threshold
one, and has a single cardinality constraint of value $D+1\leq n$.
Consequently $d(\Phi_{R,D})=2$ by convention, and each call takes
$n^{O(w)}$ time. The $n^2$ calls and the polynomial preprocessing are
absorbed into the same bound.
\end{proof}

If one requires distinct path endpoints, simply omit pairs $(s,s)$.

\subsubsection{From clique-width to the required layout}

Recall the definition of a $k$-expression from
Subsection~\ref{sec:graph-parameters}. Its syntax tree has vertex-creation
leaves, binary disjoint-union nodes, and unary join and relabelling nodes.

We use the standard conversion underlying the inequality
$\rw(G)\leq\cw(G)$; see Oum and Seymour
~\cite[Proposition~6.3]{OumSeymour2006}.

\begin{lemma}\label{lem:cw-to-layout}
From a clique-width $k$-expression $\tau$ for $G$, one can construct in
polynomial time a rooted layout $\calL_\tau$ whose cuts have cut-rank at
most $k$ and hence mim-width at most $k$. In particular, $\rw(G)\leq\cw(G)$.
\end{lemma}

\begin{proof}
Suppress the unary nodes of the syntax tree and match each vertex with
its creation leaf. This gives a rooted layout. Fix a node cut
$(A,\overline A)$, and consider the labels just after evaluating the
corresponding subexpression. Two vertices of $A$ having the same label
then keep the same label through all subsequent operations. Every edge
from $A$ to $\overline A$ is added after the two sides first meet at an
ancestor disjoint-union node, and later joins treat these two vertices
identically. Therefore they have identical final neighborhoods in
$\overline A$.

The cut adjacency matrix thus has at most $k$ distinct rows, so its rank
is at most $k$. By \cref{lem:mim-rank}, the cut's induced-matching number
is also at most $k$. This holds at every layout node. Forgetting and
suppressing the degree-two root gives a rank-decomposition of width at
most $k$ when $n\geq2$; the one-vertex case is immediate.
\end{proof}

\begin{proof}[Proof of \cref{thm:cliquewidth-xp}\textup{(a)}]
Construct $\calL_\tau$ using \cref{lem:cw-to-layout}, and apply
\cref{thm:mim-algorithm}. Processing the expression takes
$|\tau|^{O(1)}$ time, and the resulting layout has mim-width at most $k$.
\end{proof}

\subsubsection{Removing the expression and obtaining an optimal path}

We use Oum's algorithm \cite[Theorem~5.8]{Oum2008}: for each positive
integer $h$, it either produces a rank-decomposition of width at most
$3h-1$ or correctly certifies $\rw(G)>h$, in $g(h)n^3$ time for a
computable function $g$.

\begin{proof}[Proof of \cref{thm:cliquewidth-xp}\textup{(b)}]
Handle $n=1$ directly, and put $q=\rw(G)\geq1$ otherwise. Run Oum's
algorithm for $h=1,2,\ldots$ until it returns a decomposition. It must
succeed by $h=q$, so the returned decomposition has width at most
$3q-1$, and the total preprocessing time is at most
\[
 \left(\sum_{h=1}^{q}g(h)\right)n^3.
\]
By \cref{lem:mim-rank}, root it to obtain a layout of mim-width at most
$3q-1$. Then \cref{thm:mim-algorithm} gives an $n^{O(q)}$-time algorithm.
Finally, $q\leq k=\cw(G)$ by \cref{lem:cw-to-layout}, so the total time
is $f(k)n^{O(k)}$ for a computable $f$. Neither $k$ nor $q$ needs to be
known by the algorithm.
\end{proof}

\begin{corollary}\label{cor:rw}
Without a supplied decomposition, \ESP{} is solvable in
$f(q)n^{O(\max\{1,q\})}$ time, where $q=\rw(G)$. In particular, it is
in XP parameterized by rank-width.
\end{corollary}

\begin{proof}
Use the preceding algorithm and bound its running time directly in terms
of $q$, handling the one-vertex case separately.
\end{proof}

\begin{corollary}\label{cor:optimization}
A minimum-eccentricity shortest path can be found within the bounds of
\cref{thm:cliquewidth-xp}, up to an additional $O(\log n)$ factor for $n\geq2$.
\end{corollary}

\begin{proof}
The optimum lies in $\{0,\ldots,n-1\}$, and feasibility is monotone in
$R$. Binary search finds the smallest feasible radius in $O(\log n)$
decision calls. Recover a witness with one final call at this radius.
The same layout and all-pairs distances can be reused throughout.
\end{proof}

\begin{remark}
The mim-width step remains essential: the clique-width branch of the
cited model-checking theorem depends exponentially on the numerical
radius. Only the mim-width bound used here is independent of that radius.
\end{remark}

\section{W-hardness}
\label{sec:w-hardness}

\subsection{Clique cover number}
\label{sec:mesp-clique-cover}

A set $D\subseteq V(G)$ is a dominating set if $N_G[D]=V(G)$. The \ds{}
problem asks whether a graph has a dominating set of size at most $k$; it is
$\Wtwo$-complete parameterized by $k$~\cite{DowneyFellows1995}.

We show that a small vertex clique cover does not make \ESP{}
fixed-parameter tractable unless $\FPT=\Wtwo$. Recall that $\cc(G)$ denotes
the minimum number of cliques whose vertex sets partition $V(G)$; see
Subsection~\ref{sec:graph-parameters}.

Our reduction uses the choice-layer and distance-test structure of
Dragan and Leitert~\cite[Theorem~13]{DraganLeitert2017},
with clique layers and columns arranged to bound the vertex
clique-cover number.

\begin{restatable}{theorem}{thmcliquecoverhardness}
\label{thm:maincliquecover}
\ESP{} is $\Wtwo$-hard parameterized by $\cc(G)$. This remains true when a
vertex clique cover is supplied and the parameter is the size of the supplied
cover.
\end{restatable}

\subsubsection{Construction}

We give a parameterized reduction from \ds{}, which is $\Wtwo$-complete
parameterized by the solution size~\cite{DowneyFellows1995}. Let $(F,k)$ be
an instance of \ds{}, put $n=|V(F)|$, and assume that $1\leq k\leq n$, as
the remaining cases are trivial. Set $R=k$.

For every $i\in[k]$ and $0\leq h\leq R-1$, create the clique
$
 C_{i,h}=\{x_{i,h}:x\in V(F)\}.
$
For every $x\in V(F)$ and $i\in[k]$, add the edges of the column
\[
 x_{i,0}-x_{i,1}-\cdots-x_{i,R-1}.
\]
Make $C_{i,0}$ complete to $C_{i+1,0}$ for every $i\in[k-1]$. Add two
vertices $s$ and $t$, make $s$ complete to $C_{1,0}$, and make $t$ complete
to $C_{k,0}$. We call the subgraph induced by $s,t$ and the cliques
$C_{i,0}$ the \emph{backbone}. Attach the two pendant paths
\[
 s=s_0-s_1-\cdots-s_R,
 \qquad
 t=t_0-t_1-\cdots-t_R.
\]
When $R=1$, each column consists only of its backbone vertex.

Create the clique $W=\{q_x:x\in V(F)\}$ and add the edges
\[
 q_xy_{i,R-1}\in E(H)
 \quad\Longleftrightarrow\quad
 y\in N_F[x],
 \qquad x,y\in V(F),\ i\in[k].
\]
No edges other than those specified above are present. Thus $q_x$ is
adjacent to the top of every column corresponding to a vertex that dominates
$x$. Since $x\in N_F[x]$, every vertex of $W$ is connected to the backbone,
and hence $H$ is connected.

Figure~\ref{fig:cc-overview} summarizes the construction. A box represents a
clique, a thick horizontal segment represents a complete join, and a thin
vertical segment represents the matching edges belonging to the columns.

\begin{figure}[t]
\centering
\begin{tikzpicture}[
  x=0.83cm,
  y=0.78cm,
  font=\small,
  clique/.style={draw,rounded corners=2pt,fill=blue!7,
    minimum width=1.28cm,minimum height=0.62cm,align=center},
  witness/.style={draw,rounded corners=7pt,fill=orange!12,
    minimum width=1.55cm,minimum height=0.7cm,align=center},
  terminal/.style={draw,circle,fill=black!4,inner sep=1.6pt},
  incidence/.style={densely dashed,gray!80},
  column/.style={gray!75},
  backbone/.style={very thick}
]
  \node[terminal] (sr) at (0,0) {$s_R$};
  \node (sd) at (0.75,0) {$\cdots$};
  \node[terminal] (s) at (1.5,0) {$s$};
  \node[clique] (c10) at (3.0,0) {$C_{1,0}$};
  \node[clique] (c20) at (4.8,0) {$C_{2,0}$};
  \node (bd) at (6.15,0) {$\cdots$};
  \node[clique] (ck0) at (7.5,0) {$C_{k,0}$};
  \node[terminal] (t) at (9.0,0) {$t$};
  \node (td) at (9.75,0) {$\cdots$};
  \node[terminal] (tr) at (10.5,0) {$t_R$};

  \draw (sr) -- (sd) -- (s);
  \draw[backbone] (s) -- (c10);
  \draw[backbone] (c10) --
    node[above=8pt,font=\scriptsize] {complete join} (c20);
  \draw[backbone] (c20) -- (bd) -- (ck0);
  \draw[backbone] (ck0) -- (t);
  \draw (t) -- (td) -- (tr);
  \node[above=2pt of sd,font=\scriptsize] {$R$ edges};
  \node[above=2pt of td,font=\scriptsize] {$R$ edges};
  \node[below=5.8pt of bd,font=\scriptsize] {backbone};

  \node[clique] (c1top) at (3.0,2.7) {$C_{1,R-1}$};
  \node (c1d) at (3.0,1.35) {$\vdots$};
  \node[clique] (c2top) at (4.8,2.7) {$C_{2,R-1}$};
  \node (c2d) at (4.8,1.35) {$\vdots$};
  \node[clique] (cktop) at (7.5,2.7) {$C_{k,R-1}$};
  \node (ckd) at (7.5,1.35) {$\vdots$};

  \draw[column] (c10) -- (c1d) -- (c1top);
  \draw[column] (c20) -- (c2d) -- (c2top);
  \draw[column] (ck0) -- (ckd) -- (cktop);
  \node[rotate=90,font=\scriptsize] at (2.55,1.35) {columns};

  \node[witness] (w) at (5.25,4.15) {$W=\{q_x:x\in V(F)\}$};
  \draw[incidence] (w.south west) -- (c1top.north);
  \draw[incidence] (w.south) -- (c2top.north);
  \draw[incidence] (w.south east) -- (cktop.north);
  \node[anchor=west,align=left,font=\scriptsize] at (7.55,4.05)
    {$q_xy_{i,R-1}\in E(H)$ iff\\[-1pt] $y\in N_F[x]$};

  \draw[decorate,decoration={brace,mirror,amplitude=4pt}]
    (2.2,-0.58) -- (8.25,-0.58)
    node[midway,below=5pt,font=\scriptsize]
    {$P_{\boldsymbol v}$ selects $(v_i)_{i,0}\in C_{i,0}$ for every $i\in[k]$};
\end{tikzpicture}
\caption{Overview of the reduction from \ds{}. Each $C_{i,h}$ and $W$ is a
clique. Consecutive bottom cliques are completely joined, while the vertical
edges form the columns. The path $P_{\boldsymbol v}$ stays in the backbone
and selects one vertex from each clique $C_{i,0}$.}
\label{fig:cc-overview}
\end{figure}

The graph has
\[
 (kR+1)n+2R+2=(k^2+1)n+2k+2
\]
vertices and can be constructed in polynomial time.

\subsubsection{Bounding the clique cover number}

The $kR$ cliques $C_{i,h}$, the clique $W$, and the $2R+2$ singleton cliques
containing the vertices of the two pendant paths, including $s$ and $t$, form
an explicit vertex clique cover of size
\[
 kR+1+(2R+2)=k^2+2k+3.
\]
Consequently, $\cc(H)\leq k^2+2k+3$. Notice that the supplied cover need not
be minimum.

\subsubsection{Correctness}

\begin{proof}[Proof of Theorem~\ref{thm:maincliquecover}]
For a tuple $\boldsymbol v=(v_1,\ldots,v_k)\in V(F)^k$, where repetitions
are allowed, let
\[
 P_{\boldsymbol v}
 =(s,(v_1)_{1,0},\ldots,(v_k)_{k,0},t),
 \qquad
 D_{\boldsymbol v}=\{v_1,\ldots,v_k\}.
\]
Suppose first that $F$ has a dominating set of size at most $k$. List its
vertices in a $k$-tuple $\boldsymbol v$, repeating one vertex if necessary.
We first verify that the corresponding path $P_{\boldsymbol v}$ is a
geodesic.

\begin{claim}\label{lem:normal}
The shortest $s$--$t$ paths are exactly the paths $P_{\boldsymbol v}$.
\end{claim}

\begin{proof}
A shortest $s$--$t$ path cannot enter either pendant path. Consider an
excursion outside the backbone between two consecutive visits to it. If its
endpoints belong to the same clique $C_{i,0}$, the excursion can be replaced
by at most one edge. If they belong to distinct cliques $C_{i,0}$ and
$C_{j,0}$, then the excursion must reach $W$ and return, since columns with
different first indices have no other connection outside the backbone. This
uses at least $2R$ edges, whereas the backbone connects the two endpoints in
$|i-j|\leq k-1<2R$ edges. Hence every shortest $s$--$t$ path is contained in
the backbone. Its minimum length is $k+1$, and the paths attaining this
length are precisely the paths $P_{\boldsymbol v}$.
\end{proof}

The following claim relates the witness vertices to domination in $F$.

\begin{claim}\label{lem:witness}
Every vertex outside $W$ is within distance $R$ of $P_{\boldsymbol v}$. For
every $x\in V(F)$,
\refstepcounter{equation}\label{eq:distance}
\[
 \dist_H(q_x,V(P_{\boldsymbol v}))=
 \begin{cases}
 R, & N_F[x]\cap D_{\boldsymbol v}\neq\varnothing,\\
 R+1, & N_F[x]\cap D_{\boldsymbol v}=\varnothing
 \end{cases}
 \qquad\text{(\theequation)}.
\]
\end{claim}

\begin{proof}
Every pendant-path vertex is within distance $R$ of $s$ or $t$. For a column
vertex $y_{i,h}$, follow its column to $y_{i,0}$ and, if necessary, use one
edge of $C_{i,0}$ to reach $(v_i)_{i,0}$. This takes at most
$h+1\leq R$ edges. Hence every vertex outside $W$ is within distance $R$ of
$P_{\boldsymbol v}$.

Assign level $h$ to every vertex $y_{i,h}$, level $R$ to every witness
vertex, and level zero to all pendant-path vertices, including $s$ and $t$.
Every edge changes the level by at most one. Since every vertex of
$P_{\boldsymbol v}$ has level zero, a walk of length at most $R$ from $q_x$
to $P_{\boldsymbol v}$ must have exactly $R$ edges and decrease the level at
every step. Its first edge goes to some $y_{i,R-1}$ with $y\in N_F[x]$, and
the remaining edges are forced down the same column to $y_{i,0}$. In
particular, such a walk cannot use an edge inside $W$, inside a level clique,
or inside the backbone.

It follows that $\dist_H(q_x,V(P_{\boldsymbol v}))\leq R$ precisely when
$v_i\in N_F[x]$ for some $i\in[k]$, and in this case the distance is exactly
$R$. If no such index exists, then $q_x$ reaches $x_{1,0}$ in $R$ steps,
because $x\in N_F[x]$, and one additional edge of $C_{1,0}$ reaches
$(v_1)_{1,0}$. This proves~\eqref{eq:distance}.
\end{proof}

In particular, $\ecc_H(P_{\boldsymbol v})\leq R$ if and only if
$D_{\boldsymbol v}$ dominates $F$.
Since $D_{\boldsymbol v}$ dominates $F$, Claim~\ref{lem:witness} gives
$\ecc_H(P_{\boldsymbol v})\leq R$. Together with Claim~\ref{lem:normal}, this
shows that $P_{\boldsymbol v}$ is a feasible solution. Hence $(H,R)$ is a
yes-instance of \ESP{}.

For the reverse direction, suppose that $(H,R)$ is a yes-instance of
\ESP{}, and let $Q$ be a geodesic with $\ecc_H(Q)\leq R$. We next extract
from $Q$ a path of the form $P_{\boldsymbol v}$ without increasing its
eccentricity.

\begin{claim}\label{lem:cc-normalization}
There is a tuple $\boldsymbol v\in V(F)^k$ such that
$\ecc_H(P_{\boldsymbol v})\leq R$.
\end{claim}

\begin{proof}
The radius-$R$ ball around $s_R$ is exactly
$\{s_0,\ldots,s_R\}$, and the analogous statement holds for $t_R$.
Therefore $Q$ meets both pendant paths. Since they attach to the rest of $H$
only at $s$ and $t$, respectively, the path $Q$ contains both $s$ and $t$.
Its $s$--$t$ subpath is shortest and, by Claim~\ref{lem:normal}, is some
$P_{\boldsymbol v}$.

It remains to consider possible extensions of $P_{\boldsymbol v}$ in $Q$.
Every geodesic is induced. An extension at $s$ cannot return to
$(v_1)_{1,0}$, since $Q$ is simple, while every other vertex of $C_{1,0}$
is adjacent to $(v_1)_{1,0}$ and would create a chord. Thus any extension at
$s$ lies on the $s$-pendant path. Symmetrically, any extension at $t$ lies
on the $t$-pendant path.

For every vertex outside the two pendant paths, reaching such an extension
requires passing through $s$ or $t$, both of which already belong to
$P_{\boldsymbol v}$. Hence its distance to $P_{\boldsymbol v}$ equals its
distance to $Q$ and is at most $R$. Every pendant-path vertex is within
distance $R$ of $s$ or $t$. Consequently,
$\ecc_H(P_{\boldsymbol v})\leq R$.
\end{proof}

By Claim~\ref{lem:cc-normalization}, some path $P_{\boldsymbol v}$ has
eccentricity at most $R$. Claim~\ref{lem:witness} then shows that
$D_{\boldsymbol v}$ dominates $F$ and has size at most $k$. Therefore
$(F,k)$ is a yes-instance of \ds{} if and only if $(H,R)$ is a yes-instance
of \ESP{}. Together with the parameter bound, this completes the
parameterized reduction.
\end{proof}

\subsection{Cutwidth}

Recall that $\cutwidth(G,\pi)$ is the maximum number of edges crossing
a prefix cut of the ordering $\pi$, and $\cutwidth(G)$ is its minimum over
all orderings; see Subsection~\ref{sec:graph-parameters}.

The parameter is $\cutwidth(G)$ alone; $q$ is part of the input.

An instance of \CSP\ consists of a constraint graph $F$ on a variable
set $\vars$, a finite domain $D_v$ for each variable $v$, and a relation
$R_{uv}\subseteq D_u\times D_v$ for each edge $uv\in E(F)$.
The domains and relations are explicitly listed. A satisfying assignment
$\rho$ chooses $\rho(v)\in D_v$ and satisfies
$(\rho(u),\rho(v))\in R_{uv}$ for every constraint.
Repeated constraints on a pair can be intersected. Recall that $\vc(F)$
denotes the vertex cover number of the constraint graph $F$.

\subparagraph*{Example.}
Let $\vars=\{x,y,z\}$, let $D_x=D_y=D_z=\{0,1\}$, and let the
constraint graph $F$ be the path with edges $xy$ and $yz$. Set
\[
 R_{xy}=\{(0,1),(1,0)\},
 \qquad R_{yz}=\{(0,0),(1,1)\}.
\]
Thus the constraints require $\rho(x)\neq\rho(y)$ and
$\rho(y)=\rho(z)$. The assignment
$(\rho(x),\rho(y),\rho(z))=(0,1,1)$ satisfies both constraints,
whereas $(0,1,0)$ violates the constraint on $yz$.
Moreover, $\{y\}$ is a minimum vertex cover of $F$, so $\vc(F)=1$.
Here ``binary'' refers to each constraint involving two variables;
the domains need not have size two.

\begin{proposition}[Bodlaender, Groenland, and Pilipczuk
{\cite[Theorem~1]{BGP23}}]\label{prop:source}
\CSP\ parameterized by $\vc(F)$ is $\Wthree$-complete.
\end{proposition}

\begin{theorem}\label{thm:cw-whard}
\ESP\ is $\Wthree$-hard parameterized by cutwidth, even on connected
graphs of maximum degree at most five. More precisely, a \CSP\ instance
with constraint graph $F$ can be transformed in polynomial time into
an equivalent instance $(G,q)$ together with an ordering $\pi$ satisfying
\[
 \Delta(G)\leq5,
 \qquad \cutwidth(G,\pi)\leq4\vc(F)+7.
\]
If a vertex cover of size $k$ is supplied, the cutwidth bound is $2k+7$.
\end{theorem}

\subsubsection{Construction}

If $\vars=\varnothing$, return the yes-instance $(K_1,0)$.
If some domain is empty, return the no-instance $(K_3,0)$; a shortest
path in a triangle contains at most two vertices. Both outputs satisfy
the claimed bounds. Hence assume that $\vars$ and all domains are nonempty.

Let $X=\{x_1,\ldots,x_k\}$ be a vertex cover of $F$, and put
$Y=\vars\setminus X=\{y_1,\ldots,y_m\}$. For now, $X$ is supplied.
For each variable $v\in\vars$, put $N_v=|D_v|-1$ and
relabel its domain as $D_v=\{0,\ldots,N_v\}$.
Set
\[
 S=\sum_{v\in\vars}|D_v|
  =\sum_{v\in\vars}(N_v+1).
\]
A \emph{test} is a tuple $\tau=(u,v,a,b)$, where $uv\in E(F)$,
$a\in D_u$, $b\in D_v$, and $(a,b)\notin R_{uv}$.
It records that $u$ and $v$ cannot simultaneously receive the values
$a$ and $b$. Thus an assignment $\rho$ is satisfying exactly when,
for every test, $\rho(u)\neq a$ or $\rho(v)\neq b$.

Fix an arbitrary linear order on $\vars$. For each constraint edge,
write its endpoints as $u,v$ with $u$ preceding $v$, transposing the
relation if necessary. Introduce tests only in this orientation.

List each test once, first those involving two variables of $X$,
then those involving $y_1$, then $y_2$, and so on.
Since $Y$ is independent in $F$, every test involving a variable
of $Y$ has its other variable in $X$, so each test belongs to
exactly one of these blocks. Fix an arbitrary order within each block.

Let $T$ be the total number of tests. For each variable $v$, consider
the sublist consisting only of tests containing $v$, preserving their
order. Let $t_v$ be the length of this sublist and, for a test $\tau$
containing $v$, let $p_v(\tau)$ be its position in this sublist,
starting from $1$. Thus
\[
 1\leq p_v(\tau)\leq t_v\leq T,
 \qquad \sum_v t_v=2T,
 \qquad T\leq S^2.
\]
The test list is computable in polynomial time. Define
\begin{equation}\label{eq:L}
                         L=2(S+T),\qquad q=2L.
\end{equation}
All vertices introduced below are distinct unless explicitly identified.
Every added path has new internal vertices, disjoint from all other
added paths, and there are no unspecified edges.

\subparagraph*{Ladders and carriers.}
For every variable $v$, create two rail paths
\[
 A_{v,0}\cdots A_{v,N_v},\qquad B_{v,0}\cdots B_{v,N_v}.
\]
For each $a\in D_v$, join $A_{v,a}$ to $B_{v,a}$ by a path of length
$N_v+1$, called \emph{rung $a$}. For $N_v=0$, the ladder is a single edge.
Arrange the ladders in the order $x_1,\ldots,x_k,y_1,\ldots,y_m$,
adding the edge $B_{u,N_u}A_{v,0}$ for consecutive variables $u,v$.
Call the resulting graph $C$, the \emph{core}; let $s$ be its first
top-rail vertex and $t$ its last bottom-rail vertex.

For each variable $v$, designate the two \emph{ports}
\[
 U_v=A_{v,N_v},\qquad V_v=B_{v,0}.
\]
Attach a path of length $t_v$ at each port, and denote these
\emph{carrier paths} by
\[
 Q_v^U=U_v^0U_v^1\cdots U_v^{t_v},
 \qquad
 Q_v^V=V_v^0V_v^1\cdots V_v^{t_v},
\]
where $U_v^0=U_v$ and $V_v^0=V_v$.
All vertices with positive superscripts are new.
If $t_v=0$, each carrier path consists only of its port,
so no vertices or edges are added.
Let $H$ consist of the core and all carrier paths.

\subparagraph*{Tests and pendants.}
For each test $\tau=(u,v,a,b)$, create a center $c_\tau$.
With $p=p_u(\tau)$ and $r=p_v(\tau)$, add four \emph{connection paths}:
\[
\begin{array}{c|c}
 \text{Endpoints}&\text{Length}\\\hline
 c_\tau,\ U_u^p&L-N_u+a+1-p\\
 c_\tau,\ V_u^p&L-a+1-p\\
 c_\tau,\ U_v^r&L-N_v+b+1-r\\
 c_\tau,\ V_v^r&L-b+1-r
\end{array}
\]
Attach a pendant path of length $L$ at $c_\tau$, ending at $z_\tau$.
Finally, attach pendant paths of length $q+1$ at $s$ and $t$, ending
at $s^*$ and $t^*$, respectively. Call these the \emph{terminal pendants}.
This completes $G$.

Since $N_v\leq S-1$ and every attachment position is between $1$ and
$T$, each connection length $\lambda$ satisfies
\begin{equation}\label{eq:long}
 \frac L2<L-S-T+2\leq\lambda\leq L.
\end{equation}
In particular, every specified length is positive.

\subsubsection{Correctness}

\subparagraph*{Distances within $H$.}
The diameter of ladder $v$ is at most $2N_v+1$. Indeed, for any two
vertices, the sums of their rung distances to the top and bottom rails
add to $2(N_v+1)$. One of these sums is at most $N_v+1$, and the
corresponding rail segment has length at most $N_v$. These segments
give a walk of length at most $2N_v+1$ between the two vertices.
Joining ladders in sequence, and using that each carrier has length
at most $T$, gives
\begin{equation}\label{eq:diam}
 \diam(C)\leq\sum_v(2N_v+1)+|\vars|-1=2S-1,
 \qquad \diam(H)\leq2S-1+2T=L-1.
\end{equation}
Consequently, every shortest path in $G$ whose endpoints lie in $H$
stays in $H$: an excursion outside $H$ cannot use a non-root pendant
vertex and must traverse two complete connection paths, giving a subpath of length greater than
$L$ by~\eqref{eq:long}, whereas its endpoints are at distance at most
$L-1$ in $H$.

For an assignment $\rho$, define $P_\rho$ to follow the first terminal
pendant from $s^*$ to $s$, cross each ladder along the top rail to rung
$\rho(v)$, traverse that rung, continue along the bottom rail, and
finally follow the other terminal pendant from $t$ to $t^*$.

\begin{claim}\label{clm:paths}
Every $P_\rho$ is shortest, and
\begin{equation}\label{eq:ports}
 d_G(U_v,P_\rho)=N_v-\rho(v),
 \qquad d_G(V_v,P_\rho)=\rho(v).
\end{equation}
Moreover, every shortest path of eccentricity at most $q$ can be
extended to some $P_\rho$ without increasing its eccentricity.
\end{claim}
\begin{proof}
A path across ladder $v$, from $A_{v,0}$ to $B_{v,N_v}$, traverses at
least one entire rung and at least $N_v$ rail edges: rungs preserve
the column index, whereas rail edges change it by one. Its minimum
length is therefore $2N_v+1$, attained exactly by choosing one rung
and moving monotonically along the rails. The inter-ladder edges
are bridges in $C$, so the shortest $s$--$t$ paths in $C$ correspond
exactly to assignments and have length $2S-1$.
Shortest $s$--$t$ paths in $G$ stay in $H$ by the preceding argument,
and cannot traverse a carrier edge, since carriers are pendant in $H$. Hence they are
precisely these core paths. It follows that every $P_\rho$ is shortest.

For the port distances it suffices to consider routes in $C$: carriers
give no shortcut in $H$, shortest paths between vertices of $H$ stay
in $H$, and either terminal pendant is reached through $s$ or $t$,
already on $P_\rho$.
The top rail reaches $P_\rho$ from $U_v$ in $N_v-\rho(v)$ steps.
Traversing another rung takes $N_v+1$ steps, and a route in $C$ leaving
the ladder first meets $A_{v,0}$ or $B_{v,N_v}$, both on $P_\rho$.
The chosen rung cannot be entered before meeting an endpoint already
on $P_\rho$. Thus the top-rail route is shortest; the bottom-rail
argument is symmetric. This proves~\eqref{eq:ports}.

Now let $Q$ be shortest with $\ecc_G(Q)\leq q$.
It must meet the non-root part of each terminal pendant; otherwise
its distance from the corresponding leaf is at least $q+1$.
A simple path meeting both pendants has an endpoint in each, so $Q$
contains a shortest $s$--$t$ subpath. This subpath selects one rung in
every ladder. Extending $Q$ to both terminal leaves therefore gives
the corresponding shortest path $P_\rho$, and adding vertices cannot
increase eccentricity.
\end{proof}

\begin{claim}\label{clm:tests}
For every assignment $\rho$, we have $\ecc_G(P_\rho)\leq q$ if and
only if $\rho$ satisfies the \CSP\ instance.
\end{claim}
\begin{proof}
Fix a test $\tau=(u,v,a,b)$. A shortest route from $c_\tau$ to
$P_\rho$ first reaches $P_\rho$ in the core, since either terminal
pendant is entered through a vertex of $P_\rho$. It starts with one
complete connection path. Its remaining subpath has endpoints in
$H$, so stays in $H$; it follows the carrier to its port and then
reaches $P_\rho$ inside $C$.
For $(w,h)\in\{(u,a),(v,b)\}$ and $p=p_w(\tau)$, the two possible
route lengths through its carriers are, by~\eqref{eq:ports},
\begin{align*}
 (L-N_w+h+1-p)+p+(N_w-\rho(w))&=L+1+h-\rho(w),\\
 (L-h+1-p)+p+\rho(w)&=L+1+\rho(w)-h.
\end{align*}
Taking the minimum of these four values and adding the test pendant's
length $L$ yields
\begin{equation}\label{eq:test}
 d_G(z_\tau,P_\rho)
 =q+1-\max\{|\rho(u)-a|,|\rho(v)-b|\}.
\end{equation}
This distance exceeds $q$ exactly when
$(\rho(u),\rho(v))=(a,b)$. Thus all test leaves are within distance
$q$ exactly when $\rho$ avoids every forbidden pair.

It remains to check other vertices. Every vertex of $H$ is within
$L-1$ of $P_\rho$ by~\eqref{eq:diam}. Every connection-path vertex,
including its center, is within $L$ of $H$, so its distance from
$P_\rho$ is at most $2L-1<q$. The terminal pendants belong to
$P_\rho$. Finally, on a test pendant the distance from $P_\rho$
increases with the distance from its center, and is largest at its
leaf. Hence covering all test leaves covers all of $G$.
\end{proof}

Claims~\ref{clm:paths} and~\ref{clm:tests} prove that the source instance
is satisfiable if and only if $(G,q)$ is a yes-instance of \ESP.

\subsubsection{Cutwidth and completion of the proof}

We construct an ordering $\pi$. First list the terminal pendant from
$s^*$ towards $s$, omitting $s$. To \emph{process a ladder} $v$, list,
for $a=0,\ldots,N_v$,
\[
 A_{v,a},\quad\text{the internal vertices of rung $a$ from top to bottom},
 \quad B_{v,a}.
\]
To \emph{process a test} $\tau=(u,v,a,b)$, first list its four carrier
attachments
\[
U_u^{p_u(\tau)},\quad V_u^{p_u(\tau)},\quad
U_v^{p_v(\tau)},\quad V_v^{p_v(\tau)}.
\]
Next list the internal vertices of its four connection paths, one
path at a time from its carrier attachment towards $c_\tau$; then list
$c_\tau$ and its pendant vertices in order towards $z_\tau$.

Process the ladders of $x_1,\ldots,x_k$, followed by all tests internal
to $X$. For $j=1,\ldots,m$, process ladder $y_j$ and then all tests
containing $y_j$. Within each block use the previously fixed test order.
Finally, list the terminal pendant from the neighbor of $t$ to $t^*$.
Every non-root carrier vertex belongs to exactly one test, and these
tests are processed in increasing order of their positions on each
carrier. Thus every vertex is listed exactly once.

\begin{claim}\label{clm:cutwidth}
The ordering $\pi$ satisfies $\cutwidth(G,\pi)\leq2k+7$.
\end{claim}
\begin{proof}
Fix a cut in the ordering, and call the vertices before the cut
\emph{listed}. The root of each carrier is listed when its ladder
is processed, before any of its non-root vertices. The non-root
vertices are subsequently listed in increasing order of their
positions on the carrier, because tests are processed in the order
used to define $p_v(\tau)$. Thus the listed vertices of each carrier
form an initial segment, possibly empty or the whole path.
Consequently, at most one edge of that carrier crosses the cut:
the edge joining its last listed vertex to its first unlisted
vertex, when both exist.

The two carriers of each variable in $X$ therefore contribute at
most $2k$ crossing edges altogether. For each $y_j\in Y$, its block
consists of its ladder followed by all tests containing $y_j$.
Before this block, none of its carrier vertices, including the
roots, is listed; after the block, all of them are listed.
Hence a variable of $Y$ can contribute crossing carrier edges only
if the cut lies within its block. At most one such block contains
the cut, and its two carriers contribute at most two crossing edges.
Therefore, at most $2k+2$ carrier edges cross the cut.

During a ladder, no test is partially processed. At most one edge of
each rail and one edge of the current rung cross. The incoming
inter-ladder edge stops crossing when the first ladder vertex is
listed, and the outgoing edge starts crossing only when its last
vertex is listed; the terminal-pendant edges behave likewise.
Thus at most three non-carrier edges cross.

During a test, a whole initial segment of the ladder sequence has
been listed, so at most one core or terminal-pendant edge crosses.
Each of the test's four connection paths contributes at most one
crossing edge. When its center is listed, all four paths are complete,
and at most one edge of its pendant crosses. Earlier tests are complete;
later tests and their carrier attachments are unlisted. Thus at most
four test edges cross, giving at most $(2k+2)+1+4=2k+7$ edges altogether.
Cuts between operations obey the same bounds. During the initial and
final terminal-pendant operations only one edge crosses.
\end{proof}

\begin{proof}[Completion of the proof of Theorem~\ref{thm:cw-whard}]
The graph is simple and connected by construction. Each test center
has degree five, and all other vertices have degree at most three.
This includes $N_v=0$: the coinciding port and ladder terminal then
has one rung neighbor, at most one carrier neighbor, and one
inter-ladder or terminal-pendant neighbor.
The core has $O(S^2)$ vertices, the carriers add $2\sum_v t_v=4T$,
each test adds $O(L)$, and the terminal pendants add $O(L)$.
Consequently,
\[
 |V(G)|=O\bigl(S^2+(T+1)L\bigr),
\]
which is polynomial in the source input since $T\leq S^2$ and
$L=2(S+T)$. All lengths are realized by subdivisions. Both $G$ and
$\pi$ can therefore be constructed in polynomial time.

If no cover is supplied, let $X$ be the endpoints of a greedily
computed maximal matching $M$ in $F$. Maximality makes $X$ a vertex
cover, while every vertex cover meets each of the pairwise disjoint
edges of $M$. Hence $k=|X|=2|M|\leq2\vc(F)$.
The correctness claims and Claim~\ref{clm:cutwidth} now give a
polynomial-time reduction with
\[
 \cutwidth(G)\leq\cutwidth(G,\pi)\leq2k+7\leq4\vc(F)+7.
\]
Proposition~\ref{prop:source} proves the asserted $\Wthree$-hardness,
including when the ordering is supplied with the output instance.
\end{proof}

\section{Para-NP-hardness}
\label{sec:para-np-hardness}

\subsection{Tree-independence number}
\label{sec:mesp-tree-independence}

Recall that $\ta(G)$ is the minimum, over all tree decompositions, of the
largest independence number of a bag; see
Subsection~\ref{sec:graph-parameters}.

We show that \ESP{} remains NP-hard even when the input graph has
tree-independence number at most two. Moreover, the reduction produces a tree
decomposition together with a partition of every bag into two cliques.

The construction is a variant of the SAT reduction of Dragan and
Leitert~\cite[Theorem~7]{DraganLeitert2017}: we add the
literal-pair edges and make the clause vertices a clique to obtain
a tree decomposition whose bags are unions of two cliques.

\begin{restatable}{theorem}{thmtreeindependencehardness}
\label{thm:hardness}
\ESP{} is NP-hard on connected graphs $G$ with $\ta(G)\leq2$. This remains
true when the input contains a tree decomposition of $G$ and an explicit
partition of every bag into two cliques.
\end{restatable}

\subsubsection{Construction}

We give a polynomial-time reduction from \textsc{3-SAT}. Let $\Phi$ be a
$3$-CNF formula with variables $x_1,\ldots,x_n$ and clauses
$D_1,\ldots,D_m$. We may assume that $n,m\geq1$ and that every clause is
nonempty. Repeated variables within a clause cause no difficulty. Set
$R=n+1$.

Create connector vertices $a_0,\ldots,a_n$. For every $i$ with
$1\leq i\leq n$, create two literal vertices $p_i$ and $q_i$ and add the edges
\[
 a_{i-1}p_i,\quad a_{i-1}q_i,\quad
 a_ip_i,\quad a_iq_i,\quad p_iq_i
 \qquad(1\leq i\leq n).
\]
Let $H$ denote the resulting
variable-chain graph. The vertices $p_i$ and $q_i$ represent the literals
$x_i$ and $\neg x_i$, respectively.

Create a clique $C=\{c_1,\ldots,c_m\}$. For every occurrence of $x_i$ in
$D_j$, add a path of length exactly $R$ between $c_j$ and $p_i$. Similarly,
for every occurrence of $\neg x_i$ in $D_j$, add a path of length exactly
$R$ between $c_j$ and $q_i$. We call these paths the \emph{occurrence
paths}. Their internal vertices are new and pairwise distinct, and every
internal vertex has degree two. 

Finally, attach a pendant path of length $R$ to $a_0$, with other endpoint
$s$, and a pendant path of length $R$ to $a_n$, with other endpoint $t$. No further
edges are added. The output instance is $(G,R)$.

Since every clause is nonempty, the graph $G$ is connected. There are at
most $3m$ occurrence paths, each of length $R=n+1$. Hence $G$ has
$O(nm+n+m)$ vertices and $O(nm+n+m^2)$ edges and can be constructed in
polynomial time.

\subsubsection{Bounding the tree-independence number}

We first construct a tree decomposition of $G-C$ whose bags are cliques. For
$H$, consider the bags
\[
 A_i=\{a_{i-1},p_i,q_i\}
 \qquad\text{and}\qquad
 B_i=\{a_i,p_i,q_i\}
 \qquad(1\leq i\leq n),
\]
arranged in the path
\[
 A_1-B_1-A_2-B_2-\cdots-A_n-B_n.
\]
Every bag induces a triangle and every edge of $H$ is covered. Each literal
vertex occurs in the consecutive bags $A_i,B_i$. For $1\leq i<n$, the
connector $a_i$ occurs in the consecutive bags $B_i,A_{i+1}$, whereas
$a_0$ and $a_n$ each occur in one bag. Thus the running-intersection
condition holds.

After deleting $C$, every occurrence path is a pendant path attached to its
literal endpoint. The two tails are also pendant paths attached to $H$. For
each such path $v_0v_1\cdots v_\ell$, where $v_0\in V(H)$, introduce the
chain of bags
\[
 \{v_0,v_1\},\{v_1,v_2\},\ldots,\{v_{\ell-1},v_\ell\},
\]
and join its first bag to an existing bag containing $v_0$. The new bag
chains are pairwise disjoint as sets of decomposition nodes. Every new bag is
a clique, every new edge is covered, and the bags containing any fixed vertex
remain connected. We obtain a tree decomposition
$(T,(K_u)_{u\in V(T)})$ of $G-C$ with clique bags.

Replace each bag $K_u$ by $K'_u=K_u\cup C$. All edges outside $C$ remain
covered, all edges inside $C$ are covered, and an edge between $C$ and a
vertex $v\notin C$ is covered by any enlarged bag containing $v$. The bags
containing a vertex outside $C$ remain connected, while every vertex of $C$
appears in every bag. This is therefore a tree decomposition of $G$.

Each bag $K'_u$ is the union of the two cliques $K_u$ and $C$. Consequently,
every independent set in $G[K'_u]$ has size at most two. The decomposition
has polynomially many bags, each of size at most $m+3$, and is produced by
the reduction. Thus $\ta(G)\leq2$.

\subsubsection{Correctness}

\begin{proof}[Proof of Theorem~\ref{thm:hardness}]
We begin with the structure of the shortest paths through the variable chain.

\begin{claim}\label{lem:coretree}
The distance between $a_0$ and $a_n$ is $2n$. Moreover, the shortest
$a_0$--$a_n$ paths are precisely the paths
\refstepcounter{equation}\label{eq:core}
\[
 a_0,z_1,a_1,z_2,a_2,\ldots,z_n,a_n,
 \qquad z_i\in\{p_i,q_i\}\quad(1\leq i\leq n)
 \qquad\text{(\theequation)}.
\]
\end{claim}

\begin{proof}
Every displayed path exists and has length $2n$. A simple path with both
endpoints in $H$ cannot enter a pendant tail, since it would have to leave
through the same attachment vertex. If an $a_0$--$a_n$ path leaves $H$
through an occurrence path, then it must traverse that path completely to
reach $C$ and traverse another occurrence path completely to return to $H$.
This uses at least $2R=2n+2$ edges. Hence every shortest $a_0$--$a_n$ path
is contained in $H$.

The consecutive variable gadgets meet only at their connector vertices. An
$a_0$--$a_n$ path in $H$ must therefore visit
$a_0,a_1,\ldots,a_n$ in this order. Crossing the $i$th gadget requires at
least two edges, with equality precisely for the paths
$a_{i-1}p_ia_i$ and $a_{i-1}q_ia_i$. Summing over all gadgets proves the
claim.
\end{proof}

For a path $L$ of the form~\eqref{eq:core}, define the assignment $\sigma_L$
by
\[
 \sigma_L(x_i)=\mathrm{true}
 \quad\Longleftrightarrow\quad
 p_i\in V(L).
\]
Exactly one of $p_i,q_i$ lies on $L$, so this assignment is well-defined. A
literal is true under $\sigma_L$ precisely when its corresponding literal
vertex lies on $L$.

Suppose first that $\Phi$ is satisfiable, and fix a satisfying assignment.
Choose $z_i=p_i$ when $x_i$ is true and $z_i=q_i$ when $x_i$ is false, and
let
$
 L_\sigma=a_0,z_1,a_1,z_2,a_2,\ldots,z_n,a_n.
$
By Claim~\ref{lem:coretree}, the path $L_\sigma$ is shortest. We show
that its eccentricity is at most $R$.

Every connector vertex belongs to $L_\sigma$. For each $i$, the literal
vertex not chosen by $L_\sigma$ is adjacent to the chosen one through the
edge $p_iq_i$. Hence every vertex of $H$ has distance at most one from
$V(L_\sigma)$. Every vertex of the left tail is at distance at most $R$ from
$a_0\in V(L_\sigma)$, and the analogous statement holds for the right tail
and $a_n$.

Consider an internal vertex $w$ of an occurrence path and let $z$ be the
literal endpoint of this path. The $w$--$z$ segment has length at most
$R-1$, while $\dist_G(z,V(L_\sigma))\leq1$. Hence
\[
 \dist_G(w,V(L_\sigma))
 \leq \dist_G(w,z)+\dist_G(z,V(L_\sigma))
 \leq (R-1)+1=R.
\]
We have now accounted for every vertex outside $C$, and each of them is at
distance at most $R$ from $V(L_\sigma)$.

It remains to consider the clause vertices. Since the fixed assignment
satisfies $D_j$, some literal of $D_j$ is true. Its literal vertex belongs to
$L_\sigma$, and the corresponding occurrence path joins it to $c_j$ in
exactly $R$ edges. Thus
$
 \dist_G(c_j,V(L_\sigma))\leq R
$
for every $j\in\{1,\ldots,m\}$. Consequently,
$\ecc_G(L_\sigma)\leq R$, and $(G,R)$ is a yes-instance of \ESP{}.

For the reverse direction, suppose that $(G,R)$ is a yes-instance. Let $Q$
be a shortest path with $\ecc_G(Q)\leq R$. The next claim extracts from $Q$
the path encoding the desired assignment.

\begin{claim}\label{lem:forcing}
After possibly reversing $Q$, the path $Q$ contains a shortest
$a_0$--$a_n$ path
\[
 P=a_0,z_1,a_1,z_2,a_2,\ldots,z_n,a_n,
 \qquad z_i\in\{p_i,q_i\}\quad(1\leq i\leq n)
\]
as a subpath. Moreover, every vertex of $Q$ outside $P$ lies on one of the
two pendant tails.
\end{claim}

\begin{proof}
Let $T_s$ and $T_t$ denote the two pendant tails, including their attachment
vertices $a_0$ and $a_n$, respectively. Since
$\dist_G(s,a_0)=\dist_G(t,a_n)=R$, every vertex outside $T_s$ is at distance
at least $R+1$ from $s$, and every vertex outside $T_t$ is at distance at
least $R+1$ from $t$. As $\ecc_G(Q)\leq R$, the path $Q$ meets both $T_s$
and $T_t$. The two tails attach to the rest of $G$ only at $a_0$ and $a_n$.
Since $Q$ is connected, it follows that $a_0,a_n\in V(Q)$.

Let $P$ be the $a_0$--$a_n$ subpath of $Q$. Since every subpath of a
shortest path is shortest, Claim~\ref{lem:coretree} shows that, after reversing
$Q$ if necessary,
\[
 P=a_0,z_1,a_1,z_2,a_2,\ldots,z_n,a_n,
 \qquad z_i\in\{p_i,q_i\}\quad(1\leq i\leq n).
\]

It remains to locate the vertices of $Q$ outside $P$. Since every shortest
path is induced, $Q$ cannot extend $P$ at $a_0$ through the unselected
vertex of $\{p_1,q_1\}$: that vertex is adjacent to the selected one through
the edge $p_1q_1$ and would create a chord. Thus an extension at $a_0$ must
enter the left tail, whose degree-two vertices force the extension to remain
on that tail. The same argument, using the edge $p_nq_n$, shows that an
extension at $a_n$ lies on the right tail. Hence every vertex of
$V(Q)\setminus V(P)$ lies on one of the two pendant tails.
\end{proof}

Let $P$ be the subpath given by Claim~\ref{lem:forcing}. Any path from a
clause vertex to the left tail passes through $a_0$, and any path to the
right tail passes through $a_n$. Since $a_0,a_n\in V(P)$ and
$V(Q)\setminus V(P)$ contains only tail vertices, we have
$
 \dist_G(c_j,V(P))=\dist_G(c_j,V(Q))\leq R
$
for every $j\in\{1,\ldots,m\}$. It remains to translate this distance bound
back into a truth assignment.

\begin{claim}\label{lem:distances}
For every $j\in\{1,\ldots,m\}$,
\refstepcounter{equation}\label{eq:clause}
\[
 \dist_G(c_j,V(P))\leq R
 \quad\Longrightarrow\quad
 \text{$\sigma_{P}$ satisfies $D_j$}
 \qquad\text{(\theequation)}.
\]
\end{claim}

\begin{proof}
Fix $j$ and let $W$ be a shortest $c_j$--$V(P)$ path of length at most
$R$. The path $W$ starts in $C$ and ends in $H$. Consider the last vertex of
$C$ and the first vertex of $H$ on $W$. The segment between them is a
complete occurrence path of length $R$, since every internal vertex of an
occurrence path has degree two. As $W$ has length at most $R$, this segment
is the whole path $W$. In particular, $W$ starts at $c_j$ and ends at a
literal vertex of $P$. The corresponding literal occurs in $D_j$ and is
true under $\sigma_{P}$. Hence $\sigma_{P}$ satisfies $D_j$.
\end{proof}

Applying Claim~\ref{lem:distances} to every clause shows that
$\sigma_{P}$ satisfies $\Phi$. Hence $\Phi$ is satisfiable.

This proves the theorem, including the version in which the tree
decomposition and the clique partition of every bag are supplied.
\end{proof}

\begin{corollary}\label{cor:xp}
\ESP{} parameterized by tree-independence number is para-NP-hard.
\end{corollary}

\begin{proof}
Theorem~\ref{thm:hardness} gives NP-hardness already when the parameter is at
most two.
\end{proof}

\section{Kernelization}
\label{sec:kernelization}

This section presents a polynomial kernel for \ESP{} parameterized by
maximum leaf number, followed by kernel lower bounds for distance to
clique, feedback edge number, vertex cover number, and bandwidth.

\providecommand{\coNP}{\mathsf{coNP}}
\providecommand{\poly}{\mathrm{poly}}

\paragraph*{Cross-composition framework.}
We first recall the definitions used by the cross-composition
framework~\cite{BJK14}. An equivalence relation $\mathcal R$ on
$\Sigma^*$ is a \emph{polynomial equivalence relation} if it has the
following properties.
\begin{enumerate}
  \item Given strings $x$ and $y$, one can decide in time polynomial in
  $|x|+|y|$ whether $x\mathrel{\mathcal R}y$.
  \item For every finite set $S\subseteq\Sigma^*$, the number of
  equivalence classes of $\mathcal R$ represented in $S$ is bounded by a
  polynomial in the size of the largest string in $S$.
\end{enumerate}
Let $L\subseteq\Sigma^*$ be a language, let $\mathcal R$ be a polynomial
equivalence relation, and let $Q\subseteq\Sigma^*\times\mathbb N$ be a
parameterized problem. An \emph{OR-cross-composition} of $L$ into $Q$
with respect to $\mathcal R$ is a polynomial-time algorithm that receives
$t$ instances $x_1,\ldots,x_t$ of $L$ from the same equivalence class and
outputs an instance $(y,k)$ of $Q$ with the following properties.
\begin{enumerate}
  \item The parameter $k$ is bounded by a polynomial in
  $\max_{i\in[t]}|x_i|+\log t$.
  \item The instance $(y,k)$ is a yes-instance if and only if at least
  one input instance $x_i$ is a yes-instance.
\end{enumerate}

\begin{theorem}[Bodlaender, Jansen, and
Kratsch~\cite{BJK14}, Corollary~3.6]
\label{thm:cross-composition-framework}
If an NP-hard language OR-cross-composes into a parameterized problem
$Q$, then $Q$ admits no polynomial kernelization or
polynomial compression unless $\NP\subseteq\coNP/\poly$ and the
polynomial hierarchy collapses.
\end{theorem}

\subsection{Maximum leaf number}
\label{sec:mesp-max-leaf}

Recall that $\ell=\ml(G)$ is the maximum number of leaves in a spanning
tree of $G$, as defined in Subsection~\ref{sec:graph-parameters}.

We prove that \ESP{} admits a polynomial kernel when parameterized by
maximum leaf number. The main idea is to suppress maximal degree-two chains
and store their lengths in binary. A small-core representation gives a
polynomial compression, and an explicit certificate shows that the compressed
language belongs to $\NP$.

\begin{restatable}{theorem}{thmmaxleafkernel}
\label{thm:mainleaf}
\ESP{} admits a polynomial kernel parameterized by maximum leaf number.
More precisely, it admits an $O(\ell^4)$-bit polynomial compression into a
fixed language in $\NP$, and consequently an \ESP{} kernel of encoding
length $\ell^{O(1)}$.
\end{restatable}

\subsubsection{A small weighted core}

A weighted multigraph may contain parallel edges and loops. If every edge
$e$ is assigned a positive integer length $L_e$, its subdivision is obtained
by replacing $e$ by an internally vertex-disjoint path of length $L_e$; a
loop is replaced by a cycle of length $L_e$ based at its incident vertex.

\begin{restatable}{lemma}{maxleafcore}
\label{lem:core}
Let $G$ be connected, let $n\geq2$, and let $\ml(G)=\ell$. In polynomial
time, one can construct a connected multigraph $H$, possibly with loops, and
positive integer lengths $(L_e)_{e\in E(H)}$ such that replacing every edge
$e$ by a path of length $L_e$ reconstructs $G$. Moreover, writing
$p=|V(H)|$ and $q=|E(H)|$, we have
\[
 p<4\ell,\qquad 1\leq q<2\ell^2,\qquad L_e\leq n.
\]
For $0\leq R<n$, the tuple $(H,(L_e),R)$ has an encoding of length
$O(q\log n)$ bits.
\end{restatable}

\begin{proof}
Let $B=\{v\in V(G):\deg_G(v)\neq2\}$. If $B=\varnothing$, then $G$ is a
cycle; in this case, choose an arbitrary vertex and add it to $B$. Suppress the internal vertices of every maximal path whose endpoints belong
to $B$ and whose internal vertices have degree two in $G$. Replace each such
path by a single edge of a multigraph $H$, and assign to this edge the length
of the corresponding path in $G$. If both endpoints of such a path coincide,
the resulting edge is a loop; distinct paths with the same pair of endpoints
give parallel edges. The
internal vertices of distinct maximal degree-two paths are disjoint, and every
edge of $G$ belongs to exactly one such path. Therefore, replacing each edge
of $H$ by its corresponding path reconstructs $G$ exactly. These paths and
their lengths can be found in polynomial time by a traversal of $G$.

By Bouland's subdivision bound, as restated by Ku\v{c}era and
Such\'y~\cite[Lemma~28]{KuceraSuchy2023}, a connected graph of
maximum leaf number $\ell$ is a subdivision of a multigraph on fewer
than $4\ell$ vertices. Hence $p<4\ell$. We also have
$\Delta(G)\leq\ell$. Indeed, a breadth-first spanning tree rooted at a vertex
of degree $d$ has at least one leaf in each of its $d$ branches. Degrees at
the retained vertices are unchanged, with a loop counted twice. Consequently,
\[
 2q=\sum_{v\in B}\deg_G(v)\leq|B|\ell<4\ell^2.
\]
Thus $q<2\ell^2$. Since $G$ is connected and has at least two vertices, it
has at least one edge and therefore at least one maximal degree-two path.
Hence $q\geq1$.
Moreover, every vertex of $H$ is incident with at least one edge because
$H$ is connected and has at least one edge. Double-counting vertex--edge
incidences, with each loop contributing two incidences, gives
$
 p\leq\sum_{v\in V(H)}\deg_H(v)=2q.
$

Finally, every edge $e\in E(H)$ represents a maximal degree-two path in $G$.
Since such a path contains at least one and at most $n$ edges, we have
$
 1\leq L_e\leq n
$
for every $e\in E(H)$. Relabel the retained vertices by
$1,\ldots,p$. Every length is at most $n$, every vertex label uses
$O(\log n)$ bits, and $R<n$. Storing the edge list, the lengths, and $R$
therefore uses $O(q\log n)$ bits.
\end{proof}

\begin{restatable}{lemma}{maxleafenumeration}
\label{lem:enumeration}
Given an explicit graph $G$ and a representation from
Lemma~\ref{lem:core} with $q$ edges, \ESP{} can be decided in
$2^q n^{O(1)}$ time.
\end{restatable}

\begin{proof}
First test the $n$ one-vertex paths. Next enumerate all pairs of distinct
vertices $s,t\in V(G)$. Retain $s$ and $t$ as additional vertices of the
multigraph representation, splitting their chains when necessary. Let
$H_{s,t}$ be the resulting weighted multigraph. It has at most $q+2$ edges,
and its subdivision is still exactly $G$.

Every simple $s$--$t$ path in $G$ traverses each chain of $H_{s,t}$ either
completely or not at all: the endpoints are retained vertices, and every
other internal chain vertex has degree two. Thus the path is the expansion of
a simple $s$--$t$ path in $H_{s,t}$.

Enumerate all subsets $F\subseteq E(H_{s,t})$, of which there are at most
$2^{q+2}$. Keep a subset only if its nonisolated subgraph is a simple
$s$--$t$ path. This can be checked using connectivity and degrees, while
rejecting loops and cycles. The expanded path is shortest precisely when
$
 \sum_{e\in F}L_e=\dist_G(s,t).
$
For every surviving choice, expand the path in $G$ and compute its
eccentricity by a multi-source breadth-first search. Every feasible path is
considered, and every accepted path satisfies the required conditions. There
are $O(n^2 2^q)$ candidates, and each is processed in polynomial time.
\end{proof}

 Note that the FPT-ness of ESP parameterized by maximum leaf number is already established in \cite{KuceraSuchy2023}. Since $q<2\ell^2$, Lemma~\ref{lem:enumeration} also gives an
$2^{O(\ell^2)}n^{O(1)}$-time FPT algorithm for ESP parameterized by the maximum leaf number.   We next use the same enumeration
algorithm to obtain a polynomial compression.

\subsubsection{A polynomial-size compression}

Define the language \SubESP{} as follows. An input consists of a connected
multigraph $H$, a positive integer $L_e$ for every $e\in E(H)$, and a
nonnegative integer $R$, with all numbers written in binary. Replace each edge
by an internally vertex-disjoint path of the specified length; a loop is
replaced by a cycle based at its incident vertex. Denote the expanded
unweighted graph by $\operatorname{Sub}(H,L)$. The input belongs to
\SubESP{} exactly when the expansion is simple and has a shortest path of
eccentricity at most $R$.

Simplicity can be checked without performing the expansion: each loop must
have length at least three, and there cannot be two length-one edges with the
same pair of distinct endpoints. Malformed inputs are rejected. Retained
vertices may have arbitrary degree.

\begin{restatable}{proposition}{maxleafcompression}
\label{prop:compression}
There is a polynomial-time compression of \ESP{} parameterized by maximum
leaf number $\ell$ into \SubESP{} with $O(\ell^4)$ output bits.
\end{restatable}

\begin{proof}
Let $n=|V(G)|$. If $n=1$ or $R\geq n$, return a fixed yes-instance of
\SubESP{}; in the latter case any one-vertex path has eccentricity at most
$n-1\leq R$. Hence assume that $n\geq2$ and $R<n$. Compute the
representation from Lemma~\ref{lem:core}, and let $q=|E(H)|$.

If $q\leq\log_2 n$, run the algorithm of Lemma~\ref{lem:enumeration}. Since
$2^q\leq n$, this takes polynomial time. Return a fixed yes- or no-instance
of \SubESP{} according to the answer. For example, use
$(K_2,L\equiv1,R=0)$ and $(K_3,L\equiv1,R=0)$, respectively.

Otherwise, output $(H,(L_e),R)$ itself. Its expansion is exactly $G$, so the
output is equivalent to the input. Since $\log_2 n<q$, its encoding length is
$
 O(q\log n)=O(q^2)=O(\ell^4).
$

The exhaustive algorithm is used only in the branch in which it runs in
polynomial time; the other branch merely constructs the representation. The
value of $\ell$ need not be computed.
\end{proof}

\subsubsection{The compressed language belongs to NP}

The implicit subdivision may be exponentially larger than its encoding.
Thus, membership of ordinary \ESP{} in $\NP$ does not by itself imply that
\SubESP{} belongs to $\NP$. The next lemma gives a polynomial-size
certificate and an explicit verifier.

\begin{restatable}{lemma}{submespinNP}
\label{lem:np}
$\SubESP\in\NP$.
\end{restatable}

\begin{proof}
An endpoint of a proposed solution is specified either by a vertex of $H$, or
by an edge identifier and an integer offset strictly between $0$ and the
length of that edge. Fix an orientation of every edge, including loops, for
the purpose of specifying offsets. Specify both endpoints $s$ and $t$ in this
way, and split their edges at those positions. Identical endpoints are
identified. The resulting multigraph $H'$ has at most two additional vertices
and two additional edges, and its subdivision is the original expanded graph.

The rest of the certificate is a simple path $Q$ in $H'$ from $s$ to $t$,
specified by its vertices and edge identifiers; it may consist of one vertex.
This description has polynomial length. As in Lemma~\ref{lem:enumeration},
every simple path in the expansion has such a certificate after its endpoints
are retained.

Compute weighted shortest-path distances $D$ in $H'$. The expansion $P$ of
$Q$ is shortest if and only if
\refstepcounter{equation}\label{eq:shortest}
\[
 \sum_{e\in E(Q)}L_e=D(s,t).\qquad\text{(\theequation)}
\]
All weights are positive integers, so these computations take polynomial time
in the binary input size.

To verify the eccentricity, put $S=V(Q)$ and define
\[
 \delta(v)=\min_{z\in S}D(v,z)\qquad(v\in V(H')).
\]
For a retained vertex $v$, this is its distance to $P$ in the expanded graph.
Indeed, reaching an internal vertex of a selected edge requires first reaching
one of its endpoints, both of which lie in $S$, while all internal vertices of
a selected edge already belong to $P$.

Consider an unselected edge $e=uv$ of length $L$. The vertex at integer
position $i$ on its chain, where $0\leq i\leq L$, has distance
\refstepcounter{equation}\label{eq:point-distance}
\[
 \min\{\delta(u)+i,\delta(v)+L-i\}\qquad\text{(\theequation)}
\]
from $P$. The same formula applies to loops. Distances to a set are Lipschitz
along edges, and hence $|\delta(u)-\delta(v)|\leq L$. The two affine functions
in~\eqref{eq:point-distance} therefore intersect inside $[0,L]$. Evaluating
at the nearest integer or integers gives
\refstepcounter{equation}\label{eq:maximum}
\[
 \max_{i\in\{0,\ldots,L\}}
 \min\{\delta(u)+i,\delta(v)+L-i\}
 =\left\lfloor\frac{L+\delta(u)+\delta(v)}{2}\right\rfloor
 \qquad\text{(\theequation)}.
\]
Consequently, $\ecc(P)\leq R$ holds exactly when
\refstepcounter{equation}\label{eq:retained}
\[
 \delta(v)\leq R\quad\text{for every }v\in V(H')
 \qquad\text{(\theequation)}
\]
and
\refstepcounter{equation}\label{eq:chain}
\[
 L_e+\delta(u)+\delta(v)\leq2R+1
 \quad\text{for every }e=uv\notin E(Q)
 \qquad\text{(\theequation)}.
\]
The constant $1$ in~\eqref{eq:chain} accounts for the discrete subdivision
vertices.

Thus the verifier can check path validity,~\eqref{eq:shortest},
\eqref{eq:retained}, and~\eqref{eq:chain} in time polynomial in the compressed
input size, without expanding a chain. Every accepted certificate yields a
feasible path, and every feasible path supplies an accepted certificate.
Hence $\SubESP\in\NP$.
\end{proof}

\subsubsection{From compression to an ordinary ESP kernel}

We finish by translating the compressed instance back to an ordinary
unweighted \ESP{} instance.

\begin{proof}[Proof of Theorem~\ref{thm:mainleaf}]
Proposition~\ref{prop:compression} and Lemma~\ref{lem:np} give a polynomial
compression into a fixed language in $\NP$. The ordinary decision version of
\ESP{} is NP-complete, already on connected unweighted simple
graphs~\cite[Theorem~7]{DraganLeitert2017}. Hence there is a fixed
polynomial-time many-one reduction
\[
 f:\SubESP\longrightarrow\ESP.
\]
Apply $f$ to the compressed output $y$. There is a constant $c$ such that the
encoding length of the resulting ordinary instance $(G',R')$ is at most
$O(|y|^c)=\ell^{O(1)}$. Moreover,
\[
 \begin{aligned}
 (G,R)\in\ESP
 &\quad\Longleftrightarrow\quad y\in\SubESP\\
 &\quad\Longleftrightarrow\quad (G',R')\in\ESP.
 \end{aligned}
\]
The output parameter is automatically bounded, since
$\ml(G')\leq|V(G')|\leq\ell^{O(1)}$. If the formulation includes a promised
upper bound on maximum leaf number, use $|V(G')|$ as the output bound. The
entire mapping runs in polynomial time in the original input size and returns
one unweighted \ESP{} instance. It is therefore a polynomial many-one
kernelization, rather than merely a weighted compression or a Turing
kernelization.
\end{proof}

\providecommand{\ESP}{\textnormal{\textsc{Eccentricity Shortest Path}}}
\providecommand{\SAT}{\textsc{CNF-Sat}}
\providecommand{\dist}{\operatorname{dist}}
\providecommand{\ecc}{\operatorname{ecc}}
\providecommand{\NP}{\mathsf{NP}}

\subsection{Distance to clique}
\label{sec:mesp-dtc-kernel-lower-bound}

We prove a kernel lower bound for \ESP\ parameterized by distance to
clique via an OR-cross-composition from \SAT. Recall that $\dtc(G)$ is the
minimum size of a clique modulator, namely a set $X$ such that $G-X$ is a
clique; see Subsection~\ref{sec:graph-parameters}.

\begin{theorem}
\label{thm:mesp-dtc-kernel-lower-bound}
The problem \ESP\ parameterized by distance to clique has no polynomial
kernel unless $\NP\subseteq\coNP/\poly$. It also holds when the input
includes a clique modulator and is parameterized by its size.
\end{theorem}

\subsubsection{Source instances and the selector formula}

To prove Theorem~\ref{thm:mesp-dtc-kernel-lower-bound}, we give an
OR-cross-composition from \SAT. We restrict the source language to
formulas that explicitly list at least one variable, contain at least
one clause, and contain no empty clause. This restriction remains
NP-hard. We define a polynomial equivalence relation $\mathcal R$ such that all strings that do not encode a valid source instance are equivalent. Of the remaining instances, two formulas are equivalent if and only if they contain the same number of variables. Before comparing formulas, we relabel their variables as $x_1,\ldots,x_n$. Any set of valid inputs of encoding length at most $L$ is partitioned into at most $L$ equivalence classes, so $\mathcal R$ is a polynomial equivalence
relation. If the input instances are malformed, then we output a fixed
no-instance.

Let $\Phi_1,\ldots,\Phi_t$ be valid equivalent instances, each on the
variables $x_1,\ldots,x_n$. Set $q=\lceil\log_2t\rceil$ and duplicate
$\Phi_1$ until the list contains $2^q$ formulas. This does not change
whether at least one input formula is satisfiable. Introduce selector
variables $y_1,\ldots,y_q$. For every $i\in[2^q]$, let $L_i$ be the
disjunction of selector literals that is false exactly when the selector
variables encode $i-1$ in binary. Replace every clause $D$ of $\Phi_i$
by $D\lor L_i$, and let $\Psi$ be the conjunction of all resulting
clauses. If $q=0$, then $L_1$ is the empty disjunction and is false. Fixing the selector variables to the encoding of $i-1$ leaves precisely
the clauses originating from $\Phi_i$ to be satisfied; every other
clause contains a true selector literal. Thus a satisfying assignment of
$\Phi_i$ extends to a satisfying assignment of $\Psi$ by setting the
selector variables to encode $i-1$. Conversely, consider any satisfying
assignment of $\Psi$. Its values on $y_1,\ldots,y_q$ encode a unique
index $i\in[2^q]$, for which $L_i$ is false. For every clause $D$ of
$\Phi_i$, the clause $D\lor L_i$ belongs to $\Psi$ and must be true.
Since $L_i$ is false, $D$ itself must be true. Hence the restriction of
the assignment to $x_1,\ldots,x_n$ satisfies every clause of $\Phi_i$.
If $i>t$, then $\Phi_i$ is a copy of $\Phi_1$, so in either case an
original input formula is satisfiable. Consequently,
\begin{equation}
\label{eq:mesp-dtc-selector-equivalence}
  \Psi\text{ is satisfiable}
  \quad\Longleftrightarrow\quad
  \text{at least one of }\Phi_1,\ldots,\Phi_t\text{ is satisfiable}.
\end{equation}

\subsubsection{Graph construction}

The formula $\Psi$ has $N:=n+q$ variables, all its clauses are nonempty,
and it can be constructed in time polynomial in the total input size. Rename the variables of $\Psi$ as $x_1,\ldots,x_N$, write its
clauses as $D_1,\ldots,D_m$, and set $R=N+2$.
We construct an instance $(G,R)$ of \ESP\ as follows (see
Figure~\ref{fig:mesp-dtc-construction}).

\begin{enumerate}
  \item Create vertices $s,t$ and, for every $i\in[N]$, a layer
  $U_i=\{u_i^0,u_i^1\}$. Join $s$ to both vertices of $U_1$, make
  consecutive layers complete to each other, and join both vertices of
  $U_N$ to $t$. Let $B=\{s,t\}\cup\bigcup_{i=1}^N U_i$, and call the
  subgraph induced by $B$ the \emph{assignment backbone}.
  For an assignment $\alpha\in{\{0,1\}}^{N}$, define its \emph{assignment
  path} as $P_\alpha=(s,u_1^{\alpha_1},\ldots,
  u_N^{\alpha_N},t)$.

  \item For every $i\in[N]$ and $b\in\{0,1\}$, add a path of length
  $R-1$ from $u_i^b$ to a new vertex $w_i^b$. These paths are the
  \emph{literal tails}. Create a clique
  $C=\{c_1,\ldots,c_m\}$. Join $c_j$ to $w_i^1$ if $x_i\in D_j$, and
  join $c_j$ to $w_i^0$ if $\neg x_i\in D_j$. Thus $c_j$ is adjacent
  precisely to the endpoints of literal tails that represent literals
  satisfying $D_j$.

  \item Attach a path $s=a_0,a_1,\ldots,a_R$
  at $s$, and a path $t=b_0,b_1,\ldots,b_R$
  at $t$. These paths are the \emph{forcing tails}. All internal and
  pendant vertices introduced in the construction are distinct.
\end{enumerate}

Each labeled line in Figure~\ref{fig:mesp-dtc-construction} represents
a path of the indicated number of edges, with its internal subdivision
vertices suppressed. The drawing is not to scale, and unmarked
crossings are not vertices.

\begin{figure}[!htbp]
\centering
\resizebox{\linewidth}{!}{%
\begin{tikzpicture}[
  x=1cm,y=0.9cm,font=\small,
  vertex/.style={circle,fill=black,inner sep=1.7pt},
  endpoint/.style={vertex},
  clause/.style={circle,draw=black,fill=white,
    line width=0.8pt,inner sep=2.2pt},
  force/.style={thick},
  literal/.style={thick,violet!80!black},
  selector/.style={thick},
  chosen/.style={very thick,blue!65!black},
  clique/.style={very thick,teal!70!black},
  omitted/.style={thick,densely dotted},
  labelbox/.style={fill=white,text=black,inner sep=1.5pt,font=\footnotesize},
  every label/.style={fill=white,inner sep=1pt}
]
\node[vertex,label=below left:$s$] (s) at (0,0) {};
\node[vertex,label={[xshift=-3pt]above left:$u_1^1$}] (u11) at (1.4,0.8) {};
\node[vertex,label=below left:$u_1^0$] (u10) at (1.4,-0.8) {};
\node[vertex,label={[xshift=-3pt]above left:$u_2^1$}] (u21) at (3.6,0.8) {};
\node[vertex,label=below left:$u_2^0$] (u20) at (3.6,-0.8) {};
\node[vertex,label={[xshift=-3pt]above left:$u_i^1$}] (ui1) at (7,0.8) {};
\node[vertex,label=below:$u_i^0$] (ui0) at (7,-0.8) {};
\node[vertex,label={[xshift=-3pt]above left:$u_{N-1}^1$}] (uNm1) at (10.4,0.8) {};
\node[vertex,label=below left:$u_{N-1}^0$] (uNm0) at (10.4,-0.8) {};
\node[vertex,label={[xshift=-3pt]above left:$u_N^1$}] (uN1) at (12.6,0.8) {};
\node[vertex,label=below right:$u_N^0$] (uN0) at (12.6,-0.8) {};
\node[vertex,label=below right:$t$] (t) at (14,0) {};
\draw[selector] (s) -- (u11);
\draw[selector] (s) -- (u10);
\foreach \firstLayerVertex in {u11,u10}
  \foreach \secondLayerVertex in {u21,u20}
    \draw[selector] (\firstLayerVertex) -- (\secondLayerVertex);
\foreach \penultimateLayerVertex in {uNm1,uNm0}
  \foreach \lastLayerVertex in {uN1,uN0}
    \draw[selector] (\penultimateLayerVertex) -- (\lastLayerVertex);
\draw[selector] (uN1) -- (t);
\draw[selector] (uN0) -- (t);
\draw[omitted] (u21) -- (5.3,0.8) -- (ui1);
\draw[omitted] (u21) -- (5.3,0) -- (ui0);
\draw[omitted] (u20) -- (5.3,0) -- (ui1);
\draw[omitted] (u20) -- (5.3,-0.8) -- (ui0);
\draw[omitted] (ui1) -- (8.7,0.8) -- (uNm1);
\draw[omitted] (ui1) -- (8.7,0) -- (uNm0);
\draw[omitted] (ui0) -- (8.7,0) -- (uNm1);
\draw[omitted] (ui0) -- (8.7,-0.8) -- (uNm0);
\node[labelbox] at (5.3,-1.4) {$\cdots$};
\node[labelbox] at (8.7,-1.4) {$\cdots$};
\node[font=\footnotesize] at (7,-2.2)
  {assignment layers $B$; consecutive layers are complete to each other};

\draw[chosen] (s) -- (u11) -- (u21);
\draw[chosen,densely dashed] (u21) -- (ui0) -- (uNm1);
\draw[chosen] (uNm1) -- (uN1) -- (t);
\node[labelbox] at (7,-1.55)
  {one assignment path $P_\alpha$};

\node[endpoint,label=above left:$w_1^1$] (w11) at (0.85,2.65) {};
\node[endpoint,label=above:$w_1^0$] (w10) at (1.95,2.65) {};
\node[endpoint,label=above:$w_2^1$] (w21) at (3.05,2.65) {};
\node[endpoint,label=above:$w_2^0$] (w20) at (4.15,2.65) {};
\node[endpoint,label=above:$w_i^1$] (wi1) at (6.45,2.65) {};
\node[endpoint,label=above right:$w_i^0$] (wi0) at (7.55,2.65) {};
\node[endpoint,label=above:$w_{N-1}^1$] (wNm1) at (9.85,2.65) {};
\node[endpoint,label=above:$w_{N-1}^0$] (wNm0) at (10.95,2.65) {};
\node[endpoint,label=above right:$w_N^1$] (wN1) at (12.05,2.65) {};
\node[endpoint,label=above:$w_N^0$] (wN0) at (13.15,2.65) {};
\draw[literal] (u11) -- (w11);
\draw[literal] (u10) -- (w10);
\draw[literal] (u21) -- (w21);
\draw[literal] (u20) -- (w20);
\draw[literal] (ui1) -- node[labelbox,left] {$R-1$} (wi1);
\draw[literal] (ui0) -- (wi0);
\draw[literal] (uNm1) -- (wNm1);
\draw[literal] (uNm0) -- (wNm0);
\draw[literal] (uN1) -- (wN1);
\draw[literal] (uN0) -- (wN0);

\node[clause] (c1) at (5.25,6.05) {$c_1$};
\node[clause] (cBefore) at (5.1,5.1) {$\cdots$};
\node[clause] (cj) at (7,4.4) {$c_j$};
\node[clause] (cAfter) at (8.9,5.1) {$\cdots$};
\node[clause] (cm) at (8.75,6.05) {$c_m$};
\draw[clique] (c1) -- (cj) -- (cm) -- (c1);
\foreach \cliqueStart/\cliqueEnd in {
  c1/cBefore,c1/cAfter,cBefore/cj,cBefore/cAfter,
  cBefore/cm,cj/cAfter,cAfter/cm}
  \draw[clique,densely dotted] (\cliqueStart) -- (\cliqueEnd);
\node[draw=teal!70!black,rounded corners,
  fit=(c1)(cBefore)(cj)(cAfter)(cm),
  inner sep=7pt,label=above:{clause clique $C$}] {};
\draw[literal] (cj) to[bend right=10] (w11);
\draw[literal] (cj) -- (wi0);
\draw[literal] (cj) to[bend left=10] (wN1);
\node[anchor=west,align=left,font=\footnotesize]
  at (9.55,5.6)
  {representative clause-literal edges\\for
   $D_j=(x_1\lor\neg x_i\lor x_N)$};

\node[vertex] (aR) at (0,-3) {};
\node[vertex] (bR) at (14,-3) {};
\draw[force] (s) -- node[labelbox,left] {$R$} (aR);
\draw[force] (t) -- node[labelbox,right] {$R$} (bR);
\node[font=\footnotesize,below] at (aR) {forcing tail};
\node[font=\footnotesize,below] at (bR) {forcing tail};
\end{tikzpicture}%
}
\caption{A sketch of the graph construction used in the proof of
Theorem~\ref{thm:mesp-dtc-kernel-lower-bound}.}
\label{fig:mesp-dtc-construction}
\end{figure}
\FloatBarrier

\subsubsection{Parameter and size bounds}

We define the parameter bound
\begin{equation}
\label{eq:mesp-dtc-parameter}
\begin{aligned}
  k':=|X|
  &=2+2N+2N(R-1)+2R\\
  &=2N^2+6N+6
   =O\bigl({(n+\log t)}^{2}\bigr).
\end{aligned}
\end{equation}
where $X=V(G)\setminus C$. Since $G-X=C$ is a clique, $X$ is a
clique modulator and $\dtc(G)\le k'$.
The four terms in $|X|$ count, respectively, the vertices $s,t$, the
assignment-layer vertices, the vertices introduced on the literal
tails, and the vertices introduced on the forcing tails. Since
$n\le\max_{i\in[t]}|\Phi_i|$ and
$q=\lceil\log_2t\rceil$, the value $k'$ is polynomial in the size of
the largest input instance plus $\log t$. The cross-composition outputs
the instance $(G,R)$ of \ESP, whose parameter value is at most $k'$.
If a clique modulator must be supplied with the input, we supply $X$,
whose size is $k'$.

The graph $G$ can be constructed in polynomial time. Let
$L:=\sum_{i=1}^t|\Phi_i|$ denote the total input length. Since every
input formula has positive length and contains all $n$ variables, we
have $t\le L$ and $n\le L$. Hence $q=O(\log L)$ and $N=O(L)$. The
padded list contains exactly $2^q<2t$ formulas. Padding only duplicates
$\Phi_1$, and adding selector literals does not change the number of
clauses. Since every input formula has at most $L$ clauses, we have
$m<2tL\le 2L^2$. The graph has $m+2+2N+2N(R-1)+2R=m+O(N^2)$
vertices. Its assignment backbone, literal tails, and forcing tails
contribute $O(N^2)$ edges; there are at most $2Nm$ clause--literal
edges; and the clause clique contributes $O(m^2)$ edges. Hence $G$ has
$O(N^2+Nm+m^2)=O(L^4)$ edges, so the output size is polynomial in the
total input length. The algorithm expands every path explicitly, and therefore $G$ is unweighted. Since every clause of
$\Psi$ is nonempty, every clause vertex has a neighbor outside $C$, and
hence $G$ is connected. The description also shows that $G$ is simple
and undirected.

\subsubsection{Canonical paths and clause distances}

We now argue that $(G,R)$ is a yes-instance of \ESP\ if and only if at
least one input formula $\Phi_i$ is satisfiable. We first establish two
structural properties of the constructed graph. All distances below are
measured in $G$.

\begin{lemma}
\label{lem:mesp-dtc-normal-form}
The shortest $s$--$t$ paths are exactly the assignment paths. Moreover,
if $(G,R)$ is a yes-instance, then some assignment path is a solution
path.
\end{lemma}

\begin{proof}
Every assignment path has length $N+1$. Consider an $s$--$t$ path that
contains a clause vertex. It contains a maximal subpath whose endpoints
lie in the assignment backbone
$B=\{s,t\}\cup\bigcup_{i=1}^N U_i$, whose internal vertices lie
outside $B$, and which contains a clause vertex. This subpath traverses
two entire literal tails and two clause-literal edges, and hence has
length at least
$2(R-1)+2=2R$. Its endpoints are assignment-layer vertices, so the
original path also uses at least one edge from $s$ before this subpath
and at least one edge to $t$ after it. Its total length is therefore at
least $2R+2>N+1$. A shortest $s$--$t$ path that avoids $C$
cannot enter a literal tail or a forcing tail, since each such path is
pendant in $G-C$. It therefore lies in $B$. Within $B$, an $s$--$t$
path has length at least $N+1$, with equality exactly when it visits one
vertex of each layer in order. Hence the shortest $s$--$t$ paths are
precisely the assignment paths.

Suppose that $Q$ is a solution path. Let
$A_s=\{a_1,\ldots,a_R\}$ and $A_t=\{b_1,\ldots,b_R\}$. Every vertex
outside $A_s\cup\{s\}$ is at distance greater than $R$ from $a_R$. If
$s\notin V(Q)$, then the eccentricity bound forces $Q$ to contain a
vertex of $A_s$. Since $s$ separates $A_s$ from the rest of $G$ and $Q$
is connected, this implies $V(Q)\subseteq A_s$. But then
$\dist_G(b_R,V(Q))>R$, a contradiction. Thus $s\in V(Q)$, and the
symmetric argument gives $t\in V(Q)$. The $s$--$t$ subpath of $Q$ is shortest between its endpoints, so it is
an assignment path $P_\alpha$. Orient $Q$ so that this subpath runs from
$s$ to $t$. If the predecessor of $s$ on $Q$ belonged to $U_1$, then it
would be adjacent to the vertex two positions after $s$ on $P_\alpha$:
the selected vertex of $U_2$ when $N\ge2$, or $t$ when $N=1$. This
would replace a three-edge subpath of $Q$ by one edge, contradicting the
shortestness of $Q$. Therefore any extension of $Q$ before $s$ lies in
$A_s$. Symmetrically, any extension after $t$ lies in $A_t$. Trim these two extensions from $Q$. For a vertex outside
$A_s\cup A_t$, any path to one of the trimmed extensions first reaches
$s$ or $t$, both of which remain on $P_\alpha$. Hence its distance to
$P_\alpha$ equals its distance to $Q$. Every vertex of a forcing tail is
within distance $R$ of its attachment vertex. It follows that
$\ecc_G(P_\alpha)\le R$, and $P_\alpha$ is a solution path.
\end{proof}

\begin{lemma}
\label{lem:mesp-dtc-clause-distances}
Let $P_\alpha$ be the assignment path corresponding to an assignment
$\alpha$. For every clause $D_j$,
\begin{equation}
\label{eq:mesp-dtc-clause-distance}
  \dist_G(c_j,V(P_\alpha))=
  \begin{cases}
    R,   &\text{if $\alpha$ satisfies $D_j$},\\
    R+1, &\text{if $\alpha$ falsifies $D_j$}.
  \end{cases}
\end{equation}
Every non-clause vertex is within distance $R$ of $P_\alpha$.
Consequently, $P_\alpha$ is a solution path if and only if $\alpha$
satisfies $\Psi$.
\end{lemma}

\begin{proof}
Every selected layer vertex lies on $P_\alpha$. Every unselected layer
vertex is adjacent to $P_\alpha$: for the first layer, use $s$, and for
each later layer, use the selected vertex of the preceding layer. A
vertex on a literal tail is therefore within distance
$(R-1)+1=R$ of $P_\alpha$. Every forcing-tail vertex is within distance
$R$ of $s$ or $t$. Thus every non-clause vertex is within distance $R$
of $P_\alpha$. Consider a path from $c_j$ to its first vertex in $B$.
After its last clause vertex, it traverses a clause--literal edge and an
entire literal tail, which costs $1+(R-1)=R$ edges. Therefore a path
from $c_j$ to $P_\alpha$ of length at most $R$ must use an edge
$c_j w_i^b$ followed
by the literal tail to a selected vertex $u_i^b$, where $b=\alpha_i$.
Such a path exists exactly when $\alpha$ satisfies $D_j$, and then its
length is exactly $R$. If $\alpha$ falsifies $D_j$, the distance is
therefore at least $R+1$.
Since $D_j$ is nonempty, choose one of its literals. The corresponding
literal tail ends at an unselected layer vertex, which is adjacent to
$P_\alpha$. This gives a path of length $R+1$ from $c_j$ to
$P_\alpha$, proving~\eqref{eq:mesp-dtc-clause-distance}. 
The final claim
follows because the vertices of $C$ correspond to all clauses of
$\Psi$.
\end{proof}

\subsubsection{Correctness of the cross-composition}

\begin{lemma}
\label{lem:mesp-dtc-equivalence}
The constructed instance $(G,R)$ is a yes-instance if and only if at
least one of $\Phi_1,\ldots,\Phi_t$ is satisfiable.
\end{lemma}

\begin{proof}
Suppose that $(G,R)$ is a yes-instance. By
Lemma~\ref{lem:mesp-dtc-normal-form}, some assignment path $P_\alpha$ is
a solution path. Lemma~\ref{lem:mesp-dtc-clause-distances} implies that
$\alpha$ satisfies $\Psi$. By
Equation~\eqref{eq:mesp-dtc-selector-equivalence}, at least one input
formula $\Phi_i$ is satisfiable.

Suppose that some input formula $\Phi_i$ is satisfiable. By
Equation~\eqref{eq:mesp-dtc-selector-equivalence}, the formula $\Psi$
has a satisfying assignment $\alpha$. The corresponding assignment path
$P_\alpha$ is a shortest path by
Lemma~\ref{lem:mesp-dtc-normal-form}, and
Lemma~\ref{lem:mesp-dtc-clause-distances} gives
$\ecc_G(P_\alpha)\le R$. Hence $(G,R)$ is a yes-instance.
\end{proof}

\begin{proof}[Proof of Theorem~\ref{thm:mesp-dtc-kernel-lower-bound}]
The construction can be performed in polynomial time, and
Equation~\eqref{eq:mesp-dtc-parameter} bounds the output parameter by a
polynomial in the maximum input size plus $\log t$.
Lemma~\ref{lem:mesp-dtc-equivalence} establishes the required
OR-property. Consequently, the construction is an OR-cross-composition
of \SAT\ into \ESP\ parameterized by distance to clique.
Theorem~\ref{thm:cross-composition-framework} therefore yields the
claimed kernel lower bound. Since the construction explicitly provides
the clique modulator $X$, the same conclusion holds when such a
modulator is included in the input and its size is the parameter.
\end{proof}

\providecommand{\ESP}{\textnormal{\textsc{Eccentricity Shortest Path}}}
\providecommand{\dist}{\operatorname{dist}}
\providecommand{\ecc}{\operatorname{ecc}}
\providecommand{\diam}{\operatorname{diam}}
\providecommand{\NP}{\mathsf{NP}}
\providecommand{\coNP}{\mathsf{coNP}}
\providecommand{\poly}{\mathrm{poly}}
\subsection{Feedback edge number}
\label{sec:mesp-fen-kernel-lower-bound}

We prove a kernel lower bound for \ESP\ parameterized by feedback edge
number via an OR-cross-composition from \textsc{Distinct 3-Sat}.
Recall that $\fen(G)$ is the minimum number of edges whose deletion leaves
a forest and equals $|E(G)|-|V(G)|+1$ for connected $G$; see
Subsection~\ref{sec:graph-parameters}.

We use \textsc{Distinct 3-Sat} as the source language. In this problem,
every clause contains three literals on distinct variables. This
restriction of \textsc{3-Sat} is NP-hard; stronger restrictions are also
known~\cite{DD21}.

\begin{theorem}
\label{thm:mesp-fen-kernel-lower-bound}
The problem \ESP\ parameterized by feedback edge number has no
polynomial kernel unless $\NP\subseteq\coNP/\poly$. It also holds when the
input includes a minimum feedback edge set.
\end{theorem}

\subsubsection{Source instances and preprocessing}

To prove Theorem~\ref{thm:mesp-fen-kernel-lower-bound}, we show that
\textsc{Distinct 3-Sat} OR-cross-composes into \ESP\ parameterized by
feedback edge number. By
Theorem~\ref{thm:cross-composition-framework}, this is sufficient to
establish the claim. We define a polynomial equivalence relation $\mathcal R$ such that all
bitstrings that do not encode a valid instance of
\textsc{Distinct 3-Sat} are equivalent. Of the remaining instances, any
two well-formed formulas are equivalent if and only if they contain the
same number of variables. Before comparing two formulas, we relabel the
variables that occur as $x_1,\ldots,x_n$. From this definition, it
follows that any set of well-formed instances of encoding length at most
$N$ is partitioned into at most $N$ equivalence classes. Thus
$\mathcal R$ is a polynomial equivalence relation.

We now give a cross-composition algorithm that composes $t$ input
instances that are equivalent under $\mathcal R$ into a single instance
of \ESP\ parameterized by feedback edge number. If the input instances
are malformed, then we may output a fixed no-instance. Hence, in the
remainder, we may assume that all inputs are well-formed and encode
formulas $\Phi_1,\ldots,\Phi_t$, each on the variables
$x_1,\ldots,x_n$. We delete repeated clauses. If $n<3$, then every input
formula has no clause and is satisfiable, so we may output a fixed
yes-instance. Hence assume that $n\ge3$ and $t\ge1$. Let $\mathcal C$ be the set of all clauses that use three distinct
variables from $\{x_1,\ldots,x_n\}$, with arbitrary signs. The number of such clauses is $c:=|\mathcal C|=8\binom n3$.
We construct a single instance $(G,R)$ of \ESP\ parameterized by
feedback edge number. It consists of a connected graph $G$, a
nonnegative integer $R$, and parameter value $\fen(G)$.

\subsubsection{Graph construction}

Every path introduced below has new internal vertices, and its length
is its number of edges. Define
\begin{equation}
\label{eq:mesp-composition-constants}
\begin{gathered}
  s=4t,\qquad K=ns+1=4nt+1,\qquad A=K+2,\\
  D=K+t,\qquad r=K+2t,\qquad R=K+10t=r+8t.
\end{gathered}
\end{equation}

We construct the graph $G$ as follows (see
Figures~\ref{fig:construction-overview} and~\ref{fig:clause-detail}).

\begin{enumerate}
  \item Create the \emph{assignment selector} with vertices
  $z_0,\ldots,z_n$ and, for each $i\in[n]$, literal vertices $v_i^+$ and
  $v_i^-$. For each sign $\sigma\in\{+,-\}$, add a path of length $s$
  from $z_{i-1}$ to $v_i^\sigma$ and a path of length $s$ from
  $v_i^\sigma$ to $z_i$. Add vertices $a$ and $b$, together with an edge
  between $a$ and $z_0$ and an edge between $z_n$ and $b$. Denote the
  selector by $\mathcal S$. Every
  $a,b$-path in $\mathcal S$ chooses one branch in each variable gadget
  and has length $1+2ns+1=2K$.
  We call such a path an \emph{assignment path}. It encodes a truth
  assignment: the positive branch in the $i$th gadget represents the
  value true for $x_i$.

  \item Create vertices $u$ and $v$. Join $a$ to $u$ by a \emph{connector path} of length $A$, and join $b$ to $v$ by another \emph{connector path} of
  length $A$. Add a
  \emph{guard path} of length $K+6t-1$ from $u$ to $z_n$, and add a \emph{guard path}
  of length $K+4t-3$ from $v$ to $z_0$.

  \item For each clause $C\in\mathcal C$, add a \emph{test path} $T_C$ of
  length $8t$ from $u$ to $v$. Let $c_C$ be its midpoint. For each
  literal of $C$, add an \emph{incidence path} of length $D$ from $c_C$ to the
  corresponding literal vertex in $\mathcal S$. For every $j\in[t]$,
  let $w_{C,j}$ be the vertex of $T_C$ at distance $4t+j$ from $u$.

  \item Attach a path of length $R$ at $a$ and another path of length
  $R$ at $b$; these paths are the \emph{forcing tails}. For each
  $j\in[t]$ and each clause $C\in\Phi_j$, attach a path of length $8t$
  at $w_{C,j}$. This path is the \emph{test hair} for $(C,j)$, and its
  leaf is denoted by $\lambda_{C,j}$.
\end{enumerate}

Each labeled line in Figures~\ref{fig:construction-overview}
and~\ref{fig:clause-detail} represents a path of the indicated number
of edges, with its internal subdivision vertices suppressed. The
drawings are not to scale, and unmarked crossings are not vertices.

\begin{figure}[!htbp]
\centering
\resizebox{\linewidth}{!}{%
\begin{tikzpicture}[
  x=1.08cm,y=0.9cm,font=\small,
  vertex/.style={circle,fill=black,inner sep=1.7pt},
  length/.style={fill=white,inner sep=2pt,font=\footnotesize},
  connector/.style={very thick,blue!65!black},
  guard/.style={thick,dashed,orange!80!black},
  test/.style={thick,teal!70!black}
]
\node[vertex,label=below left:$a$] (a) at (0.3,0) {};
\node[vertex,label=below:$z_0$] (z0) at (1.3,0) {};
\node[vertex,label=above:$v_1^+$] (v1p) at (2.3,0.8) {};
\node[vertex,label=below:$v_1^-$] (v1m) at (2.3,-0.8) {};
\node[vertex,label=below:$z_1$] (z1) at (3.3,0) {};
\node[vertex] (omitlp) at (4.3,0.8) {};
\node[vertex] (omitlm) at (4.3,-0.8) {};
\node[vertex,label=below:$z_{i-1}$] (zim) at (5.3,0) {};
\node[vertex,label=above:$v_i^+$] (vip) at (6.3,0.8) {};
\node[vertex,label=below:$v_i^-$] (vim) at (6.3,-0.8) {};
\node[vertex,label=below:$z_i$] (zi) at (7.3,0) {};
\node[vertex] (omitrp) at (8.3,0.8) {};
\node[vertex] (omitrm) at (8.3,-0.8) {};
\node[vertex,label=below:$z_{n-1}$] (znm) at (9.3,0) {};
\node[vertex,label=above:$v_n^+$] (vnp) at (10.3,0.8) {};
\node[vertex,label=below:$v_n^-$] (vnm) at (10.3,-0.8) {};
\node[vertex,label=below:$z_n$] (zn) at (11.3,0) {};
\node[vertex,label=below right:$b$] (b) at (12.3,0) {};
\node[vertex,label=left:$u$] (u) at (0.3,4) {};
\node[vertex,label=right:$v$] (v) at (12.3,4) {};
\draw[connector] (a) -- node[length,left] {$A$} (u);
\draw[connector] (b) -- node[length,right] {$A$} (v);
\draw[guard] (u) -- node[length,pos=0.27,sloped,below]
  {$K+6t-1$} (zn);
\draw[guard,preaction={draw=white,line width=3pt}]
  (v) -- node[length,pos=0.27,sloped,below] {$K+4t-3$} (z0);
\draw[test] (u) to[bend left=16]
  node[length,above] {$T_{C'}$: length $8t$} (v);
\draw[test] (u) -- node[length,pos=0.25,above] {$4t$}
  node[length,pos=0.75,above] {$4t$} (v);
\node[vertex,label=below:$c_C$] at (6.3,4) {};
\draw[test] (u) to[bend right=12] (v);
\node[length] at (6.3,3.05) {$\vdots$};
\node[font=\footnotesize] at (6.3,5.95)
  {one $T_C$ for each $C\in\mathcal C$};
\draw[thick] (a) -- node[length,above] {$1$} (z0);
\draw[thick] (z0) -- node[length,above left] {$s$} (v1p)
  -- node[length,above right] {$s$} (z1);
\draw[thick] (z0) -- node[length,below left] {$s$} (v1m)
  -- node[length,below right] {$s$} (z1);
\draw[thick,densely dotted] (z1) -- (omitlp) -- (zim);
\draw[thick,densely dotted] (z1) -- (omitlm) -- (zim);
\node[length] at (4.3,0) {$\cdots$};
\draw[thick,densely dotted] (zim) -- (vip) -- (zi);
\draw[thick,densely dotted] (zim) -- (vim) -- (zi);
\draw[thick,densely dotted] (zi) -- (omitrp) -- (znm);
\draw[thick,densely dotted] (zi) -- (omitrm) -- (znm);
\node[length] at (8.3,0) {$\cdots$};
\draw[thick] (znm) -- node[length,above left] {$s$} (vnp)
  -- node[length,above right] {$s$} (zn);
\draw[thick] (znm) -- node[length,below left] {$s$} (vnm)
  -- node[length,below right] {$s$} (zn);
\draw[thick] (zn) -- node[length,above] {$1$} (b);
\draw[thick] (a) -- node[length,left] {$R$} (0.3,-2);
\draw[thick] (b) -- node[length,right] {$R$} (12.3,-2);
\node[vertex] at (0.3,-2) {};
\node[vertex] at (12.3,-2) {};
\node[font=\footnotesize,below] at (0.3,-2) {forcing tail};
\node[font=\footnotesize,below] at (12.3,-2) {forcing tail};
\end{tikzpicture}%
}
\caption{
A sketch of the construction used in the proof of
Theorem~\ref{thm:mesp-fen-kernel-lower-bound}. Incidence paths and test hairs are omitted here and shown in
Figure~\ref{fig:clause-detail}.}
\label{fig:construction-overview}
\end{figure}

\begin{figure}[!htbp]
\centering
\begin{tikzpicture}[
  x=1cm,y=1cm,font=\small,
  vertex/.style={circle,fill=black,inner sep=1.7pt},
  length/.style={fill=white,inner sep=2pt,font=\footnotesize},
  test/.style={very thick,teal!70!black},
  incidence/.style={thick,violet!80!black}
]
\node[vertex,label=left:$u$] (u) at (0,0) {};
\node[vertex,label=below left:$c_C$] (midpoint) at (5,0) {};
\node[vertex,label=below right:$w_{C,j}$] (root) at (8,0) {};
\node[vertex,label=right:$v$] (v) at (12,0) {};
\draw[test] (u) -- node[length,above] {$4t$} (midpoint)
  -- node[length,above] {$j$} (root)
  -- node[length,above] {$4t-j$} (v);
\node[font=\footnotesize] at (2.3,-0.55) {$T_C$: total length $8t$};
\node[vertex,label=above:$v_p^{\sigma_p}$] (literal1) at (1.5,2.5) {};
\node[vertex,label=above:$v_q^{\sigma_q}$] (literal2) at (5,2.5) {};
\node[vertex,label=above:$v_h^{\sigma_h}$] (literal3) at (8.5,2.5) {};
\draw[incidence] (midpoint) -- node[length,sloped,above] {$D$} (literal1);
\draw[incidence] (midpoint) -- node[length,right] {$D$} (literal2);
\draw[incidence] (midpoint) -- node[length,sloped,above] {$D$} (literal3);
\node[font=\footnotesize,align=center] at (5,3.6)
  {the three literal vertices of $C$ in $\mathcal S$\\
   (in three distinct variable gadgets)};
\node[vertex,label=below:$\lambda_{C,j}$] (leaf) at (8,-2) {};
\draw[thick] (root) -- node[length,left] {$8t$} (leaf);
\node[anchor=west,align=left,font=\footnotesize] at (8.5,-1.25)
  {test hair present\\exactly when $C\in\Phi_j$};
\end{tikzpicture}
\caption{A sketch of one of the test paths $T_C$ and one of the possible test hairs.}
\label{fig:clause-detail}
\end{figure}
\FloatBarrier

\subsubsection{Parameter and size bounds}

We define the parameter value
\begin{equation}
\label{eq:mesp-feedback-edge-number}
  k':=\fen(G)=n+2+4c=n+2+32\binom n3=O(n^3).
\end{equation}
The selector contributes $n$ to the feedback edge number, one for each
variable gadget. Each connector path is attached at one existing vertex
and creates no cycle. Each guard path joins two vertices already in the
connected graph and therefore contributes one. For each clause $C$, the
test path contributes one, and the three incidence paths contribute
three. The forcing tails and the test hairs are pendant paths and
contribute nothing. This proves~\eqref{eq:mesp-feedback-edge-number}.
The value $k'$ is the parameter of the output instance. Since every
variable occurs in a well-formed input formula,
$n\le\max_{j\in[t]}|\Phi_j|$. Thus $k'$ is polynomial in the size of
the largest input instance, and hence in
$\max_{j\in[t]}|\Phi_j|+\log t$. The cross-composition outputs the instance $x':=(G,R)$ of \ESP, whose
parameter value is $k'$. The graph $G$ can be constructed in polynomial
time. Indeed, the construction has $O(n+c)$ nonpendant paths, each of
length $O(nt)$.
It has at most $ct$ test hairs, each of length $8t$, and two forcing
tails of length $O(nt)$. Hence the numbers of vertices and edges are
$O(n^4t+n^3t^2)$. Let $L:=\sum_{j=1}^t|\Phi_j|$ denote the total input length. Since each input formula has positive length and contains all $n$ variables, we have
$t\le L$ and $n\le L$. Consequently, $n^4t+n^3t^2\le L^5+L^5=2L^5$,
so the output size is polynomial in the total input length. The
algorithm expands every path explicitly, and therefore $G$ is
unweighted. The description also shows that $G$ is connected, simple,
and undirected.

\subsubsection{Metric properties and canonical paths}

We now argue that $x'$ is a yes-instance of \ESP\ if and only if at
least one input formula $\Phi_j$ is satisfiable. We first establish three
structural properties of the constructed graph.

\begin{lemma}
\label{lem:mesp-metric}
The selector $\mathcal S$ and the union $J$ of all test paths $T_C$ are
isometric in $G$. Every shortest $a$--$b$ path is an assignment path of
length $2K$. Moreover,
\begin{equation}
\label{eq:mesp-terminal-distances}
\begin{aligned}
  \dist_G(a,b)&=2K, & \dist_G(u,v)&=8t,\\
  \dist_G(a,u)&=A,  & \dist_G(b,v)&=A,\\
  \dist_G(u,b)&=K+6t, & \dist_G(v,a)&=K+4t-2.
\end{aligned}
\end{equation}
\end{lemma}

\begin{proof}
Let $P$ be a shortest path in $G$ whose endpoints lie in $\mathcal S$.
Suppose that $P$ is not contained in $\mathcal S$. Then $P$ contains a
maximal subpath whose endpoints lie in $\mathcal S$ and whose internal
vertices lie outside $\mathcal S$. This subpath traverses at least two
connector, guard, or incidence paths. Each of these paths has one of the
lengths $A=K+2, D=K+t, K+6t-1$, and $K+4t-3$, all strictly larger than
$K$. Hence the subpath has length greater than $2K$. Its endpoints can
instead be joined within $\mathcal S$ by a path of length at most
$\diam(\mathcal S)\le2K$. This replacement yields a shorter walk with
the same endpoints, from which a shorter path is obtained by deleting
cycles, contradicting the choice of $P$. Thus $\mathcal S$ is
isometric. In particular, every shortest $a$--$b$ path is contained in
$\mathcal S$ and therefore is an assignment path of length $2K$.

Similarly, let $J$ be the union of all test paths. A shortest path with
endpoints in $J$ cannot enter a pendant path and return. Hence, if it
leaves $J$, the preceding argument produces a subpath of length greater
than $2K$ whose endpoints can be joined within $J$ by a path of length
at most $\diam(J)\le8t<2K$, a contradiction. Thus $J$ is isometric. In
particular,
$\dist_G(u,v)=8t$, and the distances in $G$ from the position-$i$
vertex of any $T_C$ to $u$ and $v$ are $i$ and $8t-i$, respectively.

It remains to determine the four distances between $\{u,v\}$ and
$\{a,b\}$. Consider a shortest such path and its first entry into
$\mathcal S$. Once it enters $\mathcal S$, it remains there by
isometry. We distinguish the following cases.
\begin{enumerate}
  \item If the path first enters at a literal vertex, then it spends at
  least $4t$ edges reaching a clause midpoint, $D$ edges on an incidence
  path, and at least $s+1$ edges inside $\mathcal S$ before reaching
  $a$ or $b$. Its length is at least $4t+D+s+1=K+9t+1$, which exceeds all four claimed distances.

  \item The connector paths give routes of length $A$ from $u$ to $a$
  and from $v$ to $b$. A route using the opposite connector already
  spends more than $A$ edges before traversing $\mathcal S$, and every
  guard route has length greater than $A$. Hence $\dist_G(a,u)=\dist_G(b,v)=A$.

  \item The guards followed by the adjacent selector edges give routes
  of lengths $K+6t$ from $u$ to $b$ and $K+4t-2$ from $v$ to $a$.
  In both cases, the connector alternatives have lengths $A+2K$ and
  $8t+A$, while routes using the opposite guard are longer. Therefore, $\dist_G(u,b)=K+6t$ and $\dist_G(v,a)=K+4t-2$.
\end{enumerate}
\end{proof}

For an assignment path $Q$, and integers $0\le x,y\le A$, let
$P(Q,x,y)$ be the path obtained by extending $Q$ a distance $x$ from
$a$ along its connector path toward $u$, and a distance $y$ from $b$
along its connector path toward $v$.

\begin{lemma}
\label{lem:mesp-canonical}
If $(G,R)$ is a yes-instance, then it has a solution path of the form
$P(Q,x,y)$.
For such a path, shortestness is equivalent to
\begin{equation}
\label{eq:mesp-shortestness-bounds}
  x\le3t+1,\qquad y\le2t,\qquad x+y\le4t+2.
\end{equation}
\end{lemma}

\begin{proof}

Every vertex outside the pendant path attached at $a$ is at distance
greater than $R$ from its pendant vertex; the same holds at $b$. Hence
any solution path must contain a vertex on each of these two pendant
paths. Since they are attached to the graph only at $a$ and $b$,
respectively, the solution path contains both $a$ and $b$.
Its $a$--$b$ subpath is therefore an assignment path $Q$ by
Lemma~\ref{lem:mesp-metric}. An extension of $Q$ at $a$ can enter only the forcing tail or the
connector path towards $u$. The extension cannot reach $u$, since its
$u,b$-subpath would have length $A+2K$, which exceeds
$\dist_G(u,b)=K+6t$, contradicting the shortestness of the solution
path. A similar argument holds at $b$. Any extension of the solution
path into a forcing tail may be truncated at $a$ or $b$, respectively.
This preserves shortestness and the eccentricity bound: every vertex
outside the relevant forcing tail does not become farther from the
path, while every vertex on that tail remains within distance $R$ of
the path. The resulting path is $P(Q,x,y)$ for some integers $x$ and
$y$ with $0\le x\le A$ and $0\le y\le A$.

Let $p_x$ and $q_y$ be the endpoints of $P(Q,x,y)$ on the two connector
paths. Since the internal vertices of these paths have degree two, every
$p_x$--$q_y$ path leaves each connector path through one of its
endpoints. Between these endpoints, the four possibilities use $Q$, a
shortest $u$--$b$ path, a shortest $a$--$v$ path, or a test path from
$u$ to $v$. Consequently,
\begin{align*}
  \dist_G(p_x,q_y)=\min\{&x+2K+y,\quad
    (A-x)+(K+6t)+y,\\
    &x+(K+4t-2)+(A-y),\quad
    (A-x)+8t+(A-y)\}.
\end{align*}
The path $P(Q,x,y)$ realizes the first term, so it is shortest exactly
when this term is no larger than each of the other three. Simplifying
these inequalities gives~\eqref{eq:mesp-shortestness-bounds}.
\end{proof}

\begin{lemma}
\label{lem:mesp-tests}
Let $Q$ be an assignment path that encodes an assignment $\beta$. For
every clause $C\in\mathcal C$ and every $j\in[t]$,
\begin{equation}
\label{eq:mesp-literal-test}
\begin{cases}
  \dist_G(w_{C,j},V(Q))\le r,&\text{if $\beta$ satisfies $C$},\\
  \dist_G(w_{C,j},V(Q))>r,&\text{if $\beta$ falsifies $C$}.
\end{cases}
\end{equation}
Every nonpendant vertex is within distance $R$ of $Q$. Furthermore, for
$P=P(Q,x,y)$ and $i=4t+j$,
\begin{equation}
\label{eq:mesp-distance-to-canonical-path}
  \dist_G(w_{C,j},V(P))=
  \min\{\dist_G(w_{C,j},V(Q)),\ i+A-x,\ 8t-i+A-y\}.
\end{equation}
\end{lemma}

\begin{proof}
Each selected literal vertex is on $Q$, whereas each unselected
literal vertex has distance exactly $s=4t$ to $Q$. This holds in the
selector and also in $G$ by Lemma~\ref{lem:mesp-metric}. Fix a clause $C\in\mathcal C$ and an index $j\in[t]$, and set
$i=4t+j$.

  \begin{enumerate}
    \item Suppose that $\beta$ satisfies $C$. A selected literal of $C$
    yields a path from $w_{C,j}$ to $Q$ of length $
      j+D=K+t+j\le K+2t=r$.

    \item Suppose that $\beta$ falsifies $C$. All three literal vertices
    of $C$ are then unselected. Consider how a shortest path from
    $w_{C,j}$ to $Q$ first enters $\mathcal S$.
    \begin{enumerate}
      \item If it uses an incidence path of $C$, then its length is at
      least $j+D+s=K+5t+j>r$.

      \item If it uses an incidence path of another clause, then it must
      first reach that clause's midpoint without entering $\mathcal S$.
      This requires at least $\min\{(4t+j)+4t,(4t-j)+4t\}=8t-j\ge7t$
      edges, and the incidence path increases the length to at least
      $7t+D=K+8t>r$.

      \item If it uses a connector path or a guard path, then its length
      is at least one of $i+A, 8t-i+A, i+(K+6t-1)$ and $8t-i+(K+4t-3)$.
      Each quantity exceeds $r=K+2t$ for $1\le j\le t$.
    \end{enumerate}

  \end{enumerate}
    These subcases exhaust all possible first entries into
    $\mathcal S$, proving~\eqref{eq:mesp-literal-test}. A selector vertex is within distance $s$ of $Q$. A connector-path vertex is within distance $A$ of $a$ or $b$. A guard-path vertex is within the corresponding guard length of $z_0$ or $z_n$. An incidence-path vertex is within distance $D+s$ of $Q$. A test-path vertex reaches $Q$ through its midpoint and an
    incident literal in at most $4t+D+s=K+9t<R$ edges. All these bounds are at most $R$. A shortest path from $w_{C,j}$
  to $P(Q,x,y)$ ends either on $Q$ or on one of the two connector-path
  prefixes.
  \begin{enumerate}
    \item If it ends on $Q$, its length is $\dist_G(w_{C,j},V(Q))$.
    \item A path to the prefix on the $a$--$u$ connector enters through
    $a$ or $u$. Entering through $a$ already reaches $Q$, whereas
    entering through $u$ gives length $\dist_G(w_{C,j},u)+A-x=i+A-x$.
    \item Similarly, entering the other connector prefix through $v$
    gives length $8t-i+A-y$.
  \end{enumerate}
  Here the distances to $u$ and $v$ follow from the isometry of $J$ in
  Lemma~\ref{lem:mesp-metric}. Taking the minimum of these proves~\eqref{eq:mesp-distance-to-canonical-path}.

\end{proof}

\subsubsection{Correctness of the cross-composition}

\begin{lemma}
\label{lem:mesp-equivalence}
The constructed instance $(G,R)$ is a yes-instance if and only if at
least one of $\Phi_1,\ldots,\Phi_t$ is satisfiable.
\end{lemma}

\begin{proof}
Suppose that $(G,R)$ is a yes-instance. By
Lemma~\ref{lem:mesp-canonical}, it has a solution path
$P=P(Q,x,y)$. Let $\beta$ be the assignment encoded by $Q$, and define
$L_0=x+2t-2$ and $U_0=6t+2-y$.
For a position $i$ on a test path, the left connector-path term
$i+A-x$ is at most $r$ exactly when $i\le L_0$, and the right
connector-path term $8t-i+A-y$ is at most $r$ exactly when
$i\ge U_0$. The inequalities in~\eqref{eq:mesp-shortestness-bounds}
imply $L_0\le5t-1$,
$U_0\ge4t+2$, and $U_0-L_0\ge2$.
Hence there is an integer $i\in\{4t+1,\ldots,5t\}$ with $L_0<i<U_0$;
for example, take $i=\max\{4t+1,L_0+1\}$. Write $i=4t+j$ with $j\in[t]$.
This index is the same for every clause test path. Suppose that a clause $C\in\Phi_j$ is falsified by $\beta$.
Lemma~\ref{lem:mesp-tests} gives $\dist_G(w_{C,j},V(Q))>r$, while
$L_0<i<U_0$ makes both connector-path terms
in~\eqref{eq:mesp-distance-to-canonical-path} greater than $r$.
Therefore $\dist_G(w_{C,j},V(P))>r$. The path $P$ does not enter a test
hair, so $\dist_G(\lambda_{C,j},V(P))
  =8t+\dist_G(w_{C,j},V(P))>8t+r=R$,
contradicting that $P$ is a solution path. Thus $\beta$ satisfies every
clause of $\Phi_j$, and $\Phi_j$ is satisfiable.

Suppose that $\Phi_j$ is satisfied by an assignment $\beta$, and let
$Q$ be the corresponding assignment path. Set $x=2t+j+1$ and
$y=2t-j+1$.
These values satisfy~\eqref{eq:mesp-shortestness-bounds}, with
$x+y=4t+2$, so $P=P(Q,x,y)$ is shortest. Its length is
$2K+4t+2=2r+2$. Every nonpendant vertex and both forcing tails are
within distance $R$ of $P$, by Lemma~\ref{lem:mesp-tests} and because
$a,b\in P$. Consider a test hair attached at $w_{C,h}$, where $h\in[t]$ and $C\in\Phi_h$. The two connector-path terms
in~\eqref{eq:mesp-distance-to-canonical-path} become
$(4t+h)+A-x=r+h-j+1$ and $(4t-h)+A-y=r+j-h+1$. If $h\ne j$, at least
one of these terms is at most $r$. If $h=j$, then $\beta$ satisfies $C$
because $C\in\Phi_j$, and~\eqref{eq:mesp-literal-test} gives
$\dist_G(w_{C,j},V(Q))\le r$. Thus every attachment vertex $w_{C,h}$ of a
test hair is within distance $r$ of $P$. Since each test hair has length
$8t$ and $r+8t=R$, all its vertices are within distance $R$ of $P$.
Therefore $P$ is a solution path, and $(G,R)$ is a yes-instance.
\end{proof}

\begin{proof}[Proof of Theorem~\ref{thm:mesp-fen-kernel-lower-bound}]
The construction can be performed in polynomial time, and
equation~\eqref{eq:mesp-feedback-edge-number} bounds the output parameter
by a polynomial in the maximum input size plus $\log t$.
Lemma~\ref{lem:mesp-equivalence} establishes the required OR-property.
Consequently, the construction is an OR-cross-composition of
\textsc{Distinct 3-Sat} into \ESP\ parameterized by feedback edge
number. Theorem~\ref{thm:cross-composition-framework} therefore yields
the claimed kernel lower bound. A minimum feedback edge set can be supplied by computing a spanning
tree and taking the edges outside it, without changing the parameter
or the equivalence.
\end{proof}

\providecommand{\ESP}{\textnormal{\textsc{Eccentricity Shortest Path}}}
\providecommand{\SUSC}{\textnormal{\textsc{Small Universe Set Cover}}}
\providecommand{\dist}{\operatorname{dist}}
\providecommand{\ecc}{\operatorname{ecc}}
\providecommand{\vc}{\operatorname{vc}}
\providecommand{\NP}{\mathsf{NP}}
\providecommand{\coNP}{\mathsf{coNP}}
\providecommand{\poly}{\mathrm{poly}}
\subsection{Vertex cover number}
\label{sec:mesp-vc-kernel-lower-bound}

We transfer the kernel lower bound for \SUSC\ to \ESP\ parameterized by
vertex cover and related vertex-deletion distances. Recall that $\vc(G)$ is
the minimum size of a vertex set meeting every edge, as defined in
Subsection~\ref{sec:graph-parameters}.

\begin{theorem}
\label{thm:mesp-vc-kernel-lower-bound}
The problem \ESP\ parameterized by vertex cover number has no
polynomial kernel unless $\NP\subseteq \coNP/\poly$. It also holds when the
input includes a vertex cover and is parameterized by its size.
\end{theorem}

\subsubsection{Source problem and transformation framework}

We use a polynomial parameter transformation from \SUSC. An instance of
\SUSC\ consists of a universe $U$, a family $\mathcal F\subseteq2^U$,
and an integer $b$. It is a yes-instance if at most $b$ members of
$\mathcal F$ cover $U$. The parameter is $|U|+b$. This problem has no
polynomial kernel unless $\NP\subseteq\coNP/\poly$~\cite{DLS14}.

Let $P$ and $Q$ be parameterized problems. A \emph{polynomial parameter
transformation} from $P$ to $Q$ is a polynomial-time algorithm that maps
an instance $(x,k)$ of $P$ to an instance $(x',k')$ of $Q$ such that
$(x,k)$ is a yes-instance if and only if $(x',k')$ is a yes-instance,
and $k'$ is bounded by a polynomial in $k$.

\begin{theorem}[Bodlaender, Thomass\'e, and Yeo~\cite{BTY11}]
\label{thm:mesp-vc-transformation-framework}
Let $P$ and $Q$ be parameterized problems such that the unparameterized
version of $P$ is NP-hard and the unparameterized version of $Q$ belongs
to NP. If $P$ admits no polynomial kernel and there is a polynomial
parameter transformation from $P$ to $Q$, then $Q$ admits no polynomial
kernel.
\end{theorem}

\subsubsection{Graph construction}

Let $(U,\mathcal F,b)$ be an instance of \SUSC, and set $n:=|U|$. We
may assume that
\begin{equation}
\label{eq:mesp-vc-normalization}
  U\ne\varnothing,\qquad
  \bigcup_{F\in\mathcal F}F=U,\qquad
  1\le b\le n.
\end{equation}
Indeed, an empty universe gives a yes-instance. An element that occurs in
no set, or a zero budget with a nonempty universe, gives a no-instance.
If every element occurs in some set and $b>n$, then choosing one set
for each element gives a cover of size at most $b$. These cases can be
mapped to fixed yes- and no-instances of \ESP, with fixed vertex covers
when required. Hence assume~\eqref{eq:mesp-vc-normalization} in the
remainder.

Set $L=b$ and $R=b+1$.
We construct an instance $(G,R)$ of \ESP\ as follows (see Figure~\ref{fig:mesp-vc-construction}).

\begin{enumerate}
  \item Create connector vertices $z_0,\ldots,z_b$. For every
  $i\in[b]$, create an independent candidate layer $I_i=\{c_{i,F}:F\in\mathcal F\}$.
  Join every candidate in $I_i$ to both $z_{i-1}$ and $z_i$. There are
  no edges between candidate vertices, including candidates in distinct
  layers. Let
  $B=\{z_0,\ldots,z_b\}\cup\bigcup_{i=1}^b I_i$, and call the subgraph
  induced by $B$ the \emph{selection backbone}.
  For a sequence $F_1,\ldots,F_b\in\mathcal F$, define the corresponding
  \emph{selection path} as $z_0,c_{1,F_1},z_1,c_{2,F_2},\ldots,z_{b-1},c_{b,F_b},z_b$. Here $F_i$ denotes the set selected in layer $I_i$, not the $i$th
  member of a fixed ordering of $\mathcal F$. The sets
  $F_1,\ldots,F_b$ need not be distinct.

  \item Attach a path $z_0=a_0,a_1,\ldots,a_R$ at $z_0$, and a path $z_b=d_0,d_1,\ldots,d_R$ at $z_b$. These paths are the \emph{forcing tails}.

  \item For every element $u\in U$, create a terminal vertex $x_u$.
  For every $i\in[b]$, create a hub $h_{i,u}$ and a path $Q_{i,u}$ of
  length $L=b$ from $h_{i,u}$ to $x_u$. For a fixed element $u$, these
  paths share only the terminal $x_u$; paths belonging to distinct
  elements are disjoint. We call the subgraph formed by $x_u$, the hubs
  $h_{i,u}$, and the paths $Q_{i,u}$ the \emph{element gadget} for $u$.
  Finally, join $c_{i,F}$ to $h_{i,u}$ exactly when $u\in F$.
\end{enumerate}

Each labeled line in Figure~\ref{fig:mesp-vc-construction} represents a
path of the indicated number of edges, with its internal subdivision
vertices suppressed. The drawing is not to scale, and unmarked
crossings are not vertices.

\begin{figure}[!htbp]
\centering
\resizebox{\linewidth}{!}{%
\begin{tikzpicture}[
  x=1.025cm,y=0.9cm,font=\small,
  vertex/.style={circle,fill=black,inner sep=1.7pt},
  connector/.style={vertex},
  candidate/.style={vertex},
  gadget/.style={vertex},
  backbone/.style={thick},
  chosen/.style={very thick,blue!65!black},
  omitted/.style={thick,densely dotted},
  incidence/.style={thick,violet!80!black},
  element/.style={thick,teal!70!black},
  force/.style={thick},
  labelbox/.style={fill=white,text=black,inner sep=1.5pt,font=\footnotesize},
  every label/.style={fill=white,text=black,inner sep=1pt}
]
\node[connector,label=below left:$z_0$] (z0) at (0,0) {};
\node[candidate] (c1u) at (1.2,0.92) {};
\node[candidate] (c1um) at (1.2,0.46) {};
\node[candidate] (c1) at (1.2,0) {};
\node[font=\scriptsize,inner sep=0pt] (c1dots) at (1.2,-0.5) {$\vdots$};
\node[candidate] (c1d) at (1.2,-0.92) {};
\node[connector,label=below:$z_1$] (z1) at (2.4,0) {};
\node[candidate] (omitLeftUpper) at (3.75,0.92) {};
\node[candidate] (omitLeftUpperMid) at (3.75,0.46) {};
\node[candidate] (omitLeftChosen) at (3.75,0) {};
\node[font=\scriptsize,inner sep=0pt] (omitLeftDots) at (3.75,-0.5) {$\vdots$};
\node[candidate] (omitLeftLower) at (3.75,-0.92) {};

\node[connector,label=below:$z_{i-1}$] (zim) at (5.1,0) {};
\node[candidate] (ciu) at (6.3,0.92) {};
\node[candidate] (cium) at (6.3,0.46) {};
\node[candidate] (ci) at (6.3,0) {};
\node[font=\scriptsize,inner sep=0pt] (cidots) at (6.3,-0.5) {$\vdots$};
\node[candidate] (cid) at (6.3,-0.92) {};
\node[connector,label=below:$z_i$] (zi) at (7.5,0) {};
\node[candidate] (omitRightUpper) at (8.85,0.92) {};
\node[candidate] (omitRightUpperMid) at (8.85,0.46) {};
\node[candidate] (omitRightChosen) at (8.85,0) {};
\node[font=\scriptsize,inner sep=0pt] (omitRightDots) at (8.85,-0.5) {$\vdots$};
\node[candidate] (omitRightLower) at (8.85,-0.92) {};

\node[connector,label=below:$z_{b-1}$] (zbm) at (10.2,0) {};
\node[candidate] (cbu) at (11.4,0.92) {};
\node[candidate] (cbum) at (11.4,0.46) {};
\node[candidate] (cb) at (11.4,0) {};
\node[font=\scriptsize,inner sep=0pt] (cbdots) at (11.4,-0.5) {$\vdots$};
\node[candidate] (cbd) at (11.4,-0.92) {};
\node[connector,label=below right:$z_b$] (zb) at (12.6,0) {};

\foreach \c in {c1u,c1um,c1,c1d} {
  \draw[backbone] (z0) -- (\c);
  \draw[backbone] (\c) -- (z1);
}
\foreach \c in {ciu,cium,ci,cid} {
  \draw[backbone] (zim) -- (\c);
  \draw[backbone] (\c) -- (zi);
}
\foreach \c in {cbu,cbum,cb,cbd} {
  \draw[backbone] (zbm) -- (\c);
  \draw[backbone] (\c) -- (zb);
}
\foreach \candidateVertex in {omitLeftUpper,omitLeftUpperMid,
  omitLeftChosen,omitLeftLower} {
  \draw[omitted] (z1) -- (\candidateVertex) -- (zim);
}
\foreach \candidateVertex in {omitRightUpper,omitRightUpperMid,
  omitRightChosen,omitRightLower} {
  \draw[omitted] (zi) -- (\candidateVertex) -- (zbm);
}
\draw[omitted] (z0) -- (c1dots.west);
\draw[omitted] (c1dots.east) -- (z1);
\draw[omitted] (z1) -- (omitLeftDots.west);
\draw[omitted] (omitLeftDots.east) -- (zim);
\draw[omitted] (zim) -- (cidots.west);
\draw[omitted] (cidots.east) -- (zi);
\draw[omitted] (zi) -- (omitRightDots.west);
\draw[omitted] (omitRightDots.east) -- (zbm);
\draw[omitted] (zbm) -- (cbdots.west);
\draw[omitted] (cbdots.east) -- (zb);
\node[labelbox] at (3.75,-1.4) {$\cdots$};
\node[labelbox] at (8.85,-1.4) {$\cdots$};

\node[draw=black,rounded corners,fit=(c1u)(c1d),
  inner sep=6pt,label={[xshift=7pt]above:$I_1$}] {};
\node[draw=black,densely dotted,rounded corners,
  fit=(omitLeftUpper)(omitLeftLower),inner sep=6pt] {};
\node[draw=black,rounded corners,fit=(ciu)(cid),
  inner sep=6pt,label={[xshift=7pt]above:$I_i$}] {};
\node[draw=black,densely dotted,rounded corners,
  fit=(omitRightUpper)(omitRightLower),inner sep=6pt] {};
\node[draw=black,rounded corners,fit=(cbu)(cbd),
  inner sep=6pt,label={[xshift=7pt]above:$I_b$}] {};

\draw[chosen] (z0) -- (c1) -- (z1);
\draw[chosen,densely dashed] (z1) -- (omitLeftChosen) -- (zim);
\draw[chosen] (zim) -- (ci) -- (zi);
\draw[chosen,densely dashed] (zi) -- (omitRightChosen) -- (zbm);
\draw[chosen] (zbm) -- (cb) -- (zb);
\node[labelbox] at (6.3,-1.75)
  {one selection path};

\node[gadget,label=left:$h_{1,u}$] (h1) at (1.2,2.55) {};
\node[gadget,label=left:$h_{i,u}$] (hi) at (6.3,2.55) {};
\node[gadget,label=right:$h_{b,u}$] (hb) at (11.4,2.55) {};
\node[gadget,label=above:$x_u$] (xu) at (6.3,4.65) {};
\draw[element] (h1) to[bend left=11]
  node[labelbox,pos=0.47,above left] {$Q_{1,u}$: length $L=b$} (xu);
\draw[element] (hi) -- node[labelbox,right] {$L=b$} (xu);
\draw[element] (hb) to[bend right=11]
  node[labelbox,pos=0.47,above right] {$Q_{b,u}$: length $L=b$} (xu);
\node[font=\footnotesize] at (6.3,5.3)
  {one element gadget for each $u\in U$};

\draw[incidence] (c1) to[bend left=18]
  node[labelbox,pos=0.72,left] {$u\in F$} (h1);
\draw[incidence] (ciu) -- node[labelbox,pos=0.72,right] {$u\in H$} (hi);
\draw[incidence] (cbu) -- node[labelbox,pos=0.72,right] {$u\in H'$} (hb);

\node[labelbox,anchor=north] at ([yshift=-2pt]c1.south) {$c_{1,F}$};
\node[labelbox,anchor=north] at ([yshift=-2pt]ci.south) {$c_{i,F'}$};
\node[labelbox,anchor=north] at ([yshift=-2pt]cb.south) {$c_{b,F''}$};

\node[connector] (aR) at (0,-3.3) {};
\node[connector] (dR) at (12.6,-3.3) {};
\draw[force] (z0) -- node[labelbox,left] {$R=b+1$} (aR);
\draw[force] (zb) -- node[labelbox,right] {$R=b+1$} (dR);
\node[font=\footnotesize,below] at (aR) {forcing tail};
\node[font=\footnotesize,below] at (dR) {forcing tail};
\end{tikzpicture}%
}
\caption{A sketch of the graph construction used in the proof of
Theorem~\ref{thm:mesp-vc-kernel-lower-bound}.}
\label{fig:mesp-vc-construction}
\end{figure}
\FloatBarrier

\subsubsection{Parameter and size bounds}

Let $X=V(G)\setminus\bigcup_{i=1}^b I_i$.
The graph $G-X$ is independent, so $X$ is a vertex cover. Each path
$Q_{i,u}$ contributes $b$ vertices other than its terminal, which is
counted once for each $u$. Consequently,
\begin{equation}
\label{eq:mesp-vc-cover-size}
\begin{aligned}
  |X|&=(b+1)+2R+n+nb^2\\
     &=nb^2+n+3b+3
      =O\bigl({(n+b)}^3\bigr).
\end{aligned}
\end{equation}
The four terms in $|X|$ count, respectively, the connector vertices,
the vertices introduced on the two forcing tails, the element terminals,
and the nonterminal vertices of the paths $Q_{i,u}$.
The bound is independent of $|\mathcal F|$ and polynomial in the source
parameter $n+b$.
If a vertex cover must be supplied with the output, we supply $X$. The graph $G$ can be constructed in polynomial time. Let $\ell$ denote
the encoding length of $(U,\mathcal F,b)$. Then
$n,b,|\mathcal F|\le \ell$. Besides the vertices of $X$, the graph has
$b|\mathcal F|$ candidate vertices. The selection backbone contributes
$2b|\mathcal F|$ edges, the forcing tails contribute $2R$ edges, and the
paths $Q_{i,u}$ contribute $nb^2$ edges. There are at most
$nb|\mathcal F|$ candidate-hub edges. Hence $G$ has
$O(nb^2+b|\mathcal F|)=O(\ell^3)$ vertices and
$O(nb^2+nb|\mathcal F|)=O(\ell^3)$ edges. Thus the output size and the
running time are polynomial in the source input length. The algorithm
expands every path explicitly, and therefore $G$ is unweighted. By
Equation~\eqref{eq:mesp-vc-normalization}, every element gadget has an
edge to the selection backbone, so $G$ is connected. The description
also shows that $G$ is simple and undirected.

\subsubsection{Canonical paths and correctness}

We now argue that $(G,R)$ is a yes-instance of \ESP\ if and only if
$(U,\mathcal F,b)$ is a yes-instance of \SUSC. We first establish two
structural properties of the constructed graph.

\begin{lemma}
\label{lem:mesp-vc-canonical}
The shortest $z_0$--$z_b$ paths are exactly the selection paths, and
each has length $2b$. Moreover, if $(G,R)$ is a yes-instance, then some
selection path is a solution path.
\end{lemma}

\begin{proof}
Every selection path has length $2b$, so a shortest $z_0$--$z_b$ path
has length at most $2b$. Let $P$ be such a shortest path. Since each
forcing tail is pendant at its attachment vertex, $P$ cannot enter its
interior. Suppose that $P$ contains an element terminal $x_u$ or an
internal vertex of a path $Q_{i,u}$. Since the endpoints of $P$ lie in
the selection backbone $B$, the path $P$ contains a maximal subpath
whose endpoints lie in $B$, whose internal vertices lie outside $B$,
and which contains this vertex. This subpath enters and leaves the
element gadget through candidate--hub edges. Since every internal vertex
of a path $Q_{i,u}$ has degree two, returning to $B$ without repeating a
vertex requires traversing two distinct length-$b$ paths $Q_{i,u}$ and
$Q_{j,u}$ through $x_u$. Together with the two candidate--hub edges,
this subpath has length at least $2b+2$, contradicting the length of
$P$. Hence $P$ contains no element terminal or internal vertex of a path
$Q_{i,u}$. After the terminals and the internal vertices of the paths $Q_{i,u}$ are excluded, every hub is adjacent only to candidates from one layer. Thus, for every $i\in[b-1]$, the connector $z_i$ separates the layers $I_1,\ldots,I_i$ from the layers $I_{i+1},\ldots,I_b$. Any
$z_0$--$z_b$ path must therefore visit
$z_0,z_1,\ldots,z_b$ in this order. Between $z_{i-1}$ and $z_i$, it
uses at least two edges, with equality exactly when it uses one candidate of $I_i$. Since the total length is at most $2b$, equality holds for every layer. Hence the shortest $z_0$--$z_b$ paths are precisely the selection paths.

Suppose that $P$ is a solution path. Let
$A_0=\{a_1,\ldots,a_R\}$ and $A_b=\{d_1,\ldots,d_R\}$. Every vertex
outside $A_0\cup\{z_0\}$ is at distance greater than $R$ from $a_R$.
If $z_0\notin V(P)$, then the eccentricity bound forces $P$ to contain a
vertex of $A_0$. Since $z_0$ separates $A_0$ from the rest of $G$ and
$P$ is connected, this implies $V(P)\subseteq A_0$. But then
$\dist_G(d_R,V(P))>R$, a contradiction. Thus $z_0\in V(P)$, and the
symmetric argument gives $z_b\in V(P)$. The $z_0$--$z_b$ subpath of $P$ is shortest between its endpoints and
is therefore a selection path $S$. Orient $P$ so that this subpath runs
from $z_0$ to $z_b$. If the predecessor of $z_0$ on $P$ belonged to
$I_1$, then it would be adjacent to $z_1$. This would replace the
three-edge subpath consisting of that predecessor, $z_0$, the selected
candidate of $I_1$, and $z_1$ by one edge, contradicting the
shortestness of $P$. Therefore any extension before $z_0$ lies in
$A_0$. Symmetrically, any extension after $z_b$ lies in $A_b$. Trim these extensions. For a vertex outside $A_0\cup A_b$, any path to
one of the trimmed extensions first reaches $z_0$ or $z_b$, both of
which remain on $S$. Hence its distance to $S$ equals its distance to
$P$. Every forcing-tail vertex is within distance $R$ of its attachment
vertex. It follows that $\ecc_G(S)\le R$, so $S$ is a solution path.
\end{proof}

\begin{lemma}
\label{lem:mesp-vc-element-distances}
Let $P$ be a selection path that chooses
$F_1,\ldots,F_b$. For every element $u\in U$,
\begin{equation}
\label{eq:mesp-vc-terminal-distance}
  \dist_G(x_u,V(P))=
  \begin{cases}
    b+1=R,   &\text{if }u\in\bigcup_{i=1}^b F_i,\\
    b+2=R+1, &\text{otherwise}.
  \end{cases}
\end{equation}
Every nonterminal vertex is within distance $R$ of $P$. Consequently,
$P$ is a solution path if and only if $F_1,\ldots,F_b$ cover $U$.
\end{lemma}

\begin{proof}
Every connector lies on $P$, every candidate is adjacent to a connector,
and every forcing-tail vertex is within distance $R$ of its attachment
vertex. Consider a vertex
$y\in V(Q_{i,u})\setminus\{x_u\}$, and let $r$ be its distance from
$h_{i,u}$ along $Q_{i,u}$. Then $0\le r\le b-1$. By
Equation~\eqref{eq:mesp-vc-normalization}, there is a set
$F\in\mathcal F$ containing $u$. The route from $y$ through $h_{i,u}$
and $c_{i,F}$ to $z_{i-1}$ has length $r+2\le b+1=R$. Thus every
nonterminal vertex is within distance $R$ of $P$. Every path from $x_u$ to $P$ first traverses one of the length-$b$ paths
$Q_{i,u}$ to a hub $h_{i,u}$. The hub $h_{i,u}$ has distance one from
$P$ exactly when the selected set $F_i$ contains $u$. Otherwise its
distance from $P$ is two: choose any set $F\in\mathcal F$ containing
$u$, and use the path $h_{i,u},c_{i,F},z_{i-1}$. Taking the minimum of
$\dist_G(h_{i,u},V(P))$ over all layers proves
Equation~\eqref{eq:mesp-vc-terminal-distance}: it is one if some selected set $F_i$ contains $u$, and two otherwise.

The final claim follows because the vertices $x_u$ are the only vertices that can lie
farther than $R$ from $P$.
\end{proof}

\begin{lemma}
\label{lem:mesp-vc-equivalence}
The constructed instance $(G,R)$ is a yes-instance if and only if
$(U,\mathcal F,b)$ is a yes-instance of \SUSC.
\end{lemma}

\begin{proof}
Suppose that $(G,R)$ is a yes-instance. By
Lemma~\ref{lem:mesp-vc-canonical}, some selection path is a solution
path. Let $F_1,\ldots,F_b$ be its selected sets.
Lemma~\ref{lem:mesp-vc-element-distances} implies that these sets cover
$U$. Hence $(U,\mathcal F,b)$ is a yes-instance.

Suppose that at most $b$ sets from $\mathcal F$ cover $U$. Since $U$ is
nonempty, the cover contains at least one set. Repeat a selected set if
necessary to obtain a sequence $F_1,\ldots,F_b$. The corresponding
selection path is shortest by Lemma~\ref{lem:mesp-vc-canonical}, and
Lemma~\ref{lem:mesp-vc-element-distances} gives its eccentricity at most
$R$. Therefore $(G,R)$ is a yes-instance.
\end{proof}

\begin{proof}[Proof of Theorem~\ref{thm:mesp-vc-kernel-lower-bound}]
The construction is polynomial-time, and
Equation~\eqref{eq:mesp-vc-cover-size} bounds $|X|$ by a polynomial in
the source parameter $n+b$. Since $X$ is a vertex cover,
$\vc(G)\le |X|$, so the target parameter has the same polynomial bound.
Lemma~\ref{lem:mesp-vc-equivalence} establishes equivalence of the source
and target instances. The unparameterized version of \SUSC\ is
NP-complete, and \ESP\ belongs to NP because shortestness and
eccentricity of a given path can be checked in polynomial time.
Theorem~\ref{thm:mesp-vc-transformation-framework} and the kernel lower
bound for \SUSC\ therefore yield the claimed result. For the version with
a supplied vertex cover, we include $X$ in the output. For the minimum
vertex cover parameter, we omit $X$ and use $\vc(G)\le|X|$.
\end{proof}

\subsubsection{Consequences for deletion parameters}

For the following consequence, recall the vertex-deletion distances to
cluster, split, proper interval, and interval graphs from
Subsection~\ref{sec:graph-parameters}.

\begin{corollary}
\label{cor:mesp-vc-deletion-classes}
Unless $\NP\subseteq\coNP/\poly$, \ESP\ has no polynomial kernel
parameterized by deletion distance to cluster graphs, split graphs,
proper interval graphs, or interval graphs. The same holds when a
deletion set to any of these graph classes is supplied and its size is
the parameter.
\end{corollary}

\begin{proof}
Every edgeless graph belongs to each of the four graph classes.
Since $G-X$ is independent, every listed deletion parameter is at most
$|X|=O({(n+b)}^3)$. The construction therefore gives the same polynomial
parameter transformation for each parameter. If a deletion set is
required as part of the input, we supply $X$.
\end{proof}

\subsection{Bandwidth}

Recall that the bandwidth $\bw(G)$ minimizes the maximum distance
between the positions of adjacent vertices in an ordering; see
Subsection~\ref{sec:graph-parameters}. We use the notation for distances,
path eccentricity, and \ESP\ fixed in the preliminaries.

\begin{theorem}\label{thm:main}
Unless $\mathrm{NP}\subseteq\mathrm{coNP/poly}$, \ESP{} has no polynomial
kernel when parameterized by the bandwidth of the input graph.
\end{theorem}

\subsubsection{The cross-composition construction}

We give an OR-cross-composition from \textsc{3-SAT}, which is
NP-complete~\cite{Karp1972}. We use formulas with at least one variable and
one to three literals per clause. Two well-formed formulas are equivalent
when they have the same number of variables; all malformed strings form
one additional class. This is a polynomial equivalence relation: it is
polynomial-time decidable and has at most $N+1$ classes on strings of
length at most $N$. For the malformed class, output the fixed no-instance
$(K_3,0)$.

Let $F_1,\ldots,F_q$ be formulas in one well-formed class, each on variables
$X_1,\ldots,X_n$, where $n\ge1$. Write $m_i$ for the number of clauses of
$F_i$, and set
\begin{equation}\label{eq:constants}
 A=2n+1,\qquad L=A+1=2n+2,\qquad R=2A-1=4n+1.
\end{equation}
We construct an \ESP{} instance $(H,R)$.
A \emph{private path of length $\ell$} has $\ell$ edges and new internal
vertices. All private paths below have mutually disjoint interiors, and
no edges other than those specified are added.

\paragraph*{Construction.}
Create boundary vertices $a_0,b_0,\ldots,a_q,b_q$.
For each $i\in\{1,\ldots,q\}$, construct a module $M_i$ as follows.
Its \emph{upper rail} $T_i$ is a private $a_{i-1}$--$a_i$ path of length $L$;
its \emph{lower rail} $B_i$ is a private $b_{i-1}$--$b_i$ path of length $L$.
Let $x_i$ be the neighbor of $a_i$ on $T_i$, and let $y_i$ be the neighbor
of $b_{i-1}$ on $B_i$. Thus
\[
 \dist_{T_i}(a_{i-1},x_i)=A,
 \qquad
 \dist_{B_i}(y_i,b_i)=A.
\]

The \emph{variable chain} $D_i$ has new vertices
$w_{i,0},\ldots,w_{i,n}$ and $u_{i,j}^+,u_{i,j}^-$ for $1\le j\le n$.
Add the edges $a_{i-1}w_{i,0}$ and $w_{i,n}b_i$ and, for each $j$, the
four edges forming the two paths
\[
 w_{i,j-1},u_{i,j}^+,w_{i,j}
 \qquad\text{and}\qquad
 w_{i,j-1},u_{i,j}^-,w_{i,j}.
\]
The vertices $u_{i,j}^+$ and $u_{i,j}^-$ represent $X_j$ and $\neg X_j$,
respectively. Traversing $D_i$ from $a_{i-1}$ to $b_i$ by choosing one
literal vertex for each variable takes exactly $2n+2=L$ edges.

Add two \emph{helper paths}, each of length $L$: one from $a_{i-1}$ to
$y_i$ and one from $x_i$ to $b_i$.
For each clause occurrence $C$ of $F_i$, create a vertex $c_{i,C}$.
For each distinct literal occurring in $C$, add a path of length $R$
from $c_{i,C}$ to its corresponding vertex in $D_i$: use $u_{i,j}^{+}$
for the literal $X_j$ and $u_{i,j}^{-}$ for the literal $\neg X_j$.
Each such path is called a \emph{clause path}, and all its internal
vertices are new and used only by that path.
Also join $c_{i,C}$ to each of $x_i,y_i$ by a private path of length $A$,
called a \emph{hub link}.
The helper paths, clause paths, and hub links are the \emph{auxiliary paths}.
Different modules share only their common boundary vertices.

Finally, attach a pendant path of length $R+1$ at each of $a_0$ and $b_q$,
with respective far endpoints $\alpha$ and $\beta$.
These are the \emph{forcing tails}. The resulting graph $H$ is connected(see Figure~\ref{fig:bw-construction}).

Each labeled line in Figure~\ref{fig:bw-construction} represents a
path of the indicated number of edges, with its internal subdivision
vertices suppressed. The drawing is not to scale, and unmarked
crossings are not vertices.

\begin{figure}[!htbp]
\centering
\begin{tikzpicture}[
  x=0.95cm,y=0.85cm,font=\small,
  vertex/.style={circle,fill=black,inner sep=1.7pt},
  literalvertex/.style={circle,fill=black,inner sep=1.4pt},
  rail/.style={thick,gray!75!black},
  chosen/.style={very thick,blue!65!black},
  auxiliary/.style={thick,teal!70!black},
  helper/.style={thick,dashed,teal!70!black},
  clausepath/.style={thick,violet!80!black},
  labelbox/.style={fill=white,inner sep=1.5pt,font=\footnotesize},
  every label/.style={font=\scriptsize,fill=white,inner sep=1pt}
]
\node[vertex,label=above:$\alpha$] (alpha) at (0.1,8.4) {};
\node[vertex,label=above:$a_0$] (a0) at (1.0,8.4) {};
\node[vertex,label=below:$b_0$] (b0) at (1.0,3.2) {};
\node[vertex,label=above:$a_{i-1}$] (aLeft) at (3.2,8.4) {};
\node[vertex,label=below:$b_{i-1}$] (bLeft) at (3.2,3.2) {};
\node[vertex,label=above:$a_i$] (aRight) at (11.5,8.4) {};
\node[vertex,label=below:$b_i$] (bRight) at (11.5,3.2) {};
\node[vertex,label=above:$a_q$] (aq) at (13.5,8.4) {};
\node[vertex,label=below:$b_q$] (bq) at (13.5,3.2) {};
\node[vertex,label=above:$\beta$] (beta) at (14.4,3.2) {};
\node[vertex,label=above:$x_i$] (x) at (10.3,8.4) {};
\node[vertex,label=below:$y_i$] (y) at (4.7,3.2) {};

\node[vertex,label=below left:$w_{i,0}$] (w0) at (4.3,7.25) {};
\node[literalvertex,label=above:$u_{i,1}^+$] (u1p) at (5.24,7.45) {};
\node[literalvertex,label=below:$u_{i,1}^-$] (u1m) at (4.66,6.36) {};
\node[vertex,label=below:$w_{i,1}$] (w1) at (5.6,6.56) {};
\node[vertex,label=left:$w_{i,j-1}$] (wjm) at (6.9,5.87) {};
\node[literalvertex,label=above:$u_{i,j}^+$] (ujp) at (7.84,6.075) {};
\node[literalvertex,label=below:$u_{i,j}^-$] (ujm) at (7.26,4.975) {};
\node[vertex,label=below left:$w_{i,j}$] (wj) at (8.2,5.18) {};
\node[vertex,label=below right:$w_{i,n-1}$] (wnm) at (9.5,4.49) {};
\node[literalvertex,label=above:$u_{i,n}^+$] (unp) at (10.44,4.695) {};
\node[literalvertex,label=below:$u_{i,n}^-$] (unm) at (9.86,3.595) {};
\node[vertex,label=below:$w_{i,n}$] (wn) at (10.8,3.8) {};
\node[vertex,label=above:$c_{i,C}$] (clause) at (7.35,7.75) {};
\node[vertex] (omit1p) at (6.54,6.765) {};
\node[vertex] (omit1m) at (5.96,5.665) {};
\node[vertex] (omit2p) at (9.14,5.385) {};
\node[vertex] (omit2m) at (8.56,4.285) {};

\draw[chosen] (alpha) -- node[labelbox,pos=0.5,below] {$R+1$} (a0) -- (aLeft);
\draw[rail] (b0) -- (bLeft);
\node[labelbox] at (2.1,9.35) {$M_1,\ldots,M_{i-1}$};
\draw[rail] (aLeft) -- (x) -- (aRight);
\draw[rail] (bLeft) -- (y) -- (bRight);
\node[labelbox] at (7.4,9.75) {$M_i$};
\node[labelbox] at (6.7,8.75) {$T_i$: length $L$};
\node[labelbox] at (10.0,2.7) {$B_i$: length $L$};
\draw[rail] (aRight) -- (aq);
\draw[chosen] (bRight) -- (bq) -- node[labelbox,pos=0.35,above] {$R+1$} (beta);
\node[labelbox] at (12.5,9.35) {$M_{i+1},\ldots,M_q$};

\draw[helper] (aLeft) -- node[labelbox,pos=0.55,left] {$L$} (y);
\draw[helper] (x) -- node[labelbox,pos=0.55,right] {$L$} (bRight);
\draw[auxiliary] (clause) -- node[labelbox,pos=0.55,below] {$A$} (y);
\draw[auxiliary] (clause) -- node[labelbox,pos=0.55,above] {$A$} (x);

\draw[chosen] (aLeft) -- (w0);
\draw[rail] (w0) -- (u1m) -- (w1);
\draw[chosen] (w0) -- (u1p) -- (w1);
\node[labelbox] at (6.25,6.215) {$\cdots$};
\draw[rail,densely dotted] (w1) -- (omit1p) -- (wjm);
\draw[rail,densely dotted] (w1) -- (omit1m) -- (wjm);
\draw[chosen,densely dotted] (w1) -- (omit1p) -- (wjm);
\draw[rail] (wjm) -- (ujp) -- (wj);
\draw[chosen] (wjm) -- (ujm) -- (wj);
\node[labelbox] at (8.85,4.835) {$\cdots$};
\draw[rail,densely dotted] (wj) -- (omit2p) -- (wnm);
\draw[rail,densely dotted] (wj) -- (omit2m) -- (wnm);
\draw[chosen,densely dotted] (wj) -- (omit2m) -- (wnm);
\draw[chosen] (wnm) -- (unp) -- (wn);
\draw[rail] (wnm) -- (unm) -- (wn);
\draw[chosen] (wn) -- (bRight);
\node[labelbox] at (7.4,1.85) {variable chain $D_i$};

\draw[clausepath] (clause) -- node[labelbox,pos=0.6,left] {$R$} (u1p);
\draw[clausepath] (clause) -- (ujm);
\draw[clausepath] (clause) -- (unp);
\end{tikzpicture}
\caption{A sketch of the graph construction used in the proof of
Theorem~\ref{thm:main}.}
\label{fig:bw-construction}
\end{figure}

\subsubsection{Canonical shortest paths and domination}

\begin{claim}\label{clm:selector}
The distance from $a_0$ to $b_q$ is $qL$. Every shortest $a_0$--$b_q$ path
uses upper rails before a unique module $M_i$, traverses $D_i$ through
exactly one literal vertex for each variable, and uses lower rails
afterwards. Conversely, every such path is shortest.
\end{claim}

\begin{proof}
For each $i\in\{1,\ldots,q\}$, define a real-valued local level function
$\lambda_i$ on $V(M_i)$. Assign level $0$ to $a_{i-1},b_{i-1}$ and
level $L$ to $a_i,b_i$. Along each rail, increase the level by one
at each step from left to right. In particular,
$\lambda_i(x_i)=A$ and $\lambda_i(y_i)=1$. For the variable-chain
and clause vertices, set
\[
 \lambda_i(w_{i,j})=2j+1 \quad (0\le j\le n),\qquad
 \lambda_i(u_{i,j}^{\pm})=2j \quad (1\le j\le n),\qquad
 \lambda_i(c_{i,C})=n+1.
\]
Thus the levels at the endpoints of every rail or variable-chain
edge differ by exactly one.

It remains to assign levels to the internal vertices of auxiliary
paths. For an auxiliary path $z_0,z_1,\ldots,z_\ell$, whose endpoint
levels have already been assigned, set
\[
 \lambda_i(z_t)
 =\lambda_i(z_0)+\frac{t}{\ell}
       \bigl(\lambda_i(z_\ell)-\lambda_i(z_0)\bigr)
 \qquad (1\le t<\ell).
\]
This is well-defined because auxiliary paths have disjoint interiors.
The absolute level change across each edge of this path is therefore
$|\lambda_i(z_\ell)-\lambda_i(z_0)|/\ell$.
For a helper path, the endpoint levels differ by $1$ and its length
is $L$. For a hub link, they differ by $n$ and its length is $A$.
For a clause path ending at $u_{i,j}^{\pm}$, they differ by
$|2j-(n+1)|\le n-1$ and its length is $R$.
Since $1<L$, $n<A$, and $n-1<R$, every auxiliary edge changes the
level by strictly less than one. Consequently, every edge of $M_i$
changes the level by at most one.

Adding $(i-1)L$ to the levels in $M_i$ gives a well-defined global function
$\Lambda$, since the levels agree on shared boundaries. Any
$a_0$--$b_q$ path avoids the pendant tails and therefore has length at least
$\Lambda(b_q)-\Lambda(a_0)=qL$. A path of the stated form has length $qL$.
Equality forces every edge to increase $\Lambda$ by exactly one, so no
auxiliary edge can be used. The only possible transitions through a module
are its upper rail, its lower rail, or its variable chain, corresponding to
\[
 a\longrightarrow a,\qquad b\longrightarrow b,
 \qquad a\longrightarrow b.
\]
Starting at $a_0$ and ending at $b_q$ therefore requires exactly one
variable-chain transition. Along that chain, increasing levels force
exactly one literal choice per variable.
\end{proof}

Call this unique module \emph{selected}, and interpret its literal choices
as a truth assignment: $X_j$ is true precisely when $u_{i,j}^+$ is chosen.

\begin{claim}\label{clm:hubs}
If a shortest $a_0$--$b_q$ path $P$ selects $M_i$, then
\[
 \dist_H(x_i,P)=\dist_H(y_i,P)=A.
\]
\end{claim}

\begin{proof}
The rail segments to $a_{i-1}\in V(P)$ and $b_i\in V(P)$ give the respective
upper bounds $A$.
A path of length less than $A$ from either hub to $P$ cannot traverse an
auxiliary path: every such path has length at least $A$, its internal
vertices have degree two, and neither endpoint of the proposed short path
is an auxiliary internal vertex. Hence such a path would lie in the graph
consisting only of the rails, variable chains, and forcing tails.

In this graph, the part of the upper rail belonging to $P$ ends at
$a_{i-1}$, at rail distance $A$ from $x_i$. Before meeting $P$, a route
from $x_i$ that leaves the upper rail must traverse a nonselected variable
chain to the lower rail, which costs at least $L>A$ edges. It cannot enter
$D_i$ before meeting $P$, because both attachments of $D_i$ belong to $P$;
entering a forcing tail likewise requires first meeting $a_0$ or $b_q$.
Thus no $x_i$--$P$ path has length less than $A$.
The same argument on the lower rail, whose part in $P$ starts at $b_i$,
applies to $y_i$.
\end{proof}

\begin{claim}\label{clm:domination}
Let $P$ be a shortest $a_0$--$b_q$ path selecting $M_i$.
Every vertex of every nonselected module is within distance $R$ of $P$.
Every vertex of $M_i$ is within distance $R$ of $P$ if and only if its
selected assignment satisfies $F_i$.
\end{claim}

\begin{proof}
First consider the selected module $M_i$.
Each literal vertex is on $P$ or adjacent to a vertex $w_{i,j}$ on $P$.
A route from a clause vertex $c_{i,C}$ to $P$ must first traverse either
a hub link or a clause path. By Claim~\ref{clm:hubs}, going through a hub
costs at least $A+A=R+1$. A clause path already has length $R$, so it
reaches $P$ within distance $R$ exactly when its literal endpoint is chosen.
In fact,
\begin{equation}\label{eq:clause-test}
 \dist_H(c_{i,C},P)=
 \begin{cases}
 R,   &\text{if the selected assignment satisfies $C$},\\
 R+1, &\text{otherwise}.
 \end{cases}
\end{equation}
All other vertices of $M_i$ are automatically within distance $R$:
variable-chain vertices are within distance one, and rail vertices and
helper-path vertices are within distance $L$ via $a_{i-1}$ or $b_i$.
An internal clause-path vertex is at most $R-1$ edges from its literal
endpoint, hence at distance at most $R$ from $P$.
An internal hub-link vertex is at most $A-1$ edges from its hub, hence at
distance at most $(A-1)+A=R$ from $P$.
This proves the assertion for $M_i$.

Now consider a nonselected module, suppressing its index. By symmetry,
suppose $P$ uses its upper rail. Then $x$ and the upper-left boundary are
on $P$. Every nonboundary vertex of its variable chain is within distance
$A$ of that boundary, and every clause vertex is within distance $A$ of
$x$. Thus both endpoints of each clause path are within distance $A$ of
$P$. On a path of length $\ell$ whose endpoints are within distance $D$
of $P$, every vertex is within distance $D+\lfloor\ell/2\rfloor$ of $P$.
Consequently, all clause-path vertices are within distance
\[
 A+\left\lfloor\frac R2\right\rfloor=R.
\]
Each helper path has an endpoint on $P$, so all its vertices, including
$y$ and the lower-right boundary, are within distance $L$.
Every hub link has length $A$ and both endpoints within distance $L$;
the same holds for the lower-rail segment from $y$ to the lower-right
boundary. Their vertices are therefore within distance
\[
 L+\left\lfloor\frac A2\right\rfloor=3n+2\le4n+1=R.
\]
The lower-left boundary is adjacent to $y$, giving distance at most
$L+1=2n+3\le R$. This covers the whole module.
For the lower-rail case, interchange the two rails and use that every
nonboundary variable-chain vertex is also within distance $A$ of the
lower-right boundary.
\end{proof}

\subsubsection{Correctness of the cross-composition}

If $F_i$ is satisfiable, choose the path of Claim~\ref{clm:selector} that
selects $M_i$ and its satisfying assignment, and extend it along both
entire forcing tails. The extension is a shortest $\alpha$--$\beta$ path,
since the tails are pendant. Claim~\ref{clm:domination} dominates all module
vertices within distance $R$, and all tail vertices lie on the path.
Thus $(H,R)$ is a yes-instance.

Conversely, let $Q$ be a shortest path with $\ecc_H(Q)\le R$.
To come within distance $R$ of either forcing leaf, $Q$ must contain a
nonattachment vertex of its tail, because the tail is pendant and has
length $R+1$. Thus $Q$ meets both tails away from their attachments.
Because the tails are pendant, the endpoints of $Q$ lie in different tails,
and its central subpath $P$ is a shortest $a_0$--$b_q$ path.
For every module vertex $v$, we have $\dist_H(v,Q)=\dist_H(v,P)$:
any route to a tail first reaches its attachment, which already belongs
to $P$. Claim~\ref{clm:domination} now implies that the assignment in the
selected module satisfies its formula. Therefore
\[
 (H,R)\in\ESP
 \quad\Longleftrightarrow\quad
 \text{at least one of }F_1,\ldots,F_q\text{ is satisfiable}.
\]

\subsubsection{Bandwidth bound and completion of the proof}

\begin{proof}[Proof of Theorem~\ref{thm:main}]
Let $N$ be the maximum encoding length of a source formula.
Each module has $O(n)$ rail, variable-chain, and helper-path vertices.
Each of its $m_i$ clauses adds at most three length-$R$ clause paths and
two length-$A$ hub links, hence $O(n)$ vertices. Consequently,
\[
 |V(M_i)|=O\bigl(n(m_i+1)\bigr)=O(N^2).
\]
The entire construction is polynomial in the total input length.

Let $W_i$ be the vertices of $M_i$ other than its four boundary vertices.
Order $V(H)$ by listing the left tail from $\alpha$ towards $a_0$, then
$b_0$, then, for $i=1,\ldots,q$, the vertices of $W_i$ in any order followed
by $a_i,b_i$, and finally the right tail from the neighbor of $b_q$ to
$\beta$. Every module edge has both endpoints in the consecutive block
\[
 \{a_{i-1},b_{i-1}\},\ W_i,\ \{a_i,b_i\},
\]
which has size $O(N^2)$; every tail edge has span one. Thus
$\bw(H)=O(N^2)$, independently of $q$. This ordering is explicitly
constructed. One may therefore use the formulation with a supplied
ordering and its width as parameter: a kernel parameterized by
$\bw(H)$ would also kernelize that formulation by ignoring the supplied
ordering and equipping the output with an arbitrary ordering.

We have obtained an OR-cross-composition from an NP-hard language with
output parameter polynomial in $N+\log q$. The cross-composition theorem
of Bodlaender, Jansen, and Kratsch~\cite{BJK14}
therefore excludes a polynomial kernel unless
$\mathrm{NP}\subseteq\mathrm{coNP/poly}$.
\end{proof}
}

\begin{document}

\maketitle

\mespFullAbstract

\ifthenelse{\boolean{shortver}}{%
\section{Introduction}
\label{sv:sec:introduction}

The \emph{Eccentricity Shortest Path} problem (\ESP) asks, given a
connected graph $G$ and a nonnegative integer $R$, whether some shortest
path is within distance $R$ of every vertex. Its optimization version
minimizes this eccentricity. A solution combines a covering requirement
with exact shortestness: extending a path can improve coverage while
destroying shortestness, and replacing it by a shorter route can leave
vertices uncovered. Dragan and Leitert introduced the problem and showed
its connection with approximating minimum-distortion embeddings of graph
metrics into the line~\cite{DraganLeitert2017}.

The problem is NP-hard even on planar bipartite graphs of maximum degree
three, and $\Wtwo$-hard parameterized by $R$~\cite{DraganLeitert2017}.
Nevertheless, it is polynomial-time solvable on several structured classes,
including chordal and distance-hereditary graphs~\cite{DraganLeitert2016Structured}.
Ku\v{c}era and Such\'y~\cite{KuceraSuchy2023} initiated its structural
parameterized study, obtaining FPT algorithms for modular-width, cluster
deletion distance, maximum leaf number, and treewidth combined with $R$.
Bhyravarapu et al.~\cite{BhyravarapuEtAl2023,BhyravarapuEtAl2023Full}
obtained further algorithms for deletion distance to disjoint paths and
split graphs, and proved $\Wtwo$-hardness for chordal deletion distance.
They asked whether interval deletion distance admits an FPT algorithm.
More recently, Bhyravarapu and Mohanapriya~\cite{BhyravarapuMohanapriya2026}
established $\Wone$-hardness for deletion distance to disjoint unions of
spiders, implying hardness for feedback vertex number, pathwidth, and
treewidth, and gave an FPT algorithm for feedback edge number.

We study how structural parameters control both computation and polynomial
preprocessing, with $R$ unrestricted. All deletion distances below count
vertices. Our results connect three aspects of this landscape: the
difference between closely related deletion parameters, the role of
ordering width, and the limits of polynomial kernelization.

For deletion parameters, we resolve the interval-deletion question
positively. We also obtain an XP algorithm for chordal deletion distance,
complementing the known $\Wtwo$-hardness. Thus deleting vertices to interval
graphs and to chordal graphs leads to different parameterized behavior.
Distance to cographs also yields fixed-parameter tractability.

For ordering parameters, \ESP{} is FPT for bandwidth but $\Wthree$-hard
for cutwidth, even on graphs of maximum degree at most five and with an
ordering supplied.
The latter result also strengthens the known hardness for pathwidth and
treewidth. On the algorithmic side, an XP algorithm for clique-width
follows from an algorithm on supplied mim-width layouts whose polynomial
exponent is independent of $R$.
It yields polynomial-time solvability at every fixed treewidth, answering
the question raised by Chakraborty et al.~\cite{ChakrabortyEtAl2026}.

For preprocessing, maximum leaf number admits a polynomial kernel, while
bandwidth, feedback edge number, distance to a clique, and vertex cover
number admit no polynomial kernel unless
$\NP\subseteq\mathrm{coNP}/\mathrm{poly}$. These lower bounds distinguish
fixed-parameter tractability from polynomial kernelization, and imply the
same obstruction for interval deletion distance. Further results include
$\Wtwo$-hardness for vertex clique-cover number and NP-hardness at
tree-independence number two. We also record treedepth FPT and diameter XP
as consequences of earlier algorithms.

The main text develops the bandwidth and deletion-distance algorithms,
the clique-width approach, the cutwidth and clique-cover reductions,
the tree-independence hardness result, and representative kernelization
arguments. We also derive the treedepth and diameter consequences explicitly.
The full version, containing complete proofs and additional technical
details, is appended after the references.

Table~\ref{sv:tab:results} lists the main statements, and
Figure~\ref{fig:parameter-landscape} locates them in the parameter hierarchy.

\begin{table}[t]
\centering
\small
\renewcommand{\arraystretch}{1.10}
\begin{tabularx}{\textwidth}{@{}>{\raggedright\arraybackslash}p{0.31\textwidth}
 >{\raggedright\arraybackslash}X@{}}
\hline
\textbf{Parameter} & \textbf{Result and location} \\
\hline
Bandwidth $b$ & $2^{O(b^2\log(b+1))}n^{O(1)}$ time
 (Theorem~\ref{sv:thm:bandwidth}); no polynomial kernel \\
Cutwidth & $\Wthree$-hard at maximum degree five
 (Theorem~\ref{sv:thm:cutwidth}) \\
Clique-width; rank-width & XP without a supplied decomposition
 (Theorem~\ref{sv:thm:width-xp}) \\
Supplied mim-width $w$ & $n^{O(\max\{1,w\})}$ time
 (Theorem~\ref{sv:thm:width-xp}) \\
Interval deletion distance $k$ & $13^k n^{O(1)}$ time
 (Theorem~\ref{sv:thm:interval}) \\
Chordal deletion distance $k$ & $n^{O(k)}$ time
 (Theorem~\ref{sv:thm:chordal}) \\
Maximum leaf number $\ell$ & $O(\ell^4)$-bit compression and polynomial kernel
 (Theorem~\ref{sv:thm:mainleaf}) \\
Distance to clique; feedback edge number; vertex cover number & No polynomial kernel
 (Theorem~\ref{sv:thm:kernel-lower-bounds}) \\
Cograph deletion distance & FPT (Theorem~\ref{sv:thm:cograph}) \\
Vertex clique-cover number & $\Wtwo$-hard
 (Theorem~\ref{sv:thm:clique-cover-hardness}) \\
Tree-independence number & NP-hard at value at most two
 (Theorem~\ref{sv:thm:tree-independence-hardness}) \\
Treedepth $d$; diameter $\Delta$ & Consequences of earlier algorithms:
 $F(d)n^3$ and $O(n^{2\Delta+2}m)$, respectively
 (Proposition~\ref{sv:prop:known-consequences}) \\
\hline
\end{tabularx}
\caption{Results discussed in this paper. Kernel lower bounds assume
$\NP\nsubseteq\mathrm{coNP}/\mathrm{poly}$. The treedepth and diameter
consequences follow from~\cite{KuceraSuchy2023,DraganLeitert2017}; the
clique-width and rank-width algorithms apply known model checking.
All deletion distances count vertices.}
\label{sv:tab:results}
\end{table}

\mespLandscape

\section{Preliminaries}
\label{sv:sec:preliminaries}

Graphs are connected, finite, simple, and unweighted, unless a weighted
multigraph is explicitly introduced. Write $n=|V(G)|$, $m=|E(G)|$,
and $d_G(u,X)=\min_{x\in X}d_G(u,x)$, with distance to the empty set
or between components equal to $+\infty$. A \emph{geodesic} is a path
shortest between its endpoints; singleton paths are allowed. We may
assume $0\leq R<n$, since larger thresholds give yes-instances.
All-pairs distances and singleton-path eccentricities take polynomial time.

For a graph class $\mathcal H$, an $\mathcal H$-\emph{modulator} is a
vertex set $S$ with $G-S\in\mathcal H$; deletion distance is its minimum
size. Interval graphs are intersection graphs of intervals on a line;
chordal graphs have no induced cycle of length at least four; cographs
have no induced four-vertex path. Intervals may be represented with
distinct endpoints. The maximum leaf number $\ml(G)$ is the largest
number of leaves in a spanning tree, and the feedback edge number
is $m-n+1$. Vertex clique-cover number counts cliques covering all
vertices, while tree-independence number is the minimum, over tree
decompositions, of the largest independence number of a bag.

For an ordering $\pi$, bandwidth bounds $|\pi(u)-\pi(v)|$ for every
edge $uv$, and cutwidth bounds the number of edges crossing every
prefix cut. A tree decomposition has width one less than its largest
bag size. The treedepth $\td(G)$ is the minimum height, counted in
vertices, of a rooted forest on $V(G)$ in which the endpoints of every
edge of $G$ are comparable in the ancestor order.
We use rooted nice decompositions with introduce, forget,
and join nodes. A \emph{layout} is a rooted binary tree whose leaves
are the graph vertices; its mim-width is the maximum induced-matching
size in the bipartite graph across a tree cut, and its rank-width
uses the cut-adjacency rank over $\GF(2)$. Clique-width is the number
of labels in an expression using vertex creation, disjoint union,
relabeling, and complete joins between labels. Further conventions
and standard parameter relations appear in the \emph{Preliminaries} section
of the full paper.

\section{Width parameters: algorithms and hardness}
\label{sv:sec:width}

\subsection{An FPT algorithm for bandwidth}

\begin{theorem}\label{sv:thm:bandwidth}
\ESP{} can be solved in $2^{O(b^2\log(b+1))}n^{O(1)}$ time,
where $b=\bw(G)$, without a supplied bandwidth ordering.
A witnessing path and an optimal path can be found within the same bound.
\end{theorem}

The algorithm for treewidth together with $R$~\cite{KuceraSuchy2023}
does not directly give this result, since the optimum eccentricity is
unbounded even on cycles. We instead bound the number of distance values
that a shortest path can realize. Put
$L=1+\lfloor\log_2 n\rfloor$, $M_0=1$, and
$M_b=(40bL)^b$ for $b\geq1$.

\subsubsection{A bound on separated vertices}

In a tree $T$ rooted at $r$, write $h(x)=d_T(r,x)$ and call
$A\subseteq V(T)$ \emph{separated} if
$d_T(x,y)\geq2|h(x)-h(y)|$ for all $x,y\in A$.

\begin{lemma}\label{sv:lem:separated}
Every separated set in a rooted tree with at most $n$ vertices and
bandwidth at most $b$ has size at most $M_b$.
\end{lemma}

\begin{proof}[Proof sketch]
Assume $A\ne\varnothing$. We adapt the spine-removal and dyadic-grouping arguments of Eppstein,
Goodrich, and Liu~\cite[Lemmas~12 and~13]{EppsteinGoodrichLiu2025}.
No member of $A$ is a proper ancestor of another. Prune the tree to its
root-to-$A$ paths, so $A$ is precisely its set of leaves, and fix an
ordering of bandwidth $b$. Choose leaves $a_-,a_+$ whose root paths
contain the two ordering extremes; let $S$ be the union of those paths.
It contains at most two members of $A$.

Deleting $S$ and closing the ordering gaps reduces bandwidth to $b-1$.
Indeed, every remaining edge of span $b$ has a vertex of $S$ strictly
between its endpoints: otherwise the path in $S$ between the ordering
extremes would cross the gap with span greater than $b$.
Each component $K$ of $T-S$ attaches at one vertex $u\in S$.
Assign it to one of the root paths containing $u$, with terminal
$a\in\{a_-,a_+\}$. Taking $x\in A\cap V(K)$, separation gives
$d_T(u,a)\leq3d_T(u,x)\leq3|V(K)|$.
Hence every vertex of $K$ is within distance $4|V(K)|$ of $a$.

Group the assigned components by
$\ell\leq|V(K)|<2\ell$, where $\ell$ is a power of two.
Their vertices lie within $8b\ell$ ordering positions of $a$, so
their union has at most $17b\ell$ vertices. Thus at most $17b$
components belong to each group. There are at most $L$ groups per
terminal. Separation survives inside each component, so induction gives
\[
 |A|\leq2+34bL M_{b-1}\leq(40bL)^b=M_b.
\]
The base case $b=0$ is immediate.
\end{proof}

\subsubsection{Distance spectra}

Fix distinct endpoints $s,t$, let $D=d_G(s,t)$ and $a(x)=d_G(s,x)$.
On the geodesic interval $I=\{x:d_G(s,x)+d_G(x,t)=D\}$,
retain exactly the arcs $x\to y$ with $xy\in E(G)$ and
$a(y)=a(x)+1$.
The resulting DAG $H$ has exactly the shortest $s$--$t$ paths as
its directed $s$--$t$ paths. For $v\in V(G)$ define
\[
 \Spec_{s,t}(v)=\{d_G(v,V(P)):P\text{ is a shortest }s\text{--}t\text{ path}\}.
\]

\begin{lemma}\label{sv:lem:spectrum}
After computing graph distances, each spectrum can be computed in
$O(n+m)$ time and has size at most $M_b^3$.
\end{lemma}

\begin{proof}
Put $w(x)=d_G(v,x)$. Let $A(x)$ be the largest possible minimum
$w$-value on a directed $s$--$x$ path, and define $B(x)$ for an
$x$--$t$ path. In forward and reverse layer order, compute
\begin{align*}
 A(s)&=w(s),& A(x)&=\min\{w(x),\max_{y\to x}A(y)\},\\
 B(t)&=w(t),& B(x)&=\min\{w(x),\max_{x\to y}B(y)\}.
\end{align*}
Then
$\Spec_{s,t}(v)=\{w(x):A(x)=B(x)=w(x)\}$:
a minimum-weight vertex of a geodesic satisfies the equalities, and
witnessing prefix and suffix paths concatenate to a geodesic.

List its values $r_1<\cdots<r_q$, choose corresponding witnesses
$x_i$, and write $a_i=a(x_i)$. Maximizing predecessor choices for
$A$ form a tree $T_s$ with depths $a(x)$; maximizing successor
choices for $B$ form a tree $T_t$ with depths $D-a(x)$.
Let $T_v$ be a BFS tree of $G$ rooted at $v$. All three trees are
subgraphs of $G$ and have bandwidth at most $b$.
For $i<j$, put $\Delta r=r_j-r_i$ and $\Delta a=a_j-a_i$.
The lowest common ancestor of $x_i,x_j$ in $T_s$ has weight at
least $r_j$. Since $w$ changes by at most one along an edge,
and symmetrically for $T_t$, we obtain
\begin{align*}
 d_{T_s}(x_i,x_j)&\geq\Delta a+2\Delta r,&
 d_{T_t}(x_i,x_j)&\geq-\Delta a+2\Delta r,\\
 d_{T_v}(x_i,x_j)&\geq|\Delta a|.&
\end{align*}

An increasing sequence of indices with consecutive differences
$0\leq\Delta a\leq2\Delta r$ gives a separated set in $T_s$;
one with $-2\Delta r\leq\Delta a\leq0$ gives a separated set in
$T_t$. These inequalities extend to every pair by summation.
Both sequences therefore have length at most $M_b$.
Associate with $i$ the pair of maximum lengths of such sequences
ending at $i$. There are at most $M_b^2$ pairs. Equal pairs at
$i<j$ force $|\Delta a|>2\Delta r$, since otherwise a sequence
ending at $i$ extends to $j$. Each equal-pair class is consequently
separated in $T_v$, whose witness depths are $r_i$, and has at most
$M_b$ members. Thus $q\leq M_b^3$.
\end{proof}

\subsubsection{A local certificate and dynamic program}

Allow vertex $v$ the labels
$\Lambda(v)=\Spec_{s,t}(v)\cap\{0,\ldots,R\}$, rejecting an empty
domain. A solution exists exactly when labels
$\lambda(v)\in\Lambda(v)$ satisfy the following conditions, with
$Z=\{v:\lambda(v)=0\}$:
\begin{enumerate}
\item $s,t\in Z$ and $|Z|\leq D+1$;
\item each vertex of $Z\setminus\{s\}$ has a predecessor in $Z$ in $H$;
\item each positively labelled vertex has a neighbor with label one smaller.
\end{enumerate}
True distances to a solution satisfy these conditions.
Conversely, following zero predecessors from $t$ gives a geodesic
with $D+1$ vertices, exhausting $Z$. Decreasing positive labels
give walks to $Z$ of length at most $R$, proving coverage.

On a nice tree decomposition of width $w$, store an allowed label
for each bag vertex, forcing $s,t$ to zero, and one bit recording whether its required
neighbor has already been supplied by a forgotten vertex.
For zero labels the required neighbor is a zero predecessor in $H$;
for positive labels it has label one smaller. Vertex $s$ has no
such requirement. Every forgotten vertex other than $s$ must already have support.
The state value is the minimum number of forgotten
zero-labelled vertices. Initialize the empty leaf value to zero and
minimize over transitions producing the same state.
At introduction, initialize the new support
bit to false. At forgetting, require support either in the bag or
through the bit, update the remaining bits, and charge the vertex
if its label is zero. At a join, agree on bag labels, OR the bits,
and add the values; forgotten sets are disjoint. At the empty root,
accept a value at most $D+1$. This gives at most
$(2q)^{w+1}$ states, where $q=\max_v|\Lambda(v)|$, and running time
$2^{O(w)}q^{O(w+1)}n^{O(1)}$.

Since $\tw(G)\leq b$, Korhonen's algorithm~\cite{Korhonen2021},
tried with increasing width bounds, finds a decomposition of width
at most $2b+1$ in $2^{O(b)}n$ time without a bandwidth ordering.
Test singleton paths separately, then all endpoint pairs. Substituting
$q\leq M_b^3$ gives total time
$2^{O(b^2\log(b+1))}L^{O(b^2)}n^{O(1)}$.
For $K\geq1$, maximizing $(1+x)^Ke^{-x}$ yields
$L^K\leq e(2K)^K n$; with $K=O(b^2)$ this proves the claimed FPT
bound. Backtracking recovers a path, and binary search over $R$
finds the optimum with only a polynomial additional factor.
Complete arguments and transitions appear in the \emph{Bandwidth} subsection
of the \emph{FPT algorithms} section in the full paper.

\subsection{Hardness for cutwidth}
\label{sv:sec:cutwidth}

We contrast the bandwidth algorithm with a hardness result for cutwidth.
The source problem is \CSP: for a constraint graph $F$, each variable $v$
has an explicitly listed domain $D_v$, and each edge $uv$ has an allowed
relation $R_{uv}\subseteq D_u\times D_v$. An assignment must satisfy all
relations. This problem is $\Wthree$-complete parameterized by
$\vc(F)$~\cite[Theorem~1]{BGP23}.

\begin{theorem}\label{sv:thm:cutwidth}
\ESP{} is $\Wthree$-hard parameterized by cutwidth, even on connected
graphs of maximum degree at most five. A \CSP{} instance with constraint graph
$F$ can be transformed in polynomial time into an equivalent instance
$(G,q)$ and an ordering $\pi$ with
$\cutwidth(G,\pi)\leq4\vc(F)+7$.
If a vertex cover of size $k$ is supplied, the bound is $2k+7$.
\end{theorem}

\begin{proof}[Proof sketch]
Let $X=\{x_1,\ldots,x_k\}$ be a supplied vertex cover and let
$Y=\{y_1,\ldots,y_m\}$ be the remaining variables. Empty variable sets
and empty domains have constant-size yes- and no-outputs, respectively;
otherwise relabel $D_v=\{0,\ldots,N_v\}$ and put
$S=\sum_v(N_v+1)$.
A \emph{test} $\tau=(u,v,a,b)$ records a forbidden pair
$(a,b)\notin R_{uv}$, using one fixed orientation of each constraint edge.
Order the tests first within $X$, then by their endpoint $y_1,y_2,\ldots,y_m$;
every remaining test has exactly one endpoint in $Y$.
Let $T$ be their number, $t_v$ the number containing $v$, and
$p_v(\tau)$ the position of $\tau$ in the induced sublist for $v$.
Set $L=2(S+T)$ and $q=2L$.

\emph{Construction.}
For each $v$, make a ladder with rails
$A_{v,0}\cdots A_{v,N_v}$ and $B_{v,0}\cdots B_{v,N_v}$.
For each $a\in D_v$, its rung joins $A_{v,a}$ to $B_{v,a}$ by a path
of length $N_v+1$. Join consecutive ladders in the order
$x_1,\ldots,x_k,y_1,\ldots,y_m$ by edges from the last bottom-rail
vertex to the first top-rail vertex. Denote this core by $C$, with
first and last terminals $s,t$.
At ports $U_v=A_{v,N_v}$ and $V_v=B_{v,0}$ attach two carrier paths
$U_v^0\cdots U_v^{t_v}$ and $V_v^0\cdots V_v^{t_v}$, rooted at their
superscript-zero vertices. Let $H$ consist of the core and carriers.

For each test $\tau=(u,v,a,b)$, introduce a center $c_\tau$.
For each $(w,h)\in\{(u,a),(v,b)\}$, with $p=p_w(\tau)$, connect
$c_\tau$ to $U_w^p$ and $V_w^p$ by paths of respective lengths
\[
 L-N_w+h+1-p
 \qquad\text{and}\qquad
 L-h+1-p.
\]
Attach a length-$L$ pendant at $c_\tau$, ending at $z_\tau$.
Finally attach length-$(q+1)$ terminal pendants at $s$ and $t$,
ending at $s^*,t^*$. All internal vertices are fresh and all unspecified
edges are absent. Each connection length $\lambda$ satisfies
$L/2<\lambda\leq L$. Every test center has degree five; all other
vertices have degree at most three. Since $T\leq S^2$, the graph has
$O(S^2+(T+1)L)$ vertices and is constructed in polynomial time.

\emph{Canonical paths.}
The diameter of ladder $v$ is at most $2N_v+1$: of the routes using
the top and bottom rails, one uses at most $N_v+1$ rung edges and
$N_v$ rail edges. Therefore
$\diam(C)\leq2S-1$ and $\diam(H)\leq L-1$.
A shortest path with both endpoints in $H$ stays there, since an
excursion traverses two connection paths and has length greater than $L$.

A shortest crossing of ladder $v$ uses exactly one rung, say $\rho(v)$,
and moves monotonically along the rails: its minimum length is
$2N_v+1$. Consequently the shortest $s$--$t$ paths correspond exactly
to assignments $\rho$. Extend such a path to both terminal leaves and
call it $P_\rho$. Its distances from the two ports are
\[
 d_G(U_v,P_\rho)=N_v-\rho(v),
 \qquad d_G(V_v,P_\rho)=\rho(v).
\]
Indeed, the indicated rail segments attain these distances, while
another rung is longer and leaving the ladder first meets the path.
Conversely, any shortest path of eccentricity at most $q$ must meet
both terminal pendants away from their roots, because their lengths
are $q+1$. Its endpoints lie in these pendants and its core segment
selects one rung per ladder. Extending it to $s^*,t^*$ therefore yields
some $P_\rho$ without increasing eccentricity.

\emph{Testing assignments.}
A shortest route from $c_\tau$ to $P_\rho$ follows one connection
path, then its carrier and the core. For $(w,h)\in\{(u,a),(v,b)\}$,
the two resulting lengths are $L+1+h-\rho(w)$ and
$L+1+\rho(w)-h$: the attachment position cancels the carrier length.
Adding the test pendant gives
\begin{equation}\label{sv:eq:cutwidth-test}
 d_G(z_\tau,P_\rho)
 =q+1-\max\{|\rho(u)-a|,|\rho(v)-b|\}.
\end{equation}
Thus $z_\tau$ is covered exactly when $\rho$ avoids the forbidden pair.
All vertices of $H$ are within $L-1$ of $P_\rho$, and connection-path
vertices are within a further $L$. Terminal pendants lie on $P_\rho$;
on a test pendant its leaf is farthest from $P_\rho$.
Hence $\ecc_G(P_\rho)\leq q$ precisely for satisfying assignments.

\emph{The ordering.}
Begin with the pendant from $s^*$ towards $s$, omitting $s$.
Process each ladder
column by column, listing its top vertex, rung interior, and bottom
vertex. Process the $X$ ladders and their internal tests first, then
each $y_j$ ladder immediately followed by all its tests. For a test,
list its four carrier attachments, its four connection-path interiors
one at a time towards its center, and then its center and pendant.
Finish with the pendant from the neighbor of $t$ to $t^*$.

Every carrier is listed in increasing position order and contributes
at most one crossing edge. The $X$ carriers contribute at most $2k$;
only the current $Y$ block contributes further carrier edges, at most
two. During a ladder, at most three non-carrier edges cross.
During a test, at most one core or terminal edge and four test edges
cross. Thus the ordering has cutwidth at most $(2k+2)+1+4=2k+7$.
If no cover is supplied, take the endpoints of a maximal matching;
they form a cover of size at most $2\vc(F)$.
The \emph{Cutwidth} subsection of the full paper gives the complete
supporting arguments.
\end{proof}

\subsection{An XP application for clique-width and rank-width}
\label{sv:sec:cliquewidth}

The following algorithms are applications of known model checking and
width-conversion results. Details appear in the \emph{Clique-width}
subsection of the full paper.

\begin{theorem}\label{sv:thm:width-xp}
For a connected $n$-vertex graph with $n\geq2$, \ESP{} is solvable in
$n^{O(w)}$ time given a rooted layout of mim-width $w$, and in
$|\tau|^{O(1)}+n^{O(k)}$ time given a clique-width $k$-expression $\tau$.
Without a supplied decomposition, it is solvable in
$f(q)n^{O(q)}$ time for $q=\rw(G)$, and hence in $f'(k)n^{O(k)}$
time for $k=\cw(G)$, for computable functions $f,f'$.
Each algorithm returns a witnessing shortest path.
\end{theorem}

Handle the one-vertex graph directly. For each endpoint pair $s,t$,
including $s=t$, put $D=d_G(s,t)$. A set $X$ induces an $s$--$t$
geodesic precisely when $s,t\in X$, $|X|=D+1$, and $G[X]$ is
connected: any $s$--$t$ path in $G[X]$ already needs all $D+1$
vertices, and any additional edge would shorten it. Thus, for $R\geq1$,
the distance-neighborhood logic formula
\[
 (\{s,t\}\subseteq X)\ \wedge\ |X|=D+1\ \wedge\ \con(X)
 \ \wedge\ \overline X\subseteq N_1^R(X)
\]
expresses the desired solution, using singleton predicates for the
endpoints. Here $N_1^R(X)$ contains vertices at positive distance at
most $R$ from some member of $X$; for $R=0$, replace the last conjunct
by $\overline X=\varnothing$.

By Bergougnoux, Dreier, and Jaffke~\cite[Theorem~1.2]{BergougnouxDreierJaffke2023},
this formula is solvable in $n^{O(w)}$ time on the supplied layout.
Numerical radii and cardinality constants each count as one symbol in
their formula-size convention. The formula therefore has constant size
and counting threshold, so the exponent is independent of $R$.
The $n^2$ endpoint guesses preserve the bound.

The standard conversion from a $k$-expression gives a layout with
cut-rank at most $k$~\cite[Proposition~6.3]{OumSeymour2006}.
Every cut's induced-matching number is at most its rank, since an
induced matching supplies an identity submatrix. Without an expression,
trying successive bounds in Oum's algorithm~\cite[Theorem~5.8]{Oum2008}
yields a rank-decomposition of
width at most $3q-1$ in $g(q)n^3$ time. Rooting this decomposition and
using $\rw(G)\leq\cw(G)$ gives the remaining bounds. The computable
preprocessing factor is retained when no decomposition is supplied.

\section{Deletion distances}
\label{sv:sec:deletion}

We present FPT algorithms for deletion distance to interval graphs and
to cographs, and an XP algorithm for chordal deletion distance. The latter complements the
$\Wtwo$-hardness of Bhyravarapu et al.~\cite{BhyravarapuEtAl2023Full}.
Complete proofs and transitions appear in the corresponding subsections
of the full paper.

\subsection{Distance to interval graphs}

\begin{theorem}\label{sv:thm:interval}
\ESP{} can be solved in $13^k n^{O(1)}$ time, where $k$ is the
vertex-deletion distance to interval graphs.
\end{theorem}
This resolves the open question posed by Bhyravarapu et al.~\cite{BhyravarapuEtAl2023Full} of whether ESP is fixed-parameter tractable when parameterized by the vertex-deletion distance to interval graphs.
Find a modulator $S$, $|S|=k$, in $10^k n^{O(1)}$
time~\cite{cao2015interval}, and an interval representation of $H=G-S$
in polynomial time~\cite{BoothLueker1976}. Test singleton paths
separately. Thereafter, guess distinct endpoints $a,b$ and put
$D=d_G(a,b)$. A candidate vertex set must contain $a,b$, have
$D+1$ vertices, and induce degrees one at $a,b$ and two elsewhere.
These conditions force a shortest $a$--$b$ path: its $a$--$b$ component
already requires at least $D+1$ vertices, leaving none for other
components. Write $\tau(a)=\tau(b)=1$ and $\tau(v)=2$ otherwise.

\paragraph*{Enumerating distance profiles.}
For a hypothetical solution $P$, let $\rho_s=d_G(s,V(P))$ for $s\in S$.
The central observation is that only constantly many possibilities per
modulator vertex are needed, once the endpoints and nearby anchors are
fixed. Put $T=\{s\in S:\rho_s\leq2\}$ and
\[
 \beta_s=\min\left\{d_G(s,a),d_G(s,b),
       \min_{t\in T}\bigl(d_G(s,t)+\rho_t\bigr)\right\}.
\]
Then $\rho_s\leq\beta_s\leq\rho_s+2$, with equality
$\beta_s=\rho_s$ on $T$.
To see the additive bound, a geodesic $Q$ of an interval graph satisfies
the following endpoint property: if $d_H(u,Q)=r\geq2$ is finite, an
endpoint of $Q$ is within distance $r+1$ of $u$. Otherwise, $u$ and the
two endpoints form an asteroidal triple, contrary to
AT-freeness~\cite{LekkerkerkerBoland1962}.
Follow a shortest $s$--$P$ path to its last modulator vertex and apply
this property to the encountered component of $P-S$. It either reaches
a nearby anchor or comes within additive two of an endpoint or a vertex
of $P\cap S$.

Accordingly, assign each $s$ one of six labels: three declare
$\rho_s=0,1,2$, and three declare $\rho_s=\beta_s-j>2$ for
$j=0,1,2$. The near labels determine $T$ and its distances before
$\beta$ is computed. This enumerates the true profile of every
solution. The zero set $M=\{s:\rho_s=0\}$ determines $P\cap S$;
for any prescribed $M$, at most $5^{k-|M|}$ profiles remain.
Reject profiles violating $0\leq\rho_s\leq R$, the endpoint
requirements, or the degree bounds within $M$.

\paragraph*{Checking a profile by an interval sweep.}
Write $Y=M\cup X$, where $X\subseteq V(H)$ is to be selected.
Require $d_G(s,y)\geq\rho_s$ for every $s\in S$ and $y\in Y$,
and an equality witness for each $s$. Once these conditions hold,
\[
 d_G(v,Y)=\min\left\{d_H(v,X),
            \min_{s\in S}(d_G(v,s)+\rho_s)\right\}.
\]
Thus only vertices in
$U_\rho=\{v\in V(H):d_G(v,s)+\rho_s>R\text{ for all }s\in S\}$
still require coverage within $H$.

Use an interval model with distinct endpoints and sweep in increasing
right-endpoint order. Distances from the current interval to processed
vertices are nonincreasing in their right-endpoint order; distances to
unprocessed vertices are nondecreasing in their left-endpoint order.
Consequently, retaining the last three selected processed vertices
in right-endpoint order and the first three selected unprocessed
vertices in left-endpoint order suffices to preserve selected-neighbor
counts, truncated at three, and to test whether a selected vertex
lies within any prescribed positive distance. Let $R_3$ and $L_3$
denote these summaries.

The state $(A,B,c,q,W)$ records the processed summary $A$, a promised
future summary $B$, the number $c(m)$ of processed selected neighbors of
each $m\in M$, the total selected count $q$ including $M$, and the
vertices $W\subseteq S\setminus M$ with an equality witness already
found. Each $c(m)$ ranges from zero to
$t_m=\tau(m)-|N_G(m)\cap M|\leq2$.
When processing $v$, choose $Z=\{v\}$ or $\varnothing$, forcing
selection at the endpoints, and guess the next promise $B'$ from at most
three remaining vertices, subject to
$B=\mathrm L_3(B'\cup Z)$.
If selected, $v$ must satisfy the distance lower bounds and
\[
 \min\{3,|N_G(v)\cap M|+|N_H(v)\cap A|
                       +|N_H(v)\cap B'|\}=\tau(v).
\]
If skipped and in $U_\rho$, require a vertex of $A\cup B'$ within
distance $R$ in $H$; reject such a skip when $R=0$.
Update $A$ to $\mathrm R_3(A\cup Z)$, increase the counters for
selected neighbors, update $q$ and $W$, and reject exceeded bounds.
Initialize with $A=\varnothing$, $c=0$, $q=|M|$ and the witnesses
already in $M$, enumerating all initial promises.
Accept precisely when $B=\varnothing$, $q=D+1$,
$W=S\setminus M$, and every $c(m)=t_m$.
\paragraph*{Correctness of the promised summaries.}
Consider an accepting sequence of transitions, and let $X$ be its
selected vertices of $H$. Write $h_1,\ldots,h_N$ for the sweep order
and $L_i=\{h_1,\ldots,h_i\}$. The meanings of $A,c,q$ follow by
forward induction. The future summary requires backward induction:
$B_N=\varnothing$, and, writing $Z_i=X\cap\{h_{i+1}\}$, every
transition satisfies
\[
 B_i=\mathrm L_3(B_{i+1}\cup Z_i)
     =\mathrm L_3(X\setminus L_i).
\]
The second equality follows inductively because taking the first three
vertices commutes with truncating a set before adjoining $Z_i$.
Thus every future vertex used in a degree or coverage check belongs
to the eventual selection. The summaries need this interpretation
only along accepting computations; unsuccessful promises cause no
false acceptance.

The witness set $W$ also remains valid throughout: once a selected
vertex attains distance $\rho_s$ from $s$, no later selection can
decrease that distance below $\rho_s$, because all selected vertices
satisfy the lower bounds. At acceptance, $W=S\setminus M$ and the
vertices of $M$ witness their own zero distances. Hence the entire
guessed profile equals the actual distances to $M\cup X$.
The separator identity now certifies coverage outside $U_\rho$,
while the sweep certifies coverage inside $U_\rho$; vertices of $S$
are covered because $\rho_s\leq R$. Together with the degree and
cardinality checks this proves soundness. Conversely, following any
valid path's actual summaries, counters, and equality witnesses
satisfies every transition and the acceptance conditions, proving
completeness.

For fixed $M$ and $\rho$, the summaries and selected count have
polynomially many possibilities; the counters and witness set contribute
$3^{|M|}2^{k-|M|}$ states. Transitions take polynomial time per state.
Summing over profiles and zero sets gives
\[
 \sum_{M\subseteq S}5^{k-|M|}3^{|M|}2^{k-|M|}
 =\sum_{m=0}^k\binom{k}{m}3^m10^{k-m}=13^k.
\]
Endpoint enumeration and modulator computation preserve the claimed bound.

\subsection{Distance to chordal graphs}

\begin{theorem}\label{sv:thm:chordal}
\ESP{} is solvable in $n^{O(k)}$ time, where $k$ is the
vertex-deletion distance to chordal graphs.
\end{theorem}

Enumerate a modulator $S$ and construct a nice tree decomposition with
bags $S\cup C$, where $C$ is a clique of $H=G-S$, using the standard
clique-tree representation~\cite{Gavril1974,BlairPeyton1993,Kloks1994}.
The essential property is that a geodesic $Q$ of a chordal graph,
disjoint from a clique $C$, contains one vertex simultaneously closest
to every vertex of $C$. This follows by repeatedly deleting a simplicial
vertex outside $C$: deletion preserves distances, and when the vertex
is an endpoint of $Q$, its neighbor on $Q$ is no farther from any
vertex of $C$. Dirac's theorem ensures the required simplicial
vertex~\cite{Dirac1961}.

Consequently, a union of $h$ geodesics of $H$ disjoint from $C$ has a set of at most
$h+k$ representatives preserving its distances to every vertex of
$S\cup C$: retain one vertex per geodesic and one nearest vertex for
each member of $S$. A solution meets $C$ in at most two consecutive
vertices. Removing them and $S$ leaves at most $k+2$ geodesic pieces,
so each side of a bag has a distance profile represented by at most
$2k+2$ vertices. Enumerating these representatives gives $n^{O(k)}$
possible profiles.

After testing singleton paths, fix distinct endpoints $a,b$ and put
$D=d_G(a,b)$. Only vertices $v$ with
$d_G(a,v)+d_G(v,b)=D$ may be selected. Orient edges by increasing
$d_G(a,v)$ and force selection of the endpoints.
The dynamic program stores the selected bag
vertices (at most two from $C$), their partial directed degrees, and
the profiles of selected vertices inside and outside the bag.
Separator-distance identities transfer profiles; pointwise minima
combine them at joins. Inside profiles are all-infinite at leaves;
the outside profile is all-infinite at the root. Induction from the
leaves and from the root, respectively, certifies both profiles.
We forbid selected edges within a distance layer.
Requiring degree pairs $(0,1)$ at $a$, $(1,0)$ at $b$, and $(1,1)$
elsewhere certifies one geodesic, and
the minimum of the two profiles and the selected bag distances
certifies coverage. There are $n^{O(k)}$ states and transitions.

\subsection{Distance to cographs}
\label{sv:sec:cographs}

\begin{theorem}\label{sv:thm:cograph}
For a computable function $f$, \ESP{} and its optimization version
can be solved in $f(k)n^{O(1)}$ time, where $k=\dcog(G)$ is the
vertex-deletion distance to cographs.
No cograph modulator is required as part of the input, and an optimal
shortest path can be returned.
\end{theorem}

\begin{proof}[Proof sketch]
Greedily select a maximal collection of vertex-disjoint induced
$P_4$'s and let $X$ contain their vertices. Every cograph modulator
meets each selected path, so $t:=|X|\leq4k$; maximality ensures that
$G-X$ is a cograph. Scanning all four-vertex sets finds $X$ in
$O(n^4)$ time.

A binary cotree of $G-X$ gives a clique-width $2$-expression in
polynomial time~\cite{Corneiletal1985,BodlaenderMohring1993}.
Refine each of its two labels by the neighborhood in $X$, giving
$2^{t+1}$ labels, and give each member of $X$ its own label.
Replace every cotree join by joins between the corresponding refined
labels, and then add the vertices of $X$ and their incident edges.
This constructs a $(2^{t+1}+t)$-expression of size $O(4^t n)$;
compare the standard vertex-addition construction~\cite{Gurski2017}.
If $G-X$ is empty, construct $G[X]$ directly with $t$ labels.

Since connected cographs have diameter at most two, the deletion-distance
bound of Kobayashi and Otachi~\cite{KobayashiOtachi2022} gives
$\diam(G)\leq D:=4t+3$. Thresholds $R\geq D$ are answered by a
singleton path. For $R<D$, use a logical sentence that guesses
$p_0,\ldots,p_j$ for some $0\leq j\leq D$, requires consecutive
vertices to be adjacent, forbids a walk of length below $j$ between
$p_0$ and $p_j$, and requires every vertex to lie within distance $R$
of some $p_i$. These distance tests are first-order formulas obtained
by quantifying walks of bounded length. The adjacency and endpoint
conditions force the guessed walk to be a shortest path. Both the
expression width and the sentence length depend only on $t$, so
Courcelle--Makowsky--Rotics model checking~\cite{CMR2000} yields
$f(k)n^{O(1)}$ time.

Testing $R=0,\ldots,D$ finds the optimum. Fixing an accepted path
length and then its vertices one at a time, using the model checker
to test each candidate, recovers a witnessing path in FPT time.
The full argument appears in the \emph{Distance to cographs}
subsection of the full paper.
\end{proof}

\section{Kernelization}
\label{sv:sec:kernelization}

We contrast a polynomial kernel for maximum leaf number with four
kernel lower bounds. Complete proofs and the lower-bound constructions
appear in the \emph{Kernelization} section of the full paper.

\subsection{A polynomial kernel for maximum leaf number}

FPT for this parameter is known~\cite[Theorem~27]{KuceraSuchy2023}.
Our kernel uses a small weighted representation and verifies its
solutions without expanding the binary-encoded edge lengths.

\begin{theorem}\label{sv:thm:mainleaf}
\ESP{} parameterized by maximum leaf number $\ell$ admits an
$O(\ell^4)$-bit compression into a fixed language in $\NP$, and hence
a polynomial kernel consisting of an ordinary unweighted \ESP{} instance.
\end{theorem}

\begin{proof}
We may assume $n\ge2$ and $R<n$, since the other cases are immediate
yes-instances. Suppress maximal degree-two chains, retaining one vertex
when $G$ is a cycle. The resulting connected multigraph $H$ may have
loops and parallel edges. Give every edge $e$ its original chain length
$L_e\le n$. By Bouland's subdivision bound, as restated in
\cite[Lemma~28]{KuceraSuchy2023}, and the elementary bound
$\Delta(G)\le\ell$, we obtain
\[
 p:=|V(H)|<4\ell,\qquad q:=|E(H)|<2\ell^2.
\]
Indeed, $2q\le p\ell$. The tuple $(H,(L_e),R)$ reconstructs $G$
exactly and uses $O(q\log n)$ bits.

There is a $2^q n^{O(1)}$-time algorithm: guess both endpoints in $G$,
retain them by splitting their chains, and enumerate subsets of the at
most $q+2$ resulting edges. Keep only subsets forming a simple path
between these endpoints. Every simple path in $G$ arises this way.
Test shortestness by comparing its total length with the endpoint
distance, and test eccentricity by breadth-first search in $G$.
Also test one-vertex paths. If $q\le\log_2 n$, this algorithm runs in
polynomial time, so output a fixed yes- or no-instance. Otherwise,
output the weighted representation, whose length is
$O(q\log n)=O(q^2)=O(\ell^4)$. This algorithm need not compute $\ell$.

It remains to show that the target language belongs to $\NP$.
Let \SubESP{} contain precisely the binary-encoded weighted
multigraphs and thresholds whose unweighted subdivisions are simple
and admit a suitable shortest path. Simplicity is checked by rejecting
loops of length below three and parallel length-one edges.

A certificate specifies each endpoint by a retained vertex or an edge
and a binary offset. Split the relevant edges at these endpoints,
obtaining $H'$, and specify a simple path $Q$ in $H'$.
These operations add at most two vertices and two edges. Compute
weighted distances $D$ in $H'$. The expansion $P$ of $Q$ is shortest
exactly when its total weight equals $D(s,t)$. Put
$\delta(v)=\min_{z\in V(Q)}D(v,z)$. For an unselected edge $e=uv$ of
length $L$, the vertex at offset $i$ has distance
$\min\{\delta(u)+i,\delta(v)+L-i\}$ from $P$. Since
$|\delta(u)-\delta(v)|\le L$, its maximum over integer offsets is
\[
 \max_{0\le i\le L}
 \min\{\delta(u)+i,\delta(v)+L-i\}
 =\left\lfloor\frac{L+\delta(u)+\delta(v)}2\right\rfloor.
\]
The formula also holds for loops. Selected edges lie on $P$.
Consequently, checking $\delta(v)\le R$ for retained vertices and
$L_e+\delta(u)+\delta(v)\le2R+1$ for unselected edges verifies the
eccentricity bound using polynomially many operations on binary
integers. Thus \SubESP{} belongs to $\NP$.

Finally, ordinary \ESP{} is NP-complete on connected unweighted simple
graphs~\cite[Theorem~7]{DraganLeitert2017}. A fixed polynomial-time
many-one reduction from \SubESP{} to \ESP{} converts the compression
into an equivalent instance of encoding length $\ell^{O(1)}$.
Its maximum leaf number is bounded by its number of vertices, yielding
the required kernel.
\end{proof}

\subsection{Kernel lower bounds}

\begin{theorem}\label{sv:thm:kernel-lower-bounds}
Unless $\NP\subseteq\mathrm{coNP}/\mathrm{poly}$, \ESP{} has no
polynomial kernel parameterized by distance to clique, feedback edge
number, vertex cover number, or bandwidth. The same holds for deletion
distance to cluster, split, proper interval, or interval graphs.
\end{theorem}

For distance to clique, an OR-cross-composition combines CNF formulas
using $O(\log t)$ selector variables. A layered assignment path must
dominate a clique of clause vertices; deleting a set of size polynomial
in the number of variables plus $\log t$ leaves that clique.
For feedback edge number, the input index is encoded by positions on
long subdivided test paths. Attaching pendant tests records the source
clauses, while subdivision and pendant attachments preserve the number
of independent cycles. The two endpoint extensions select an input
whose clauses must all be satisfied.

For bandwidth, two rails connect a sequence of formula modules. A
shortest path switches from the upper to the lower rail in exactly one
module, selecting its assignment. Clause tests constrain that module,
while all others are automatically dominated. Ordering each module
consecutively bounds bandwidth by a polynomial in the largest
source-instance size.
These three constructions invoke the cross-composition framework
of Bodlaender, Jansen, and Kratsch~\cite{BJK14}.

\paragraph*{The vertex-cover lower bound.}
We give a polynomial parameter transformation from
\textsc{Small Universe Set Cover}, whose parameter is $|U|+b$ and
which has no polynomial kernel under the stated assumption~\cite{DLS14}.
The input is a family $\mathcal F\subseteq2^U$ and a budget $b$;
it asks whether at most $b$ sets cover $U$.
Handle trivial instances directly, and assume $U\ne\varnothing$,
$\bigcup_{F\in\mathcal F}F=U$, and $1\leq b\leq|U|$.
Set $R=b+1$.

Create connectors $z_0,\ldots,z_b$ and independent candidate layers
$I_i=\{c_{i,F}:F\in\mathcal F\}$, joining each candidate in $I_i$
to $z_{i-1}$ and $z_i$. Attach a pendant path of length $R$ at
each of $z_0,z_b$. For every $u\in U$, create a terminal $x_u$ and
hubs $h_{i,u}$, $i\in[b]$. Join each hub to $x_u$ by a length-$b$
path; these paths share only $x_u$, and different element gadgets
are disjoint. Join $c_{i,F}$ to $h_{i,u}$ exactly when $u\in F$.
All paths are expanded explicitly. The construction is polynomial-time,
and deleting all vertices outside the candidate layers leaves an
independent set. These deleted vertices form a vertex cover $X$ of size
\[
 |X|=|U|b^2+|U|+3b+3=O((|U|+b)^3),
\]
independent of $|\mathcal F|$.

A \emph{selection path} chooses one candidate per layer between
$z_0$ and $z_b$, and has length $2b$. These are exactly the shortest
$z_0$--$z_b$ paths: traversing an element gadget between distinct layers
costs at least $2b+2$. After excluding such crossings, every intermediate
connector is unavoidable, so the $2b$ bound forces exactly one candidate
in each layer.
To cover both pendant endpoints, any feasible geodesic must contain
$z_0$ and $z_b$. Its intervening subpath is therefore a selection path.
An extension from $z_0$ into $I_1$ would admit a shortcut to $z_1$,
and symmetrically at $z_b$; hence extensions lie only in the pendant
tails. Trimming them preserves coverage: their attachment vertices
remain, and every tail vertex is within $R$ of its attachment.

Suppose the selection path $P$ chooses $F_1,\ldots,F_b$. Every route
from $x_u$ to $P$ first traverses a length-$b$ spoke. Its hub is
adjacent to $P$ precisely when its selected set contains $u$;
otherwise the hub reaches a connector in two edges. Thus
\[
 \dist_G(x_u,V(P))=
 \begin{cases}
 R, & u\in\bigcup_{i=1}^b F_i,\\
 R+1, & u\notin\bigcup_{i=1}^b F_i.
 \end{cases}
\]
Every nonterminal is within $R$ of $P$: a nonterminal spoke vertex
reaches a connector in at most $(b-1)+2$ edges, and the backbone and
tail vertices satisfy the bound directly. A cover of size at most $b$,
padded by repetitions, therefore gives a feasible selection path.
Conversely, canonicalizing any solution and applying the displayed
identity recovers such a cover.

The transformation theorem of Bodlaender, Thomass\'e, and
Yeo~\cite{BTY11} transfers the kernel lower bound.
Since $G-X$ is edgeless, the same construction bounds the deletion
distances to cluster, split, proper interval, and interval graphs by
$|X|$, proving their lower bounds as well.

\section{Further results}
\label{sv:sec:further-results}

\subsection{Consequences of earlier algorithms}
\label{sv:sec:known-consequences}
The following consequences are included to complete the parameter
landscape; their underlying algorithms are due to
Ku\v{c}era and Such\'y~\cite[Corollary~26]{KuceraSuchy2023}
and Dragan and Leitert~\cite[Corollary~12]{DraganLeitert2017}.

\begin{proposition}\label{sv:prop:known-consequences}
For a computable function $F$, \ESP{} and its optimization version
can be solved in $F(\td(G))n^3$ time, without a supplied treedepth
decomposition. Both versions are also in XP parameterized by diameter.
\end{proposition}

\begin{proof}
Put $d=\td(G)$ and $\Delta=\diam(G)$. The standard bounds
$\tw(G)\leq d-1$ and $\Delta\leq2^d-2$ follow, respectively,
from ancestor bags and the maximum path length in the closure of a
height-$d$ forest. Ku\v{c}era and Such\'y's algorithm takes
$A(\tw(G),R)n^3$ time and returns a witness. After answering
$R\geq\Delta$ with a singleton path, its two parameters are bounded
by functions of $d$. Testing thresholds up to $\Delta$ also gives
optimization within $F(d)n^3$ time.

For diameter, the optimum eccentricity $r$ satisfies $r\leq\Delta$.
Dragan and Leitert's optimization algorithm takes
$O(n^{2r+2}m)$ time, even without knowing $r$; thus it takes
$O(n^{2\Delta+2}m)$ time for $n\geq2$. Handle the singleton graph
directly. The \emph{Treedepth} and \emph{Diameter} subsections of the
full paper give the detailed derivations.
\end{proof}

\subsection{Vertex clique-cover number}
\label{sv:sec:clique-cover}

\begin{theorem}\label{sv:thm:clique-cover-hardness}
\ESP{} is $\Wtwo$-hard parameterized by the vertex clique-cover number,
even when a vertex clique cover is supplied and its size is the parameter.
\end{theorem}

\begin{proof}[Proof sketch]
We reduce from \ds{}, parameterized by the solution size
$k$~\cite{DowneyFellows1995}. The construction adapts the choice layers
and distance tests of Dragan and
Leitert~\cite[Theorem~13]{DraganLeitert2017}; clique layers and columns
keep the clique-cover number bounded in terms of $k$.
Let $(F,k)$ be the source instance, with $1\leq k\leq |V(F)|$, and set
$R=k$.

For each $i\in[k]$ and $h\in\{0,\ldots,R-1\}$, create a clique
$C_{i,h}=\{x_{i,h}:x\in V(F)\}$. Join consecutive copies of each vertex
to form columns $x_{i,0}-\cdots-x_{i,R-1}$.
Completely join consecutive bottom cliques $C_{i,0}$ and $C_{i+1,0}$.
Add $s$ complete to $C_{1,0}$ and $t$ complete to $C_{k,0}$;
these vertices and the bottom cliques form the \emph{backbone}.
Attach pendant paths of length $R$ at $s$ and $t$.
Finally, create a witness clique $W=\{q_x:x\in V(F)\}$ and join
$q_x$ to $y_{i,R-1}$ exactly when $y\in N_F[x]$.
There are no further edges. The resulting connected graph $H$ has a
supplied clique cover consisting of the $kR$ layer cliques, $W$, and
the $2R+2$ individual vertices of the two pendant paths, including
$s,t$. Thus its size is $k^2+2k+3$.

For a tuple $\boldsymbol v=(v_1,\ldots,v_k)\in V(F)^k$, allowing
repetitions, put
\[
 P_{\boldsymbol v}=(s,(v_1)_{1,0},\ldots,(v_k)_{k,0},t),
 \qquad D_{\boldsymbol v}=\{v_1,\ldots,v_k\}.
\]
The shortest $s$--$t$ paths are exactly the paths $P_{\boldsymbol v}$.
Indeed, an excursion outside the backbone returning to the same bottom
clique can be replaced by at most one edge. An excursion between distinct
bottom cliques must pass through $W$ and uses at least $2R$ edges,
whereas the backbone connects its endpoints in at most $k-1$ edges.
Within the backbone, a shortest path selects exactly one vertex per
bottom clique.

Every vertex outside $W$ is within distance $R$ of
$P_{\boldsymbol v}$: descend a column and use at most one bottom-clique
edge, or follow a pendant path to its attachment.
Give column vertices $y_{i,h}$ level $h$ and witnesses level $R$.
Reaching the backbone from $q_x$ in at most $R$ steps requires decreasing
the level at every step, hence following one column whose label belongs
to $N_F[x]$. Consequently,
\[
 \dist_H(q_x,V(P_{\boldsymbol v}))\leq R
 \quad\Longleftrightarrow\quad
 N_F[x]\cap D_{\boldsymbol v}\neq\varnothing.
\]
Thus a dominating set of size at most $k$, padded by repetitions, gives a
feasible path.

Conversely, the radius-$R$ ball around either pendant endpoint consists
exactly of its pendant path. A feasible geodesic $Q$ therefore meets
both pendant paths and contains $s,t$; its $s$--$t$ subpath is some
$P_{\boldsymbol v}$. Since a geodesic is induced, an extension beyond
$s$ into $C_{1,0}$ would create a chord, and similarly at $t$.
All extensions therefore lie on the pendant paths. Removing them
preserves distances to vertices outside those paths, while every pendant
vertex remains within distance $R$ of $s$ or $t$.
Hence $P_{\boldsymbol v}$ is feasible, and the displayed equivalence
shows that $D_{\boldsymbol v}$ dominates $F$.
The construction is polynomial and the supplied cover has size depending
only on $k$. Complete distance checks appear in the \emph{Clique cover
number} subsection of the full paper supplied with the submission.
\end{proof}

\subsection{Tree-independence number}
\label{sv:sec:tree-independence}

The following reduction adapts the SAT construction of Dragan and
Leitert~\cite[Theorem~7]{DraganLeitert2017}. The added literal-pair
edges and clause clique give the required decomposition certificate.

\begin{theorem}\label{sv:thm:tree-independence-hardness}
\ESP{} is NP-hard on connected graphs $G$ with $\ta(G)\leq2$, even
when a tree decomposition and a partition of every bag into two cliques
are supplied. Consequently, parameterization by tree-independence
number is para-NP-hard.
\end{theorem}

\begin{proof}[Proof sketch]
Let $\Phi$ be a $3$-CNF formula with variables $x_1,\ldots,x_n$ and
nonempty clauses $D_1,\ldots,D_m$, where $n,m\geq1$, and put $R=n+1$.
Build a variable chain $H$ with connectors $a_0,\ldots,a_n$ and
literal pairs $p_i,q_i$ representing $x_i,\neg x_i$: both are adjacent to $a_{i-1}$ and
$a_i$, and $p_iq_i$ is an edge. Make the clause vertices
$C=\{c_1,\ldots,c_m\}$ a clique. For each literal occurrence in $D_j$,
join $c_j$ to its literal vertex by a length-$R$ path, using distinct
internal vertices for all these paths. Attach a pendant length-$R$
tail to each of $a_0,a_n$.

The graph $G-C$ has a tree decomposition with clique bags: arrange the
triangles $\{a_{i-1},p_i,q_i\}$ and $\{a_i,p_i,q_i\}$ in chain order,
and attach chains of edge bags for the pendant paths.
Adding $C$ to every bag gives a decomposition of $G$ in which every
bag is explicitly partitioned into two cliques. Hence $\ta(G)\leq2$.

The shortest $a_0$--$a_n$ paths have length $2n$ and are precisely
\[
  P=a_0,z_1,a_1,z_2,\ldots,z_n,a_n,
  \qquad z_i\in\{p_i,q_i\}.
\]
Indeed, leaving $H$ and returning through occurrence paths costs at
least $2R>2n$; inside $H$, each variable gadget requires two edges.
A satisfying assignment chooses its true literal at each gadget.
The resulting $P$ has eccentricity at most $R$: the unchosen literals
are adjacent to $P$, internal occurrence vertices are within $R-1$
of a literal, tail vertices are within $R$ of their attachment, and
each clause has an occurrence path to a chosen literal.

Conversely, a shortest path $Q$ of eccentricity at most $R$ must meet
both tails, including their attachments: a tail endpoint is farther
than $R$ from every vertex outside its tail. Thus $Q$ contains a core
path $P$ of the displayed form. Every extension of $P$ stays on a
tail, because extending through an unchosen literal would create a
chord via its literal-pair edge. Consequently,
$\dist_G(c_j,V(Q))=\dist_G(c_j,V(P))$.
A path from $c_j$ to $P$ traverses a complete occurrence path of
length $R$; if its total length is at most $R$, it must consist
entirely of such a path from $c_j$ to a chosen literal.
The assignment encoded by $P$ therefore satisfies every clause.
The full paper supplies the detailed argument in its
\emph{Tree-independence number} subsection.
\end{proof}

}{%
\tableofcontents
\thispagestyle{empty}
\newpage
\mespFullIntroduction
\mespFullProofs
}

\paragraph*{Declaration of AI use.}
All results obtained in this paper were developed using OpenAI's
GPT-5.6/6 Pro. These models were also used to generate, refine, and
write the proofs and exposition. The authors take full responsibility
for the correctness, originality, and attribution of the results and
for the final content of the paper.

\bibliography{main}


\end{document}